\documentclass[12pt]{article}
\usepackage[english]{babel}

\usepackage{collcell}
\usepackage{mathtools}
\usepackage{diagbox}
\RequirePackage{amsthm,amsmath,amsfonts,amssymb}
\mathtoolsset{showonlyrefs}

\usepackage{xcolor}
\usepackage{graphicx}
\usepackage{listings}
\usepackage{algorithm}
\usepackage{algpseudocode}

\usepackage{bm}
\usepackage{url}

\usepackage{booktabs}
\usepackage{caption}
\usepackage{array}
\usepackage{enumitem}

\usepackage{algorithm}
\usepackage{algpseudocode}

\RequirePackage[authoryear]{natbib}

\RequirePackage{bm} 
\RequirePackage{subcaption} 

\DeclareMathOperator*{\argmin}{arg\,min}
\DeclareMathOperator{\Log}{Log}
\DeclareMathOperator{\Exp}{Exp}

\usepackage{scalerel}
\newcommand{\E}{{\mathbb E}}
\newcommand{\R}{\mathbb{R}}
\newcommand{\bR}{\mathbf{R}}
\newcommand{\bSigma}{\mathbf{\Sigma}}
\newcommand{\bM}{\mathbf{M}}

\newcommand{\x}{\mathbf{x}}

\newcommand{\iw}{\mathcal{IW}}
\newcommand{\siw}{\mathcal{SIW}}
\renewcommand{\P}{\mathrm{P}}

\newcommand{\ind}{\perp\!\!\!\!\perp}
\newcommand{\iidsim}{\stackrel{iid}{\sim}}

\newcommand{\cov}{covariance }
\renewcommand{\arraystretch}{1.4}

 \usepackage[table]{xcolor}

\definecolor{lightgray}{rgb}{0.96,0.96,0.96}
\definecolor{lightergray}{rgb}{0.98,0.98,0.98}

\usepackage{fancyhdr}
\usepackage{graphicx}
\usepackage{color}

\usepackage{thmtools}
\usepackage{thm-restate}
\newtheorem{theorem}{Theorem}
\newtheorem{prop}[theorem]{Proposition}
\newtheorem{definition}{Definition}[section]
\newtheorem{lemma}[theorem]{Lemma}
\newtheorem{property}[theorem]{Property}
\newtheorem{corollary}[theorem]{Corollary}
\newtheorem{remark}{Remark}
\usepackage[colorlinks=true, allcolors=blue]{hyperref}

\usepackage[left=2.5cm,right=2.5cm,top=2.5cm,bottom=2.5cm]{geometry}
\begin{document}
\renewcommand{\headrulewidth}{0pt}

\fancyhead[R]{ }


\begin{center}
{\Large
	{\sc  Prior elicitation for Bayesian estimation of single-subject connectivity networks}
}
\bigskip

 Yiye Jiang$^{1}$, Alice Chevaux$^{1}$, Wendy Meiring$^{3}$, Alex Petersen$^{4}$, Guillaume Kon Kam King$^{2}$, Julyan Arbel$^{1}$, Sophie Achard$^{1}$
\bigskip

{\it
$^{1}$ Univ. Grenoble Alpes, Inria, CNRS, Grenoble INP, LJK, 38000 Grenoble, France, \\
yiye.jiang@inria.fr, alice.chevaux@inria.fr, 
julyan.arbel@inria.fr, sophie.achard@cnrs.fr \\
$^{2}$ Université Paris-Saclay, INRAE, MaIAGE, 78350, Jouy-en-Josas, France, guillaume.konkamking@inrae.fr \\
$^{3}$ University of California, Santa Barbara, Santa Barbara, 93106, CA, USA, meiring@ucsb.edu\\
$^{4}$ Brigham Young University, Provo, 84602, UT, USA, petersen@stat.byu.edu

}
\end{center}
\bigskip




{\bf Abstract.} 
Inference of brain functional connectivity networks from resting-state fMRI data is a key focus in neuroimaging. This paper introduces new Bayesian approaches for inferring a functional connectivity graph from multivariate resting-state fMRI time series of a single subject. Our methods rely on novel Bayesian priors on correlation matrices and a dedicated prior elicitation framework, which translates prior beliefs about the expected level and variability of correlations into interpretable hyperparameter choices, enabling the construction of expert-informed priors. When combined with a Gaussian likelihood, these priors also exhibit computational advantages. Compared to most existing methods for this problem that estimate constant weights, our model provides distributional weights defined by the posterior distributions for the connectivity graph, yielding more robust point estimates through the regularizing effect of expert-informed priors, evaluating uncertainty, and enabling a range of post-inference analyses. In particular, we derive a procedure for identifying significant connectivities based on posterior distributions of weights and credible sets. To the best of our knowledge, only one existing Bayesian functional connectivity model is applicable to single-subject resting-state fMRI data, making our approach a valuable addition to the field and demonstrating superior performance in our experiments.


\bigskip\bigskip


\newpage

\tableofcontents

\newpage

\section{Introduction}
Functional connectivity  seeks to quantify the extent to which different parts of the brain show similar temporal characteristics. It is often measured by the level of temporal dependence during functional magnetic resonance imaging (fMRI). Statistical methods play a crucial role in learning and analyzing functional connectivity from fMRI data \citep{lazar2008statistical}. Classical methods include calculating temporal correlations between regions of interest, clustering, and Gaussian graphical lasso \citep{varoquaux2010brain,belilovsky2016testing}.
Among others, Bayesian approaches have emerged as an influential framework for functional connectivity inference since one of their earliest works \citep{patelBayesianApproachDetermining2006a}. Advantages of the Bayesian framework include  a built-in measure of uncertainty provided by the posterior distributions, 
which may be used to derive more robust decision rules, and the flexibility to develop realistic models. 

Although many works have shown the success of Bayesian methods in inferring functional connectivity networks, most of them are proposed for task-based fMRI \citep{bowmanBayesianHierarchicalFramework2008,warnickBayesianApproachEstimating2018,bergerBayesianAnalysisCovariance2020,renBayesianVaryingeffectsVector2024,kunduFlexibleBayesianProduct2024}.In such settings, methodological developments typically concentrate on modeling task-related fixed effects, while the functional connectivity structure is often treated as a secondary component. By contrast, the present work carefully develops a Bayesian framework for the connectivity component for resting-state fMRI. Moreover, it has the potential to be incorporated into existing task-based models by refining their treatment of connectivity, including models that consider dynamic connectivity, which often represent temporal variation through a finite set of discrete static connectivity states. In the following, we review existing Bayesian models for inferring functional connectivity structure from resting-state fMRI.


Two main types of observations are considered in Bayesian inference of functional connectivity: iid fMRI samples or Fisher’s z-transformed pairwise correlations. 
Based on iid fMRI samples, early work \citep{zhangInferringFunctionalInteraction2014} considers the inference of chain and V structures in static functional connectivity connectivity among regions for a single subject. The model supposes that all regions can be partitioned into three groups: A, B, and
C. Among A, B, and C, either a chain or a V is present. The output of inference is the posterior probabilities of each region being in the three groups, together with the probability of the whole brain forming a chain or a V. \cite{kunduScalableBayesianMatrix2021} considers the inference of an overall static functional connectivity graph among regions for a group of subjects. In particular, the fMRI signals are only considered iid across subjects, but not in time. To handle spatio-temporal dependence, a matrix-variate Gaussian distribution is adopted as the likelihood function, with the overall precision matrix decomposed into an functional connectivity and a temporal component via a Kronecker product. A Bayesian graphical lasso prior \citep{wangBayesianGraphicalLasso2012} and an Inverse-Wishart prior are then assigned respectively to the functional connectivity and temporal components.  \cite{kampmanTimevaryingFunctionalConnectivity2024,mengDynamicCovarianceEstimation2023} both consider the inference of a dynamic weighted functional connectivity graph among regions for a single subject based on a Gaussian likelihood and a Wishart process-type prior.

Among Bayesian approaches based on transformed correlations, \cite{chenBayesianHierarchicalFramework2016} and \cite{gorbachHierarchicalBayesianMixture2020} both consider static functional connectivity analysis at the group level while incorporating subject-specific clinical covariates, such as age, sex, and psychiatric status. In these models, latent binary variables are introduced to indicate the presence or absence of a connection, either between pairs of voxels \citep{chenBayesianHierarchicalFramework2016} or between pairs of regions \citep{gorbachHierarchicalBayesianMixture2020}. Each transformed correlation is then modeled through a mixture distribution whose components correspond to connected and non-connected pairs.
In particular, \cite{gorbachHierarchicalBayesianMixture2020} specifies a Lognormal distribution for the component associated with the presence of a connection, in order to reflect the prior belief that only positive correlations correspond to reliable connections:
\begin{equation}\label{eq:gorbach}
\begin{aligned}
    &\rho_{i,kk'} \mid w_{i,kk'} \iidsim w_{i,kk'}\text{lognormal}(\alpha^\top\mbox{x}_{i} + a_i, \sigma_{1i}^2) + (1-w_{i,kk'})\mathcal{N}(\mu_{0i}, \sigma_{0i}^2),    \\
    &w_{i,kk'} \mid \mbox{x}_{i} \iidsim \mathcal{B}(\Phi(\delta^\top\mbox{x}_{i} + d_i)), \;  \mu_{0i} \iidsim N(0,c_0^2), \;   \sigma^2_{0i}\iidsim IG(1.5,10^{-3}),\\
     &c_0^2\iidsim IG(1.5,10^{-3}).
\end{aligned}
\end{equation}
where $\rho_{i,kk'}$ denotes the Fisher $z$-transformed correlation between ROIs $k$ and $k'$ for subject $i$, $\mbox{x}_{i}$ is the vector of covariates, and $\Phi$ denotes the Gaussian cumulative distribution function.
Marginally, the connection indicator follows a Bernoulli distribution whose success probability depends on the covariates through a generalized linear model.
A key strength of these approaches is that they allow region pairs with no true connection to exhibit nonzero observed correlations, while still distinguishing them from correlations arising from genuinely connected pairs. This feature is arguably more realistic than assuming zero correlation for unconnected pairs. However, these methods treat correlation estimation as a preliminary step rather than as part of the inferential model itself. The empirical correlations are treated as observed data, without incorporating the uncertainty from preprocessing. In addition, the transformed correlations are modeled separately, thereby ignoring the global dependence structure they must satisfy as entries of a correlation matrix. These limitations are substantial.  Neglecting the positive-definiteness constraint results in an estimated network that does not correspond to a valid joint multivariate distribution, which undermines the mathematical coherence of the generative model and restricts its utility for subsequent multivariate network analyses.

Lastly, in addition to the two types of observations, \cite{acharyyaBayesianHierarchicalModeling2023} considers a collection of iid covariance matrices of a subject to infer a static functional connectivity graph among regions. The covariance matrices are calculated using a sliding-window method. A Wishart distribution is assigned as the likelihood function, where a shrinkage prior is adopted for the mean of the Wishart so as to encourage a low effective rank structure. 

We consider modeling directly fMRI samples, because processing will generally cause a loss of information. The modeling is on the region level. We rely on a Gaussian likelihood function as it is a common choice in literature, and measure the inter-regional connectivity by the correlation matrix. In particular, we wish to take into account possible expert knowledge on the values of correlation by controlling the means and variances of marginal prior distributions of correlation entries. In the related work, only \cite{chenBayesianHierarchicalFramework2016,gorbachHierarchicalBayesianMixture2020} consider similar controls of correlation entries. They proposed assigning the marginal distribution independently to each correlation entry, as in Equation \eqref{eq:gorbach}. All correlations are assumed to be iid, which facilitates the designation of the desired marginal prior to each $\rho_{i,kk'}$. However, as mentioned above, this assumption is philosophically and practically problematic, as it cannot yield a bona fide positive-definite correlation matrix. Thus, when we propose a marginal prior to a correlation entry, we should take into account the impact on the values of other correlation entries. This raises a significant technical challenge to the proposed model.

In this work, we derive two new priors supported on correlation matrices; however, they still allow for direct control of the marginal prior distribution of each correlation entry. These priors are then used to construct Bayesian models for inferring static functional connectivity networks from resting-state fMRI data at the single-subject level. This leads to our first contribution, which is of a methodological nature. The second contribution is of an applicative nature. To the best of our knowledge, only a very limited number of Bayesian approaches have been proposed for single-subject functional connectivity graph learning in the resting-state setting \citep{acharyyaBayesianHierarchicalModeling2023,zhangInferringFunctionalInteraction2014}. Moreover, the approach of \cite{zhangInferringFunctionalInteraction2014} does not appear to be supported by publicly available code, which limits its practical reproducibility. We therefore compared our method with the approach of \cite{acharyyaBayesianHierarchicalModeling2023} on real data. The results show that the proposed method yields improved performance, thereby providing a useful addition to the existing literature.


\paragraph{Organization.} In Section \ref{sec: iw}, we derive a new prior on correlation matrices based on the Inverse Wishart distribution. In Section \ref{sec: siw0}, we generalize the prior via the Shrinkage Inverse Wishart distribution. In Section \ref{sec: app}, we test the derived models with real data, where we show an application of posterior distributions-edge detection. We also compare with \cite{acharyyaBayesianHierarchicalModeling2023} in the empirical studies section\footnote{For the sake of reproducible research, R code is available at
\url{https://github.com/yiyej/Bayesian-single-subject-functional connectivity-models}, which implements the posterior sampling procedure developed in Section \ref{sec: sampling} and the empirical studies carried out in this paper.}. 

\paragraph{Notations.} To facilitate reading, we use bold font for random values, plain font for constants. 
$A \succ 0$ denotes that $A$ is a positive definite matrix, and $A_{kk'}$ denotes the $(kk')$-entry of matrix $A$. To index the $(kk')$-entry of a matrix already indexed, such as $P_0$ and $P_{0,n}$, we add $kk'$ at the end, namely, $P_{0,kk'}$ and $P_{0,n,kk'}$. $R(\Sigma)$ denotes the induced correlation matrix of covariance matrix $\Sigma$.

\section{Elicitation for simple prior}\label{sec: iw}

\subsection{Connectivity model}
Let us define wavelet-decomposed fMRI measurements at time $i$ and region $k$ by random variables describing wavelet coefficients $\x_{ki}, i = 1, \ldots, n, k = 1, \ldots, K$. We assume a Gaussian likelihood for $\bm x_{ki}$ as it is a common choice in the literature:
\vspace{-0.1in}
\begin{equation}\label{eq:model}
\x_i|\bSigma\stackrel{iid}{\sim} \mathcal{N}\left(0,\bSigma\right),
\end{equation}
where $\x_i = \{\bm x_{1i}, \ldots, \bm x_{Ki}\}^\top$.
The iid assumption is convenient, and proves to be approximately valid by performing basic wavelet checks using existing packages such as $\mathtt{waveslim}$ in R \citep{achard2016multivariate}. In addition, there is no loss of generality in assuming the mean equal to zero. 

For the functional connectivity metric, we consider the induced correlation matrix $\bR = R(\bSigma)$ defined by 
\[\bR_{kk'} = \frac{\bSigma_{kk'}}{\sqrt{\bSigma_{kk}}\sqrt{\bSigma_{k'k'}}}.\] When taking $\bR$ as the adjacency matrix of graph, a fully-connected weighted undirected graph can be retrieved to represent the overall (static) functional connectivity levels among the studied regional pairs $(k,k')$, where $\forall k, k' = 1, \ldots, K$ during the course of fMRI scan $i = 1, \ldots, n$. The posterior distribution of each edge weight $\bR_{kk'}$ provides a deeper insight into the functional connectivity of pair $(k,k')$ than a simple point estimate, especially on the uncertainty. 

Even though the variable of interest is the correlation $\bR$, we first work with priors on the covariance $\bSigma$. This is because such priors on $\bSigma$ can be conjugate to the normal likelihood, thus leading to fast inference. We then study the properties of the prior induced on $\bR$. In this context, we need a prior on $\bSigma$, parametrized in such a way that the parameters 
can control directly the prior induced on $\bR$ and its marginals of $\bR_{kk'}$. 
In the following, we use the Inverse-Wishart ($\iw$) distribution to construct the desired prior. The $\iw$ is a classical distribution on positive-definite matrices, widely used as the prior for the \cov matrix of a Gaussian likelihood due to the conjugacy. We recall in the following its definition.

\subsection{Inverse-wishart prior}

\paragraph{Inverse-Wishart distribution.}
A random matrix $\bSigma \succ 0$ of size $K\times K$ follows an \emph{Inverse-Wishart} distribution with degrees of freedom $\nu > K + 3$ and scale matrix $\Psi \succ 0$, if its density is
\begin{equation}\label{eq: iw den}
\pi(\bSigma | \Psi, \nu) \;=\;
\frac{|\Psi|^{\nu/2}}{2^{\nu K/2}\,\Gamma_K(\nu/2)}
\,|\bSigma|^{-(\nu+K+1)/2}\,
\exp\!\left(-\frac{1}{2}\operatorname{tr}\!\big(\Psi\,\bSigma^{-1}\big)\right),    
\end{equation}
where $\Gamma_K(\cdot)$ is the multivariate gamma function. We denote this by $\bSigma \sim \mathcal{IW}(\Psi, \nu)$. 

To be well-defined, this family requires only that $\nu > K - 1$. Our additional constraint, $\nu > K + 3$, ensures the distribution has finite variance, which we will use in the following development.

\paragraph{Reparametrization.} An equivalent parametrization of $\iw$ is in terms of $(\P, \Vec{\sigma},\nu)$. Here, $\P = R(\Psi)$, so that $
\P_{kk'} = \Psi_{kk'}/(\sqrt{\Psi_{kk}}\sqrt{\Psi_{k'k'}})$ and $\Vec{\sigma} = (\sigma_1, \sigma_2, \ldots, \sigma_K)$ with $\sigma_k = \Psi_{kk}$, $\forall k, k' = 1, \ldots, K$. Thus $\P$ takes the form of a correlation matrix, and $\sigma_k > 0$. In this paper, we will adopt the correlation-variance parameterization for $\iw$, denoted by $\iw(\P, \Vec{\sigma},\nu).$

We propose to use $\iw(\P, \Vec{\sigma},\nu)$ as the prior for $\bm\Sigma$, the covariance matrix in the Gaussian likelihood model \eqref{eq:model}. The parametrization in $(\P, \Vec{\sigma},\nu)$ leads into the direct control on the prior induced on $\bR$ as stated in Proposition  \ref{prop: IW prior}. 

\begin{prop}\label{prop: IW prior} 
Let $\bSigma \sim \mathcal{IW}(\P, \Vec{\sigma},\nu)$, and denote $\bR = R(\bSigma)$. Then we have for $k\neq k'$ and as $\nu \rightarrow \infty$:
 \begin{enumerate}
 \vspace{-0.05in}
 \item $\bR \rightarrow \P$ in probability,
    \item $ \mathbb{E}(\bR_{kk'}) =
    \P_{kk'} \left(1-\frac{1-(\P_{kk'})^2}{2(\nu - K+2)}\right) +  o\left(\frac{1}{\nu^2}\right),$ 
     \vspace{-0.05in}
    \item $\mathbb{V}
    (\bR_{kk'}) = \frac{(1-(\P_{kk'})^2)^2}{\nu - K+1} + O\left(\frac{1}{\nu^3}\right).$ 
\end{enumerate}   
\end{prop}
\begin{remark}\label{rem: mean var iw}
Proposition \ref{prop: IW prior} can be more generic with $K$ growing as well. In this case, the results are maintained given $\nu - K \rightarrow \infty$, and with $o\left(\frac{1}{\nu^2}\right), O\left(\frac{1}{\nu^3}\right)$ in Points 2 and 3 changed to $o\left(\frac{1}{(\nu-K)^2}\right), O\left(\frac{1}{(\nu-K)^3}\right)$. See Appendix \ref{app: proof mean var iw} for a proof. 
\end{remark}

\paragraph{Hyperparameter elicitation.} Proposition~\ref{prop: IW prior} shows that, at first order as $\nu\to\infty$, $\mathbb{E}(\bR_{kk'})$ and $\mathbb{V}(\bR_{kk'})$ are determined by $\P_{kk'}$ and $\nu$. This provides a practical rule for eliciting the hyperparameters $\P_{kk'}$ and $\nu$ from prior knowledge about the level and variability of the correlations. More specifically, one may specify the prior mean $\mathbb{E}(\bR_{kk'})$ and prior variance $\mathbb{V}(\bR_{kk'})$ using expert knowledge or information from a comparable dataset. For example, in a healthy-subject resting-state fMRI study, one may use empirical correlations computed from a reference cohort such as the Human Connectome Project. In a clinical application involving aging, mild cognitive impairment, or Alzheimer's disease, one may instead use a cohort from ADNI or another study with comparable acquisition and preprocessing pipelines. Then by inverting the formulas of $\mathbb{E}(\bR_{kk'})$ and $\mathbb{V}(\bR_{kk'})$, the values of the hyperparameters $\P$ and $\nu$, which encode the information a priori, are obtained. The resulting prior furthermore contributes to improving the precision of the posterior estimation. In practice, the following simple approximate formulas are sufficient for the elicitation.
   $$ \mathbb{E}(\bR_{kk'}) \approx
    \P_{kk'} , \quad \mathbb{V}
    (\bR_{kk'}) \approx \frac{(1-(\P_{kk'})^2)^2}{\nu - K+1}.$$
Considering $\Vec{\sigma} = (\sigma_1, \ldots, \sigma_K)$, this hyperparameter does not have an impact on the prior moments of the correlation. It will be shown in the next section (Theorem \ref{thm: IW post}) that it controls the relative weight of the prior correlation directly in the posterior correlation mean. 
We discuss a tuning strategy for $\Vec{\sigma}$ in the next section.

\paragraph{A special case of the $\iw$ prior: encoding a global correlation level.}

When prior information on the level of each pairwise correlation is not available, it may
be difficult to specify a fully structured prior matrix $\P$. In such situations, a
conservative strategy is to encode only a global prior belief on the overall correlation
level, rather than pair-specific information. We therefore consider a practical special
case of the general prior $\mathcal{IW}(\P,\Vec{\sigma},\nu)$, designed for hyperparameter
elicitation when only coarse prior information is available.

\begin{definition}\label{def: iw}
We define a prior distribution $\bSigma \sim \mathcal{IW}(\P, \Vec{\sigma}, \nu)$,
with $\nu > K + 3$, by setting
\begin{equation}\label{eq: prior}
\P = (1-\rho)I_K + \rho J_K,
\end{equation}
where $\rho \in \left(-\frac{1}{K-1},1\right)$, $I_K$ is the identity matrix of size $K$,
and $J_K = \bm{1} \bm{1}^\top$, with $\bm{1} = (1,\ldots,1) \in \R^K$.
\end{definition}

The domain of $\rho$ ensures that $\P \succ 0$, and hence that the corresponding
$\iw(\Psi,\nu)$ distribution is well defined. In this specification, the single parameter
$\rho$ encodes the prior belief on the overall level of correlations. The following
corollary gives the induced marginal prior mean and variance of the correlation entries,
and can therefore be used for hyperparameter elicitation.

\begin{corollary}\label{coro: IW prior}
The prior defined in Definition~\ref{def: iw} induces the following properties on the
prior correlation matrix. For $k \neq k'$, and as $\nu \rightarrow \infty$: 
\begin{enumerate}
    \vspace{-0.05in}
    \item 
    $\mathbb{E}(\bR_{kk'})
    =
    \rho \left(1-\frac{1-\rho^2}{2(\nu - K+2)}\right)
    + o\left(\frac{1}{\nu^2}\right).$
    \vspace{-0.05in}
    \item 
    $ \mathbb{V}(\bR_{kk'})
    =
    \frac{(1-\rho^2)^2}{\nu - K+1}
    + O\left(\frac{1}{\nu^3}\right).$
\end{enumerate}
\end{corollary}

Corollary~\ref{coro: IW prior} follows directly from Proposition~\ref{prop: IW prior}.
It provides a simple way to translate coarse prior information into the hyperparameters
$\rho$ and $\nu$. For instance, if one only expects significant functional correlations to
be mostly positive, without knowing which pairs should have larger prior correlations, one
can choose a common positive value of $\rho$. The parameter $\nu$ then controls the prior
variability around this common correlation level.

Moreover, because all
off-diagonal entries of $\P$ are identical, the induced prior on the correlation matrix is
exchangeable in the following sense.

\begin{prop}\label{prop:exch}
Let $\bSigma \sim \mathcal{IW}(\P, \Vec{\sigma},\nu)$, and denote $\bR = R(\bSigma)$. If the off-diagonal terms of $\P$ are all equal, that is
$\P_{kk'} \equiv \mathrm{constant}$ for all $k\neq k'$, then
$(\bR_{kk'})_{k < k'}$ is exchangeable. Equivalently,
\begin{equation}
Q\bR Q^{-1} \stackrel{d}{=} \bR,
\end{equation}
for any permutation matrix $Q$.
\end{prop}

See Appendix~\ref{app: proof mean var iw} for a proof. Proposition~\ref{prop:exch}
implies that, under this prior, no pair of variables is favored a priori over another. This is consistent with the goal of encoding only a global correlation level
through $\rho$, rather than a detailed pair-specific prior structure.

\subsection{Posterior moments}
Under the Gaussian likelihood $\x_i
|\bSigma
\stackrel{iid}{\sim} \mathcal{N}\left(
0
,
\bSigma\right), \, i = 1, \ldots, n$, and the $\iw$ prior $\bSigma \sim  \iw(\P, \Vec{\sigma}, \nu)$, denoting the induced correlation matrix by $\bR = R(\bSigma)$, we have the following approximations. 
\begin{theorem}\label{thm: IW post}
Denote $\nu_n = \nu + n$, for $k \neq k'$ and  as $\nu_n \rightarrow \infty$:
\begin{align*}
    \mathbb{E}(\bR_{kk'} \mid \x_{1:n}) &= \mathbf{P}_{n,kk'} \left(1-\frac{1-(\mathbf{P}_{n,kk'})^2}{2(\nu_{n} - (K-2))}\right) +  o\left(\frac{1}{\nu_{n}^2}\right), \\
    \mathbb{V}(\bR_{kk'} \mid \x_{1:n}) &= \frac{(1- (\mathbf{P}_{n,kk'})^2)^2}{\nu_n - K + 1} + O\left(\frac{1}{\nu_n^3}\right),  
\end{align*}
where $\x_{1:n}$ denotes $(\x_{1}, \x_{2}, \ldots, \x_{n})$, and 
\begin{equation}\label{eq: sigma_sample_var0}
\mathbf{P}_{n,kk'} = \frac{\sigma_k  \sigma_{k'}  \P_{kk'} + n\bm S_{kk'}  }{\bm\sigma_{n,k}\bm\sigma_{n,k'}}, \mbox{ with }   \bm\sigma_{n,k} = \sqrt{\sigma_{k}^2 + n \bm S_{kk}}, \mbox{ and }  \bm S_{kk'} = \frac{\sum_{i=1}^n\bm x_{ki}\bm x_{k'i}}{n}.
\end{equation}
\end{theorem}

The posterior results are also a direct application of Proposition \ref{prop: IW prior}, because $\iw(\P, \Vec{\sigma}, \nu)$ is a conjugate prior for the Gaussian likelihood, with the posterior defined by $\iw(\P_n, \Vec{\sigma}_{n}, \nu_n)$.
The posterior is a compromise between the prior information over correlation values specified via $\P$ and the data represented by $\bm S$. Compared to a Frequentist estimator such as the sample empirical correlation $\frac{\bm S_{kk'}}{\sqrt{\bm S_{kk}}\sqrt{\bm S_{k'k'}}}$, the posterior mean $\mathbb{E}(\bR_{kk'} \mid \x_{1:n})$ offers better point estimation thanks to prior knowledge. Additionally, the prior offers a regularization.  In high dimension, the sample correlation is singular with abnormal values of $\pm 1$. Aggregating with $\P$, a positive-definite matrix, yields a well-conditioned estimate, thereby avoiding extreme correlation values. More generally, aggregation makes the output estimation less sensitive to outliers. Thus, Bayesian estimation works both in high and low dimensions and is generally robust. By naturally shrinking the empirical estimates towards the structured prior, this Bayesian approach typically strictly dominates the standard frequentist sample correlation in terms of mean squared error, a well-known phenomenon in multivariate estimation\citep[see, e.g.,][]{robert2007bayesian}.

A simple and easy-to-interpret closed-form expression for the leading term in the posterior mean is available if we do not use expert knowledge to elicit $\Vec{\sigma}$, but instead use an Empirical Bayes-inspired approach, setting $\sigma_k$ to the sample standard deviation up to a constant $\sqrt{n_0}$. 
In this case, we have:
\begin{equation}\label{eq: sigma_sample_var}
 \mathbf{P}_{n,kk'} 
 = \frac{n_0}{n_0+n}\P_{kk'} + \frac{n}{n_0+n} \frac{\bm S_{kk'}}{\sqrt{\bm S_{kk}}\sqrt{\bm S_{k'k'}}}. 
\end{equation}
Note that this expression of $\mathbf{P}_{n,kk'}$ is a convex combination between the prior guess $\mathbf{P}_{kk'}$ weighted by $n_0$, and the sample correlation weighted by the sample size $n$. This dichotomy allows interpreting $n_0$ as a prior sample size, in the standard sense that conjugate hyperparameters in exponential families act as pseudo-counts/pseudo-sufficient statistics \citep{diaconis1979conjugate}.

When this strategy is used, we propose to define $n_0$ in terms of $\nu$, so that $\nu$ can have the same role of controlling the impact of $\P_{kk'}$ in the posterior mean of correlation as it did for the prior and posterior variances shown by Corollary \ref{coro: IW prior} and Theorem \ref{thm: IW post}. In this way, the model is easier to interpret. However, when one needs to control the mean and the variances separately, $n_0$ can be regarded as an additional hyperparameter. For the specific definition, we propose $n_0 = \nu - K - 1$, following the property: $\mathbb{E} (\bm \Sigma_{kk}) = \sigma_{k}^2/(\nu - K +1)$, giving the hyperparameter assignment $\sigma_k^2 = (\nu - K -1)\bm S_{kk}.$


\section{Connectivity model with $\iw / \siw_1$ mixture prior}\label{sec: siw0}
In this section, we extend the  $\iw$ framework by considering the Shrinkage Inverse-Wishart ($\siw$) distribution for covariance matrices introduced by \cite{berger2020Bayesian}. The $\siw$ family generalizes $\iw$ while preserving conjugacy with the Gaussian likelihood.

The motivation for this generalization can be understood from the spectral reparametrization of the $\iw$ density. Let $\Sigma = \Gamma \Lambda \Gamma^\top$ be the eigendecomposition of $\Sigma$, where $\Lambda = \mathrm{diag}(\lambda_1,\dots,\lambda_K)$, the $\iw$ density of Equation \eqref{eq: iw den} after changing variables in $(\Lambda,\Gamma)$ can be written as
\[
\pi(\Lambda, \Gamma \mid \Psi, \nu, b)
\propto
\prod_{i<j}(\lambda_i-\lambda_j)\,
\frac{
\exp\left\{-\frac12 \mathrm{tr}\!\left(\Psi \Gamma \Lambda^{-1}\Gamma^\top\right)\right\}
}{
|\Lambda|^{\nu}
}.
\]
The Vandermonde term $\prod_{i<j}(\lambda_i-\lambda_j)$, which arises from the Jacobian of the change of variables, implies little probability mass over matrices of concentrated eigenvalues.  In particular, if
\[
A_\delta=\{\Sigma:\ \exists\, i<j \text{ such that } |\lambda_i-\lambda_j|<\delta\},
\]
then $\mathbb{P}(\Sigma\in A_\delta)$ is very small for small $\delta$. It is a well-known limitation of the inverse-Wishart prior \citep{berger2020Bayesian}. The $\siw$ family was introduced precisely to alleviate this behavior. 

Reweighting the density to cancel this limitation, \cite{berger2020Bayesian} propose $\siw$, whose density's kernel is  
\begin{equation}\label{eq: siw den}
    \pi(\Sigma | \Psi, \nu, b) \propto \frac{\exp\left(-\frac{1}{2}\mbox{tr}\left( \Psi \Sigma^{-1}\right)\right)}{| \Sigma|^\nu \left[\prod_{i < j} (\lambda_i - \lambda_j)\right]^b},
\end{equation} 
where $\lambda_1 > \cdots > \lambda_K > 0$ are the eigenvalues of $\Sigma$, $\Psi \succ 0$, $b \in [0,1]$ and $\nu \in \R$. Equation \eqref{eq: siw den} is $\iw$ density's kernel divided by $\prod_{i < j} (\lambda_i - \lambda_j)^b$, with $b$ additional parameter controlling the mass on \cov matrices of small eigengaps. Similarly to $\iw$, we reparametrize the parameter $\Psi$ into $(\P, \Vec{\sigma})$. We thus use the notation, $\bSigma \sim \siw_b(\P, \Vec{\sigma}, \nu)$.

Even though the proposed rebalance by $\siw_b$ is for covariance matrices, it induces a reshape in the induced distribution of correlations, allowing us to handle more data cases and yielding benefits in the inference of correlations. A potentially ideal approach would be to use $\siw_b$ with $b \in (0,1)$ to control the desired eigengap of $\Sigma$. However, sampling from a $\siw_b$ with $b \in (0,1)$ is more time-consuming than $\siw_1$ and $\siw_0$ (namely $\iw$), and no algorithms exist at the moment. Note that the only existing algorithms for $\siw_1$, \cite{berger2020Bayesian} and \cite{jiang2025bayesian}, are both for the case $b=1$. Thus, we propose a computationally feasible alternative, which is a mixture of $\iw$ and $\siw_1$ as follows.
\begin{definition} We define the mixture $\bSigma \sim \iw/\siw_1(\eta, \P_r, \Vec{\sigma}_r, \nu_r, r=0,1)$ by 
\begin{equation}\label{eq: IW/SIW}
 \bSigma = \bm z \bSigma_0 + (1 - \bm z)\bSigma_1,    
\end{equation}
where 
\begin{equation}
\begin{cases}
    &\bSigma_0 \sim \mathcal{IW}(\P_0, \Vec{\sigma}_0, \nu_0), \\ 
    &\bSigma_1 \sim \siw_1(\P_1, \Vec{\sigma}_1, \nu_1), \\
    & \bm z \sim \mathcal{B}(\eta)
\end{cases}
\end{equation}
with $\eta \in [0,1]$, $\nu_0 > K - 1$, and $\nu_1 \in \R$.    
\end{definition}
This mixture framework has two computational advantages. First, it remains conjugate to the Gaussian likelihood, by inheritance from both mixture components, that we will show later on in Section \ref{sec: IW/SIW}. Second, posterior sampling from this model reduces to sampling separately from the $\iw$ component and the $\siw$ component. Similarly, the derivation of theoretical properties can be directly based on those of the two components.
In particular, for the derivation of the formula for $\mathbb{E}(\bR_{kk'})$, which allows us to elicit hyperparameter values from expert knowledge about the correlation between regions $k$ and $k'$, we have
\[
\mathbb{E}(\bR_{kk'}) = \eta \,\mathbb{E}(\bR_{0,kk'}) + (1-\eta)\,\mathbb{E}(\bR_{1,kk'}),
\]
where $\bR_r = R(\bSigma_r)$ for $r=0,1$.
The expression of $\mathbb{E}(\bR_{0,kk'})$ has already been established in Corollary \ref{coro: IW prior}. We therefore focus in the following section on the study of $\mathbb{E}(\bR_{1,kk'})$. To the best of our knowledge, a rigorous theoretical analysis for a general $\siw_b$ distribution is very challenging. We therefore provide a theoretical expression for a special case of $\siw_1$, in Theorem \ref{thmsiw}. In the general case, which is the one relevant in practice, numerical methods are required to approximate $\mathbb{E}(\bR_{1,kk'})$. We illustrate the proposed numerical approach for $K=20$, which corresponds to the dimension considered in the real-data application.

In Section \ref{sec: siw}, we first approximate the formula of $\mathbb{E}\left(\bR_{kk'}\right)$ in terms of the parameters of $\siw, b =1$, as well as the one of $\mathbb{V}\left(\bR_{kk'}\right)$. Given these results, in Section \ref{sec: IW/SIW}, we derive the formulas of $\mathbb{E}(\bR_{kk'})$,  $\mathbb{V}(\bR_{kk'})$, $\mathbb{E}(\bR_{kk'}\mid\x_{1:n})$,  $\mathbb{V}(\bR_{kk'}\mid \x_{1:n})$ for $\iw/\siw$ mixture, as well as the other properties including the conjugacy of the mixture prior.

\subsection{Approximation of 
$\mathbb{E}\left(\bR_{kk'}\right)$ and $\mathbb{V}\left(\bR_{kk'}\right)$ for 
$\siw_1$-induced correlations}\label{sec: siw}
We wish to characterize the prior moments for the correlation matrix $\bR$ induced by
$\bSigma \sim \siw_1(\P, \Vec{\sigma}, \nu)$, as in the previous sections, because this is useful for prior specification. However, closed-form expressions are not available in
general. 
We therefore construct approximations by invertible parametric functions. We first give a theoretical result below that motivates linear functions as the approximate form for $\mathbb E(\bR_{kk'})$.

\begin{restatable}{theorem}{thmsiw}
\label{thmsiw}
Let $\bSigma\sim \siw_1(\P,\bm 1,\nu)$ with $K>3$ and $\nu>5$, and let $\bR = R(\bSigma)$.
Define $\varepsilon = \max_{i\neq j}|\P_{ij}|$. As $\varepsilon \to 0$, $K\to\infty$, and $\nu\to\infty$, we have
\[
\mathbb E(\bR_{kk'})
=
\frac{2}{K+2}\,\P_{kk'} + O\!\left(\varepsilon^2 + \frac{1}{K^2\nu}\right),
\qquad \forall\, k \neq k'.
\]
\end{restatable}

The proof is given in Appendix~\ref{app: siw1_first_order}. Theorem~\ref{thmsiw}
shows that, in the weak-correlation regime (small
off-diagonal entries of $\P$) and when $\sigma_k\equiv 1$, the induced prior
mean $\mathbb E(\bR_{kk'})$ is approximately linear in $\P_{kk'}$, with a shrinkage factor
depending mainly on $K$. The objective is providing a rule for
hyperparameter elicitation via the approximate formula of $\mathbb E(\bR_{kk'})$ so that the prior information of the correlation matrix can be translated to hyperparameter values. Given in practice, such information is often itself indicative, a highly accurate approximation is not necessary, in particular when the computational cost is high. We propose $c(K)\P_{kk'}$ as the approximate form for $\mathbb E(\bR_{kk'})$ in general case. To fix the value of $c(K)$, we are based on Monte Carlo estimates. The corresponding numerical scheme is as follows.
\paragraph{Numerical approximation scheme.}
For a fixed dimension $K$, define the collection of hyperparameters
\[
\Theta := \left(\P_{ll'},\, l \neq l',\ \nu,\ \sigma_k,\ k=1,\ldots,K\right).
\]
We evaluate $\mathbb E_\Theta(\bR_{kk'})$ on a predefined grid for $\nu$ and on random
grids for $\P$ and $\Vec{\sigma}$. Each value of $\Theta$ specifies a distribution
$\siw_1(\P,\Vec{\sigma},\nu)$, from which we draw Monte Carlo samples in order to estimate
the induced prior mean of $\bR_{kk'}$. We denote the resulting Monte Carlo estimate by
$\widehat{\mathbb E}_\Theta(\bR_{kk'})$.
Motivated by Theorem~\ref{thmsiw}, we then fit the simple first-order model
\[
\mathbb E_\Theta(\bR_{kk'}) \approx c(K)\P_{kk'}
\]
to the simulated pairs $(\Theta,\widehat{\mathbb E}_\Theta(\bR_{kk'}))$, which yields an
estimate of the shrinkage coefficient $c(K)$. The grid used in the evaluation is as follows:
\begin{itemize}
    \item $\nu = 4, 5, 8, 10, 13, 15, 18, 22, 25, 29, 32, 36, 40$;
    \item for each fixed $\nu$, we test $3$ combinations of values of $\P$ and
    $\sigma_k$, $k=1,\ldots,K$, generated by
    \begin{itemize}
        \item $\P \sim LKJ_K(1)$,
        \item $\sigma_k \stackrel{iid}{\sim} U(0.5,10)$, with
        $\sigma_k \ind \P$, $k=1,\ldots,K$.
    \end{itemize}
\end{itemize}

Figure~\ref{fig: ER_K20} plots the Monte Carlo estimates of $\mathbb{E}_{\Theta} (\bR_{kk'})$ as a function of $\P_{kk'}$ for $K$ fixed by $20$. The black line corresponds to the fitted linear function
$0.066\,\P_{kk'}, \; k\neq k'.$ It can be seen that even though simple the linear form can still translate the prior value on $\mathbb{E}_{\Theta} (\bR_{kk'})$ informatively into the value of $\P_{kk'}$. Moreover the numerical slope is of the same order as the first-order factor $2/K+2 \approx 0.09$ in theorem \ref{thmsiw}. We also colored each evaluation of $\mathbb{E} \bR_{kk'}$ by the value of $\nu$. However, there is no visible color pattern, which implies that adding a dependence on $\nu$ in the parametric function will improve a lot the precision of approximation. The rest of the variation come from the impacts of $\sigma_k, k = 1, \ldots, K$, and possibly other values of $\P_{ll'}$ with $(l,l') \neq (k, k')$. 

Additionally, the range of $\mathbb{E}(\bR_{kk'})$ remains small across the explored configurations, which indicates a strong shrinkage of the induced correlations under $\siw_1$. The small range of $\mathbb{E}(\bR_{kk'})$ in Figure~\ref{fig: ER_K20} suggests that $\siw_1$ is not flexible enough to be used alone when one wishes to encode moderate or large prior correlations, because the prior moment $\mathbb{E}(\bR_{kk'})$ is small whatever the hyperparameter values. Despite its advantages in conjugacy and computation, the induced prior mean remains strongly shrunk toward zero for all explored hyperparameter values.
Since the posterior of the covariance matrix is still a $\siw_1$, the posterior mean of the induced correlations will also be close to zero. A possible explanation could be that $b=1$ introduces too much shrinkage on the eigengaps in the \cov samples, which limits the size of the induced correlation through the following deterministic bound (The proof is given in Appendix~\ref{app: propboundeigengap}.). 
\begin{restatable}{prop}{propboundeigengap}
\label{propboundeigengap}
Let $\Sigma$ be a covariance matrix of size $K$, and $R = R(\Sigma)$. Then for all $k\neq k'$,
\[
|R_{kk'}|
\;\le\;
\frac{\kappa(\Sigma)-1}{\kappa(\Sigma)+1},
\qquad \kappa(\Sigma):=\frac{\lambda_{\max}(\Sigma)}{\lambda_{\min}(\Sigma)}.
\]
\end{restatable}
This would encourage using a smaller for $b$ if $\siw_1$ were to be used alone. 
But as mentioned earlier, posterior sampling of $\siw, b \in (0,1)$ is computationally inconvenient. Therefore, we propose to use a $\iw/\siw_1$ mixture as a feasible computational alternative of $\siw, b \in (0,1)$, via which we explore an enhanced solution of $\iw$.

\begin{figure}[ht]
    \centering
    \includegraphics[width=0.6\linewidth]{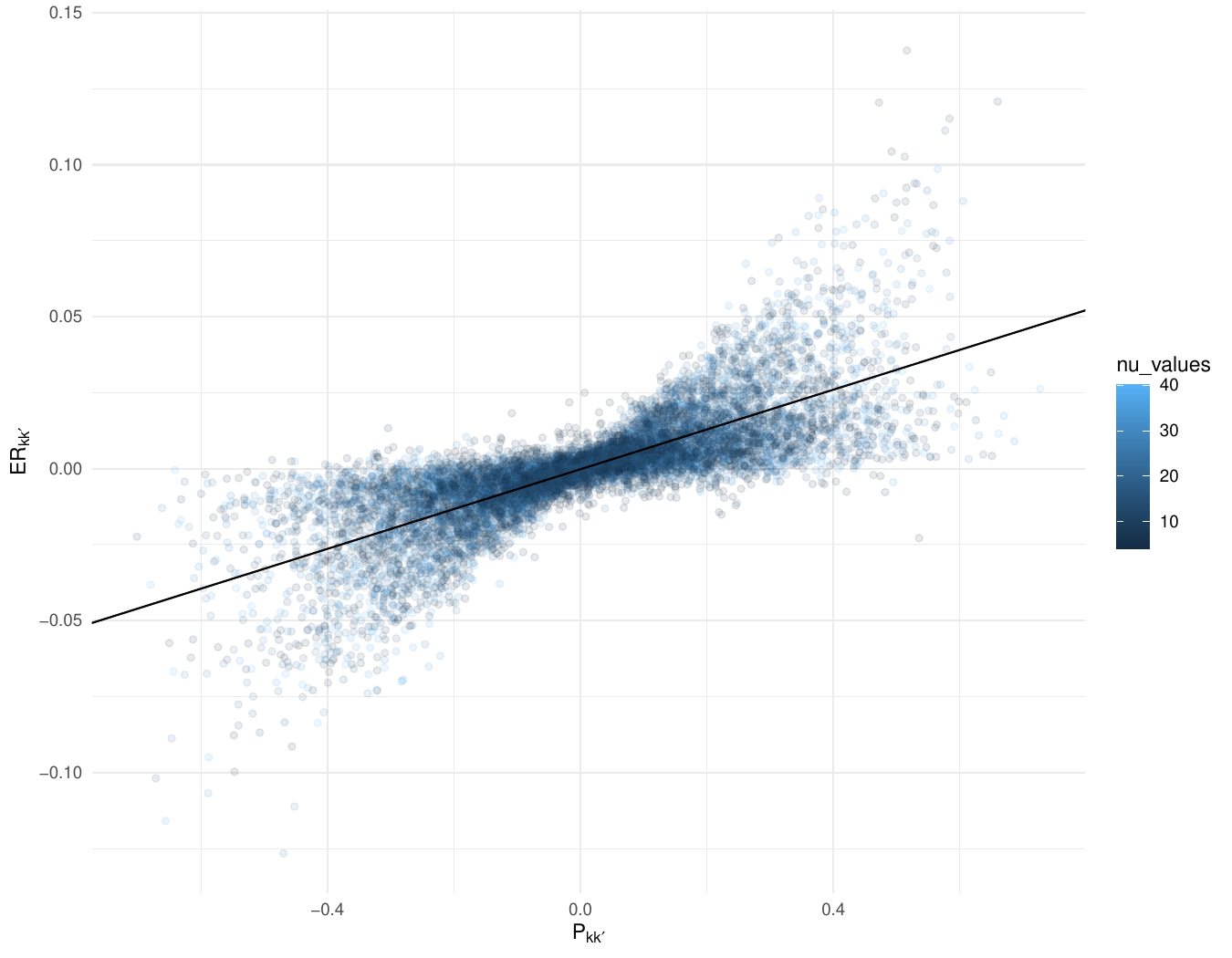}
    \caption{\textit{First-order approximation of $\mathbb{E} \bR_{kk'}$ by $\P_{kk'}$.}
    Motivated by the asymptotic relation in Theorem~\ref{thmsiw}, the black line shows the
    simple linear approximation $\mathbb{E}\bR_{kk'} \approx 0.066\P_{kk'}$, for
    $K = 20$.}
    \label{fig: ER_K20}
\end{figure}


Next, Figure~\ref{fig: VR_K20} plots $\mathbb{V} \bR_{kk'}$ as a function of $\nu$.
Similarly, we adopt a simple parametric approximation in $\nu$:
\[
\mathbb{V}\bR_{kk'} \approx 0.09+\exp(-0.23\nu-1.55).
\]

\begin{figure}[ht]
    \centering
    \includegraphics[width=0.55\linewidth]{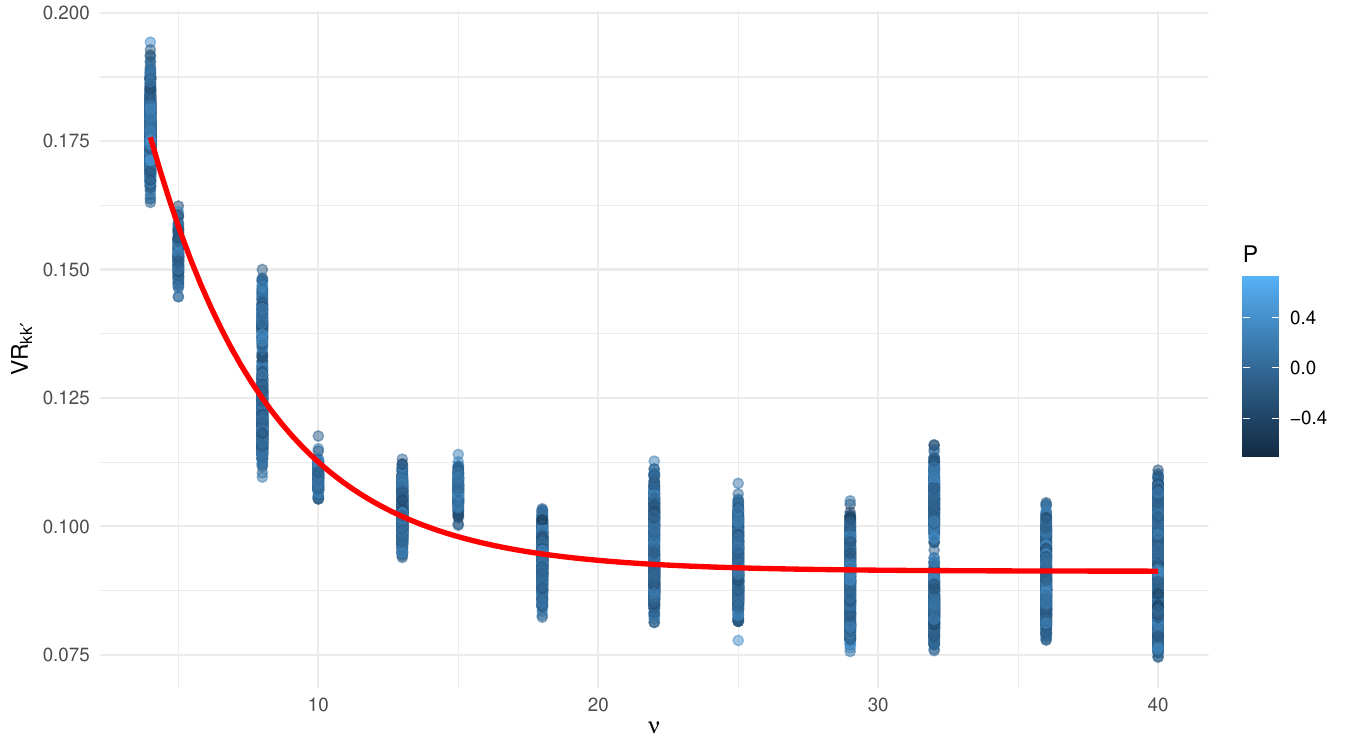}
    \caption{\textit{Approximation of $\mathbb{V} (\bR_{kk'})$ by $\nu$.}
    The fitted approximation is $0.09+\exp(-0.23\nu-1.55)$.}
    \label{fig: VR_K20}
\end{figure}

We summarize the proposed approximation formulas for $\mathbb{E}(\bR_{kk'})$ and
$\mathbb{V}(\bR_{kk'})$ in Property~\ref{prop: siw}, as a counterpart of
Proposition~\ref{prop: IW prior} for $\iw$.

\begin{property}\label{prop: siw} 
Let $\bSigma \sim \siw_1(\P, \Vec{\sigma}, \nu)$ with $K = 20$, $\nu > 3$, and denote
$\bR = R(\bSigma)$. For $k \neq k'$, we have
\begin{align*}
    \mathbb{E}(\bR_{kk'}) &\approx 0.066\P_{kk'},\\
    \mathbb{V}(\bR_{kk'}) &\approx 0.09 + \exp(-0.23\nu -1.55).  
\end{align*}
When $\siw_1(\P, \Vec{\sigma}, \nu)$ is used as prior for the Gaussian likelihood
$\x_i|\bSigma \stackrel{iid}{\sim} \mathcal{N}(0,\bSigma)$, $i = 1,\ldots,n$, we have for
the posterior correlation:
\begin{align*}
    \mathbb{E}(\bR_{kk'} \mid \x_{1:n}) &\approx 0.066\bm\P_{n,kk'}, \\
    \mathbb{V}(\bR_{kk'} \mid \x_{1:n}) &\approx 0.09 + \exp(-0.23\nu_n -1.55),  
\end{align*}
where $\nu_n = \nu+n/2$, and
\begin{equation}
\mathbf{P}_{n,kk'} =
\frac{\sigma_k  \sigma_{k'}  \P_{kk'} + n\bm S_{kk'}}
{\bm\sigma_{n,k}\sigma_{n,k'}},
\quad
\bm\sigma_{n,k} = \sqrt{\sigma_{k}^2 + n \bm S_{kk}},
\quad
\bm S_{kk'} = \frac{\sum_{i=1}^n\bm x_{ki}\bm x_{k'i}}{n}.
\end{equation}
\end{property}
The constraint on $\nu$
ensures that the variance exists. The simple parametric form for the fitting functions allows closed-form inversion, which makes it possible to determine easily the hyperparameter values associated to any desired prior moments.
The posterior formulas come from the conjugacy of $\siw_1$. The same comments as in Theorem~\ref{thm: IW post} can be made here. The posterior mean is impacted by both $\P$ and $\sigma_k$ as in Equation~\eqref{eq: sigma_sample_var0}. The choice of $\sigma_k$ should therefore be made as carefully as that of $\P$. Moreover, for
the mixture prior, $\sigma_k$ also impacts the posterior weights of the components, since the posterior remains an $\iw/\siw_1$ mixture by conjugacy of both components. We provide a strategy for fixing $\sigma_k$ in Section~\ref{sec: IW/SIW}.

We finish this section with an additional theoretical property of exchangeability. Recall that for $\iw$-induced $\bm R$, we have
$Q\bR Q^{-1} \stackrel{d}{=} \bR$ whenever $\P_{kk'} \equiv \rho$ for all $k\neq k'$, as
shown in Proposition~\ref{prop: IW prior}. For $\siw_1$-induced $\bm R$, the same exchangeability requires in addition that $\sigma_k\equiv\sigma$, $k=1,\ldots,K$, as stated
in Proposition~\ref{propExchange}. The proof is given in Appendix~\ref{app: exchange}.

\begin{restatable}{prop}{propExchange}
\label{propExchange}
Let $\bSigma \sim \siw_1(\P, \Vec{\sigma},\nu)$, and denote $\bR = R(\bSigma)$. Given
$\P_{kk'} \equiv \rho$ and $\sigma_{k} \equiv \sigma$, for all
$k\neq k'$, $k,k'=1,\ldots,K$, we have
\begin{equation}
Q\bR Q^{-1} \stackrel{d}{=} \bR,
\end{equation}
for any permutation matrix $Q$.   
\end{restatable}

\subsection{$\iw/\siw_1$ mixture framework}\label{sec: IW/SIW}
Given the previous results, we now derive the properties for the mixture prior. We recall the definition as follows. 
\begin{definition}\label{def: IW/SIW} We define the mixture $\bSigma \sim \iw/\siw_1(\eta, \P_r, \Vec{\sigma}_r, \nu_r, r=0,1)$ by
\begin{equation}\label{eq: IW/SIW_var}
 \bSigma = \bm z \bSigma_0 + (1 - \bm z)\bSigma_1,    
\end{equation}
where 
\begin{equation}
\begin{cases}
    &\bSigma_0 \sim \mathcal{IW}(\P_0, \Vec{\sigma}_0, \nu_0), \\ 
    &\bSigma_1 \sim \siw_1(\P_1, \Vec{\sigma}_1, \nu_1), \\
    & \bm z \sim \mathcal{B}(\eta)
\end{cases}
\end{equation}
with $\eta \in [0,1]$, $\nu_0 > K - 1$, and $\nu_1 \in \R$.    
\end{definition}
As mentioned before, the construction is inspired by $\siw_b$, which reweighs $\iw$ by reducing the density of the matrices with large eigengaps. Given the constraint that sampling from $\siw_b, b\in (0,1)$ is difficult unless $b=1$, the reweighting is realized by mixing $\iw$ with $\siw_1$, with a relative weight controlled by $\eta$.

Given the approximate mean and variance formulas derived previously for the induced correlation independently in the case of $\iw$ and of $\siw_1$, we obtain the following property. 

\begin{restatable}{property}{propertyIWSIWprior}
\label{propertyIWSIWprior}
Let $\bSigma \sim \iw/\siw_1(\eta, \P_r, \Vec{\sigma}_r, \nu_r, r=0,1)$ with $K =20, \nu_0 > 23, \nu_1 > 3$. Then 
 \begin{enumerate}
 \vspace{-0.05in}
    \item $ \mathbb{E}(\bR_{kk'}) \approx
    \eta \P_{0,kk'} + 0.066(1-\eta)\P_{1,kk'},$ 
     \vspace{-0.05in}
    \item $\mathbb{V}
    (\bR_{kk'}) \approx \frac{\eta}{\nu_0 - 19} + (1-\eta)(0.09 + \exp(-0.23\nu_1 -1.55)) + \eta(1-\eta)(\P_{0,kk'})^2.$ 
\end{enumerate}   

\end{restatable}
The proof is given in Appendix \ref{app: mean var mix}.
The prior mean is controlled mainly by $P_{0,kk'}$ and $\eta$, while the prior variance is controlled by $\nu_r, r=0,1, \eta$ and $P_{0,kk'}$.

The exchangeability can be derived easily from the ones of $\iw$ and $\siw_1$, as stated by the following result. 
\begin{corollary}\label{coro: exchange}
 Let $\bSigma \sim \iw/\siw_1(\eta, \P_r, \Vec{\sigma}_r, \nu_r, r=0,1)$, and denote $\bR = R(\bSigma)$. Given $\P_{0,kk'} \equiv \rho_0, \P_{1,kk'} \equiv \rho_1$ and $\sigma_{1,k} \equiv \sigma, \forall k \neq k', k, k', = 1, \ldots, K$:
\begin{equation}
Q\bR Q^{-1} \stackrel{d}{=} \bR,
\end{equation}
for any permutation matrix $Q$. 
\end{corollary}
Even though the strict identity of the marginals of $\bm R_{kk'}$ in the $\iw / \siw$ mixture case can be only obtained when all $\sigma_{1,k}$ are fixed by a unique constant $\sigma$. However, in most cases, approximately equality is already enough. Given Property \ref{propertyIWSIWprior}, when $\P_{r,kk'} \equiv \rho_r, r=0,1$, the resulting expectation $\mathbb{E}(\bR_{kk'})$ for all pairs will still be around the same value $\eta  \rho_0 + 0.066(1-\eta)  \rho_1$, with the exact values impacted by $\sigma_{1,k}$ in an asymmetric way. However, if one wishes the strict identity, $\sigma_{1,k}$ needs to be fixed identically. When the observation obviously does not indicate all the component variances are equal, rescaling can be performed on raw data. Note that rescaling does not change the correlation value.

We now study the posterior properties of the $\iw / \siw$ mixture prior with Gaussian likelihood. Firstly, the mixture inherits the Gaussian-conjugacy from the components of $\mathcal{IW}$ and $\mathcal{SIW}$ as shown in Proposition \ref{propIWSIWpost}.
\begin{restatable}{prop}{propIWSIWpost}\label{propIWSIWpost} 
Given that   
$\x_i
|\bSigma
\stackrel{iid}{\sim} \mathcal{N}\left(
0
,
\bSigma\right), \, i = 1, \ldots, n$, and $\bSigma \sim  \iw/\siw_1(\eta, \P_r, \Vec{\sigma}_r, \nu_r, r=0,1)$, we have 
\begin{equation}
\bSigma | \x_{1:n} \sim \iw/\siw_1(\eta_n, \bm\P_{r,n}, \Vec{\bm \sigma}_{r,n}, \nu_{r,n}, r=0,1),
\end{equation}
where the updating rules are: 
\[
\begin{aligned}
     &\nu_{0,n} = \nu_0 + n, \nu_{1,n} = \nu_1 + n/2,\\
    &\bm\sigma_{r,n,k} = \sqrt{(\sigma_{r,k})^2 + n \bm S_{kk}},\\
    &\bm\P_{r,n,kk'} = \frac{\sigma_{r,k}\sigma_{r,k'}\P_{r,kk'} + n\bm S_{kk'}}{ \bm\sigma_{r,n,k}\bm\sigma_{r,n,k'}}, \mbox{ with }     \bm S_{kk'} = \sum_{i=1}^n\bm x_{ki}\bm x_{k'i} /n,   
\end{aligned}
\]
and the posterior mixture weight 
\begin{equation}
    \eta_{n} = \frac{\eta L(\x_{1:n} \mid \bm z = 0)}{L(\x_{1:n})}, 
\end{equation}
with $ L(\x_{1:n}) 
= \eta L(\x_{1:n} \mid \bm z = 0)  + (1-\eta)L(\x_{1:n} \mid \bm z=1), $
$$L(\x_{1:n} \mid \bm z = 0) 
= \mathbb{E}_{\bSigma \sim \mathcal{IW}(\P_0, \Vec{\sigma}_0, \nu_0)}\left[\prod_{i=1}^n p_{\mathcal{N}}(\bm x_i; 0, \bSigma)\right],$$
$$L(\x_{1:n} \mid \bm z=1) 
= \mathbb{E}_{\bSigma \sim \siw_1(\P_1, \Vec{\sigma}_1, \nu_1)}\left[\prod_{i=1}^n p_{\mathcal{N}}(\bm x_i; 0, \bSigma)\right],$$
where $p_\mathcal{N}(\cdot; \mu, \Sigma)$ denotes the probability density function of $\mathcal{N}(\mu, \Sigma)$.   

\end{restatable}

See Appendix \ref{app: IW/SIW post} for a proof. The evaluation of $L(\x_{1:n} \mid \bm z = 0)$ is detailed in Section \ref{sec: sampling}. The result indicates that the $\iw$ and $\siw_1$ components constituting the posterior mixture are exactly the posteriors of the individual models:
\begin{equation}\label{eq: model ind iw}
\x_{1:n} \stackrel{iid}{\sim} \mathcal{N}(0, \bSigma_0), \, 
    \bSigma_0 \sim \mathcal{IW}(\P_0, \Vec{\sigma}_0, \nu_0), 
\end{equation}
and 
\begin{equation}\label{eq: model ind siw}
\x_{1:n} \stackrel{iid}{\sim} \mathcal{N}(0, \bSigma_1), \, 
    \bSigma_1 \sim \siw_1(\P_1, \Vec{\sigma}_1, \nu_1),  
\end{equation}
in another words, when the $\iw$ and $\siw_1$ components in the prior mixture are used independently with the same Gaussian likelihood. Thus the posterior $\bSigma | \x_{1:n}$ can be also represented as: 
\[ \bSigma | \x_{1:n} = \bm z_n \bSigma_0 | \x_{1:n} + (1 - \bm z_n)\bSigma_1 | \x_{1:n},\]
where $\bm z_n \sim \mathcal{B}(\eta_n)$. Furthermore, the constants $L(\x_{1:n} \mid \bm z = r), r=0,1$ used in the weight update are the likelihoods of individual models. Thus if the $\iw$ prior is closer to the data, then the corresponding individual model is better, leading to a higher likelihood of $L(\x_{1:n} \mid \bm z = 0) $. Its posterior weight $\eta_n$ will increase with respect to its prior weight $\eta$, similar for the $\siw_1$ prior. The mixture posterior will therefore be closer to the $\iw$ posterior. An extreme case is that $\siw_1$ does not fit the data, thus $L(\x_{1:n} \mid \bm z = 0) \gg L(\x_{1:n} \mid \bm z = 1)$, thus $\eta_n = 1$. The mixture becomes the $\iw$ independent model with 
\[ \bSigma | \x_{1:n} = \bSigma_0 | \x_{1:n}.\]
For the induced correlation matrix $\bR$, we have 
\begin{equation}\label{eq: cor mix}
\bm R | \x_{1:n} = \bm z_n \bm R_0 | \x_{1:n} + (1 - \bm z_n)\bm R_1 | \x_{1:n}, 
\end{equation}
where $\bm R_r = R(\bSigma_r), \, r = 0,1.$

Because the posterior is still a $\iw/\siw_1$ mixture, applying the posterior results of $\iw$ and $\siw_1$ respectively developed in Theorem \ref{thm: IW post} and Proposition \ref{prop: siw}, we can also approximate the posterior mean and variance of the mixture, as shown in Property \ref{prop: iw siw post corr}. 

\begin{property}\label{prop: iw siw post corr}
Given that   
$\x_i
|\bSigma
\stackrel{iid}{\sim} \mathcal{N}\left(
0
,
\bSigma\right), \, i = 1, \ldots, n$, and $\bSigma \sim  \iw/\siw_1(\eta, \P_r, \Vec{\sigma}_r, \nu_r, r=0,1)$,  with $K =20, \nu_0 > 23, \nu_1 > 3$, we have 
 \begin{enumerate}
 \vspace{-0.05in}
    \item $ \mathbb{E}(\bR_{kk'} \mid \x_{1:n} ) \approx
    \eta_n \bm\P_{0,n,kk'} + 0.066(1-\eta_n)\bm\P_{1,n,kk'},$ 
     \vspace{-0.05in}
    \item $\mathbb{V}
    (\bR_{kk'} \mid \x_{1:n} ) \approx \frac{\eta_n}{\nu_{0,n} - 19} + (1-\eta_n)(0.09 + \exp(-0.23\nu_{1,n} -1.55)) + \eta_n(1-\eta_n)(\bm\P_{0,n,kk'})^{2},$ 
\end{enumerate}   
where $\eta_n, \bm\P_{r,n}, {\bSigma}_{r,n}, \nu_{i,n}, r=0,1$ are defined in Proposition \ref{propIWSIWpost}. 
\end{property}

\subsubsection{Choice of hyperparameters $\sigma_{rk}$ of the mixture prior}\label{sec: Designation sigma}

Analogous to Section \ref{sec: iw}, we can obtain simple posterior update formulas by setting $\Vec{\sigma_r}$ such that the prior expected variance matches the sample variance, namely, $\mathbb{E}\left[\bSigma_{kk} | \bm z = r\right] = \bm S_{kk}$, for both the $\iw$ ($r=0$) and $\siw_1$ ($r=1$) components.

For the $\iw$ component ($r=0$), we have from the properties of the $\iw$ distribution the following result: 
\begin{equation}\label{eq: nnn}
\mathbb{E}_{\bSigma \sim \iw(\P_0, \Vec{\sigma_0}, \nu_0)} (\bSigma_{kk}) = \sigma_{0k}^2/(\nu_0 - K - 1).    
\end{equation}
Therefore, setting $\sigma_{0k} = \sqrt{(\nu_0 - K -1)\bm S_{kk}}$ yields the simple and intuitive closed-form expression for the posterior update:
\begin{equation}
 \mathbf{P}_{0, n,kk'} 
 = \frac{\nu_0 - K -1}{\nu_0 - K -1+n}\P_{0,kk'} + \frac{n}{\nu_0 - K -1+n} \frac{\bm S_{kk'}}{\sqrt{\bm S_{kk}}\sqrt{\bm S_{k'k'}}},
\end{equation}
where $\nu_0$ seamlessly controls the shrinkage towards the prior. The same principle applies to the $\siw_1$ component ($r=1$) through the choice of $\sigma_{1k}$. However, 
there is no closed form expression for $\mathbb{E}_{\bSigma \sim \siw_1(\P_1, \Vec{\sigma_1}, \nu_1)} \bSigma_{kk}$. Thus as previously, we approximate $\mathbb{E}_{\bSigma \sim \siw_1(P_1,\Vec{\sigma_1},\nu_1)} (\bSigma_{kk})$ numerically. The results are reported in Figure \ref{fig:sigma_sigma}.

Similar to Equation \eqref{eq: nnn}, $\mathbb{E}_{\bSigma \sim \siw_1(P_1,\Vec{\sigma_1},\nu_1)} (\bSigma_{kk})$ increases with $\sigma_{1k}$, and decreases with $\nu_1$. Thus the value of  $\sigma_{1k}$ allowing $\mathbb{E}_{\bSigma \sim \siw_1(P_1,\Vec{\sigma_1},\nu_1)} (\bSigma_{kk})$ to approach $\bm S_{kk}$ increases with $\nu_1$ as well. Implicitly, we make $\nu_1$ control the impact $\P_{1,kk'}$ in $\mathbf{P}_{1, n,kk'}$. Since $\P_{0,kk'}$ and $\P_{1,kk'}$ define the prior correlation mean, and $\P_{0,n,kk'}$ and $\P_{1,n,kk'}$ define the posterior one. We have a similar model interpretation, that $\nu_r, r=0,1$ controls the weight of prior correlation mean in the posterior mean, higher values leading to greater weight. 

Further, we remark that Figure \ref{fig:sigma_sigma} shows that $\mathbb{E}_{\bSigma \sim \siw_1(P_1,\Vec{\sigma_1},\nu_1)} (\bSigma_{kk})$ can not reach small values near zero, no matter the value of $\sigma_{1k}$. If the sample variance falls within this unreachable range, we must rescale the raw data.

Finally, when a priori there is no constraint on the identity of the correlation marginal, $\sigma_{1k}$ can be left distinct. In this case, the choice of $\sigma_{rk}, r =0, 1$ could be made based on their impacts on the posterior distribution of $\bR$. Proposition \ref{propIWSIWpost} shows that in the mixture case, $\sigma_{rk}$'s impact on the posterior mean correlation remains the same as the $\iw$ component model shown in Equation \eqref{eq: sigma_sample_var0}, in that it controls the weight of the prior mean correlation in the posterior. In addition, $\sigma_{rk}$ influences the posterior mixture weight $\eta_n$. For example, a component prior on $\bSigma$ whose variance is far from the data variance (due to a poor choice of $\sigma_{rk}$) will be severely down-weighted in the posterior mixture, even if its correlation structure matches the data well.

\begin{figure}
    \centering
\includegraphics[width=0.7\linewidth]
{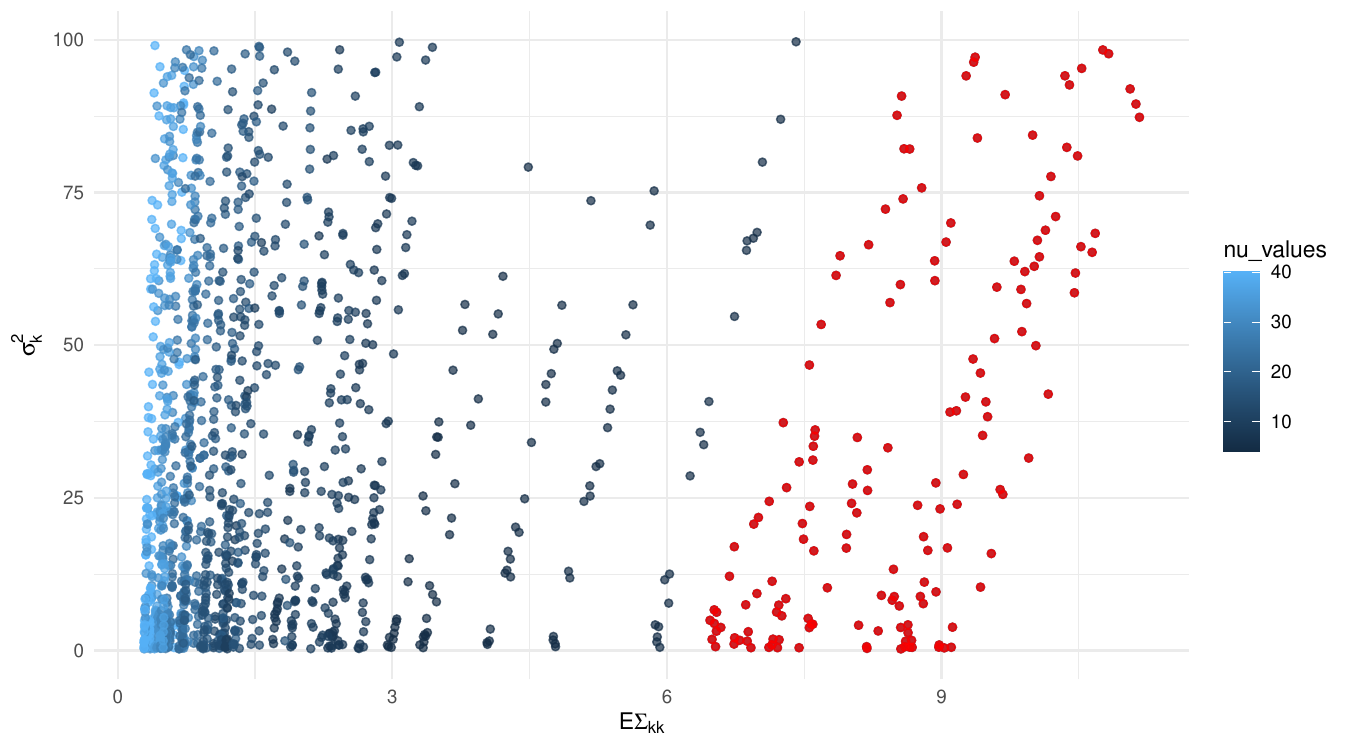}
    \caption{\textit{Values of $\mathbb{E}_{\bSigma \sim \siw_1(P,\Vec{\sigma},\nu)} (\bSigma_{kk})$ projected onto $\sigma_k^2$.} The red dots corresponds to $\nu = 4$. 
    In practice, we fix first $\nu$ according to the desired $\mathbb{V}_{\bSigma \sim \siw_1(P,\Vec{\sigma},\nu)} (\bR_{ll'})$. Given the value fixed, we approximate $\mathbb{E}_{\bSigma \sim \siw_1(P,\Vec{\sigma},\nu)} (\bSigma_{kk})$ using a linear function of $\sigma_k$. Then $\sigma_k$ is fixed such that $\mathbb{E}_{\bSigma \sim \siw_1(P,\Vec{\sigma},\nu)} (\bSigma_{kk})$ equals sample variance.
    }
    \label{fig:sigma_sigma}
\end{figure}



\subsubsection{Posterior inference}\label{sec: sampling}
A computational advantage of $\iw / \siw_1$ mixture is that posterior sampling does not require new algorithms. One only needs to sample a binary allocation variable to one of the $\iw$ and $\siw_1$ with a probability equal to the mixture weights, then to generate a sample from the corresponding $\iw$ or $\siw_1$. 
Fast sampling solutions for $\iw$ and $\siw_1$ both exist (packages in R and \cite{jiangsiw25}), and it is possible to combine them to build an efficient and parallelisable approximate mixture sampling algorithm. As mentioned at the beginning of Section \ref{sec: siw0}, $\siw_b, b \in (0,1)$ would not have such a computational advantage. In the following, we derive the evaluation of $ L(\x_{1:n} \mid \bm z = r), r=0,1$ needed to compute the posterior mixture weights. First, $ L(\x_{1:n} \mid \bm z = 0)$ admits an explicit formula thanks to properties of the $\iw$ distribution.

\begin{restatable}{prop}{ppropL}\label{propL}
Given $\x_{1:n} \stackrel{iid}{\sim} \mathcal{N}(0, \bSigma), \, 
    \bSigma \sim \mathcal{IW}(P_0, \Vec{\sigma_0}, \nu_0),$ define $\Psi_0 = \Delta_0 P_0 \Delta_0$ with $\Delta_0 = \operatorname{Diag}(\Vec{\sigma_0})$, we have 
\begin{equation}
    L(\x_{1:n} \mid \bm z = 0) =\frac{c_\iw(\Psi_0, \nu_0)}{(2\pi)^{\frac{KT}{2}}c_\iw(\Psi_0+ \bm S, \nu_0+n)},    
\end{equation}
where $\bm S = \sum\limits_{i=1}^n\x_i\x_i^\top$, and $c_\iw(\Psi_0, \nu_0) = |\Psi_0|^{\frac{\nu_0}{2}}/(2^{\frac{\nu_0K}{2}})\Gamma_K(\frac{\nu_0}{2})$ with $\Gamma_K$ multivariate gamma function. 
\end{restatable}
The proof is provided in Appendix \ref{app: L1}. In practice, calculations should be performed on the logarithmic scale to avoid overflow and underflow, which can occur easily when $K$ or $n$ are large. Thus, we will use the formula:
\begin{equation}
    \log L(\x_{1:n} \mid \bm z = 0) = -\frac{KT}{2}\log (2\pi) + \log c_\iw(\Psi_0, \nu_0) - \log c_\iw(\Psi_0+S, \nu_0+n).
\end{equation}
Second, $L(\x_{1:n} \mid \bm z=1)$ does not admit a closed form because the density of $\siw$ is only defined up to a normalization constant, but it can be evaluated numerically. An adaptation of Algorithm 3 in \cite{jiangsiw25} can be used to obtain an estimate of $L(\x_{1:n} \mid \bm z=1)$ via self-normalized importance sampling with a theoretical guarantee of consistence. Although the resulting estimator is biased, the computational efficiency makes it possible to obtain sufficiently accurate estimates within a short amount of time. The adapted algorithm and the setting we use for the data sets in Section \ref{sec: app} are given in Appendix \ref{app: post inf}.


\section{Application on rats data set}\label{sec: app}

In this section, we apply the two developed methods: $\iw$ and $\iw/\siw_1$ on a real data set. We first report point estimation of the functional connectivity in Section~\ref{sec: point estim}. For the $\iw$ model, the point estimation is given by the posterior mode of $\bm R$, which is more robust. For the $\iw / \siw$ model, it is given by the posterior mean, because the posterior is bi-modal. We test different hyperparameter values and show their corresponding estimations to illustrate the impacts of the former. Secondly, to highlight the advantage of Bayesian inference over frequentist inference, especially on the uncertainty evaluation, we derive a procedure to detect regional pairs $(k,k')$ of significant connectivity based on the corresponding posterior distribution of $\bm R_{kk'}$ in Section~\ref{sec: detection}. We apply the same procedure on the posteriors of the competitor Bayesian model \cite{acharyyaBayesianHierarchicalModeling2023}. The detection results indicate superior performance of our methods. 
We present the details of the real data and the hyperparameter elicitation as follows. 


\paragraph{Rats data set.} 
The data\footnote{available at \url{https://zenodo.org/records/2452871}.} concern experiments to evaluate the influence of anesthetics on the brain functional connectivity of rats. Time series of $3600$ time points are extracted from resting state fMRI on brain areas defined in an atlas, consisting of 51 regions. After the preprocessing using wavelets implemented by the  R package $\mathtt{waveslim}$ \citep{achard2016multivariate}, we have $219$ time steps. 
In this work, we show the results over a dead rat and an live rat. 

For this real dataset, prior information on correlation values is available only at a
global level. We therefore consider the priors
$\mathcal{IW}(\P,\Vec{\sigma},\nu)$ and
$\iw/\siw_1(\eta,\P_r,\Vec{\sigma}_r,\nu_r,\ r=0,1)$ with reduced correlation
hyperparameter matrices, namely
\[
\P_{kk'} \equiv \rho,
\qquad
\P_{r,kk'} \equiv \rho_r,\quad r=0,1,\quad k\neq k'.
\]
We fix $\eta$ by $0.5$, other details are given as follows.
\paragraph{Hyperparameter elicitation of $\P$ through $\mathbb{E}(\bR_{kk'})$.}
For the live rat, we use the same prior belief as in
\cite{gorbachHierarchicalBayesianMixture2020, chenBayesianHierarchicalFramework2016}:
\begin{center}
\textit{only positive correlations correspond to reliable functional connections}.    
\end{center}
This prior belief can be translated into positive prior means for the induced correlations, which leads to
\[
\mathbb{E}(\bR_{kk'}) \approx \P_{kk'} = \rho >0,
\]
for the $\iw$ prior, and
\[
\mathbb{E}(\bR_{kk'})
\approx
\eta \P_{0,kk'} + 0.066(1-\eta)\P_{1,kk'}
=
0.5\rho_0 + 0.5\times 0.066\rho_1
>0,
\]
for the $\iw/\siw_1$ mixture prior. Since no precise prior information is available on the exact level of the correlations, we test several positive target values for $\mathbb{E}(\bR_{kk'})$, which correspond to different choices of $\rho$ and $\rho_r$, reported with their corresponding results in Sections \ref{sec: point estim} and \ref{sec: detection}. 
 
For the dead rat, we use instead the prior belief that:

\begin{center}
 \textit{no functional correlation should be favored a priori}.    
\end{center}
We therefore set $\P=I_K$ and $\P_r=I_K$, $r=0,1$, which gives for both priors 
\[
\mathbb{E}(\bR_{kk'}) \approx 0,\qquad k\neq k'.
\]

\paragraph{Hyperparameter elicitation of $\nu$ through $\mathbb{V}(\bR_{kk'})$.} The choice of $\nu$ is made after fixing $\P$, because the approximate formula for $\mathbb{V}(\bR_{kk'})$ depends on both $\P$ and $\nu$. In this application, we do not have specific prior information on the variability of the correlations. We therefore do not elicit a single value of $\nu$ from expert knowledge. Instead, we evaluate several values of $\nu$ and compare the resulting posterior estimates.

\paragraph{Choice of $\Vec{\sigma}$.}
One could treat $\Vec{\sigma}$ as an independent hyperparameter and elicit it from
additional expert knowledge, for instance through formulas involving
$\mathbb{E}(\bSigma_{kk})$. However, this would make model interpretation more
complicated. In this work, we instead fix $\Vec{\sigma}$ using the data and the chosen
value of $\nu$, following the strategy detailed in Section~\ref{sec: Designation sigma}.
This choice gives $\nu$ a more consistent interpretation: it controls both the prior
variance of the induced correlations and the relative weight of $\P$ in the posterior mean.

\paragraph{Combining prior elicitation with data-driven model selection.}
When prior information is not sufficient to determine all hyperparameters, as in the present application, one may complement elicitation with data-driven model selection\footnote{Note that, such tuning is driven by the observed data and may therefore be sensitive to outliers or other data-quality issues. Consequently, although information criteria can in principle be used to tune all hyperparameters, reliable expert knowledge should still be incorporated through prior elicitation whenever available. This helps obtain more robust inference by preventing the model from relying exclusively on possibly noisy or contaminated data.}. For example, information criteria adapted to Bayesian models, such as leave-one-out
cross-validation or WAIC \citep{vehtari2017practical}, can be used to compare models' predictive fit associated to different hyperparameter settings. 

\subsection{Point estimation of functional connectivity}\label{sec: point estim}
In this section, we report point estimates of the correlation matrix, defined as posterior
means. We first present the results obtained with the $\iw$ prior in
Figure~\ref{fig: rats}.

For the dead rat, the prior belief can be elicited precisely by setting $\rho=0$, since no
functional correlation is expected a priori. However, no prior information is available for
$\nu$. We therefore consider two values, $\nu=53$ and $\nu=83$, corresponding respectively
to a flatter and a more concentrated prior distribution on the induced correlations
$\bR_{kk'}$. We compare the two resulting models using the information criterion proposed
in \cite{vehtari2017practical}, and the best predictive fit is obtained for $\nu=83$.

Recall that Theorem~\ref{thm: IW post} implies that the posterior mean of the induced
correlation is approximately $\mathbf{P}_{n,kk'}$. Moreover, when the hyperparameter
$\sigma_k^2$ is fixed as $(\nu-K-1)\bm S_{kk}$ and $\P_{kk'} \equiv \rho$, Equation~\eqref{eq: sigma_sample_var}
yields
\begin{equation}\label{eq:tiredof naming}
 \mathbf{P}_{n,kk'} 
 = \frac{n_0}{n_0+n}\rho + \frac{n}{n_0+n}
 \frac{\bm S_{kk'}}{\sqrt{\bm S_{kk}}\sqrt{\bm S_{k'k'}}}, 
\end{equation}
where $n_0 = \nu - 52$ and $n = 219$. Thus, when $\nu=53$, the prior belief encoded in
$\P_{kk'}$ has almost no impact, and the posterior mean is nearly equal to the sample
correlation. By contrast, $\nu=83$ gives more weight to the prior matrix. Since
$\rho = 0$ for the dead rat, this leads to posterior mean correlations that are
more strongly shrunk toward zero.

For the living rat, the prior information on $\rho$ is only qualitative: reliable
functional connections are expected to correspond to positive correlations. We therefore
test two positive values, $\rho=0.3$ and $\rho=0.5$, together with the same two values of
$\nu$. As expected from Equation~\eqref{eq:tiredof naming}, for a fixed $\rho$, a larger
value of $\nu$ gives more weight to the prior and therefore leads to higher posterior
correlation levels. Conversely, for a fixed $\nu$, a larger value of $\rho$ also leads to
higher posterior correlation levels. Applying the same information criterion to the four
resulting models, the best predictive fit is obtained for $\rho=0.3$ and $\nu=83$. Moreover, the models with two positive $\rho$'s and larger $\nu$ both have better predictive fit than the models with nearly flat prior\footnote{
We use  the expected log pointwise predictive density
estimated by leave-one-out cross-validation \citep{vehtari2017practical}. Larger values
indicate better expected out-of-sample predictive performance. The corresponding values are:
\[
\begin{array}{c c c}
\hline
 & \nu = 53 & \nu = 83 \\
\hline
\rho= 0.3 & 6373.643 & 6536.627 \\
\rho= 0.5 & 6361.904 & 6511.489 \\
\rho= 0.7 & 6335.595 & 6463.318 \\
\hline
\end{array}
\]
The row $\rho_0=0.7$ corresponds to two additional hyperparameter settings included in the
model comparison but not displayed in Figure~\ref{fig: rats}.
}. 

For both rats, the models with $\nu=83$ fit the data better than those with
$\nu=53$, whose posterior means are almost identical to the empirical correlation matrix.
Although a direct comparison between the Bayesian estimator and the frequentist estimator
is not fully meaningful, with the former a posterior distribution whereas the latter a
single point estimate, the case $\nu=53$ provides a useful benchmark, since its posterior
mean is essentially the classical sample correlation. The fact that the best models are
those in which the posterior mean is more influenced by the prior belief illustrates the benefit of Bayesian inference when reliable expert knowledge is available. Moreover, for
the living rat, the fact that the models with both positive $\rho$'s and larger $\nu$ improve over
the nearly flat-prior benchmark implies that even coarse prior information can improve estimation. It also supports the discussion in Section~\ref{sec: siw} that highly accurate
approximations of prior moments are not always necessary, because simple directional prior
information may already be useful for posterior inference. 

\begin{figure}
\begin{tabular}{ccc}

\begin{minipage}[t]{0.05\linewidth} 
\vspace{0.5in}
\rotatebox[origin=c]{90}{\small Dead} \end{minipage}&
\begin{minipage}[t]{0.5\linewidth}\centering
{\small Flat prior (\(\nu=53\))\par}\vspace{2pt}

\includegraphics[width=0.48\linewidth]{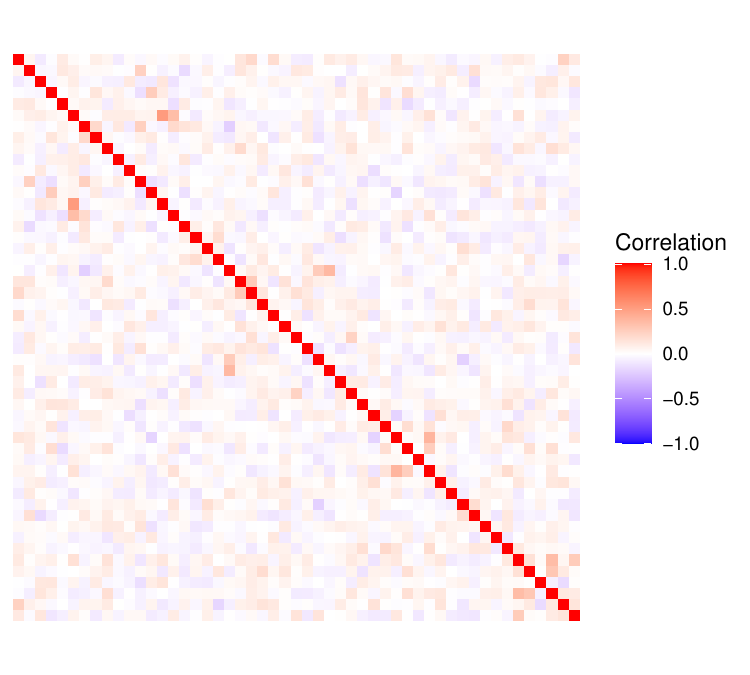}\hfill
\par

\vspace{1pt}
\begin{tabular}{@{}p{0.48\linewidth}@{}p{0.48\linewidth}@{}}
\hspace{0.3in}\scriptsize\(\rho_0=0\) 
\end{tabular}
\end{minipage}
&
\begin{minipage}[t]{0.5\linewidth}\centering
{\small Informative prior (\(\bm \nu=\bm{83}\))\par}\vspace{2pt}

\includegraphics[width=0.48\linewidth]{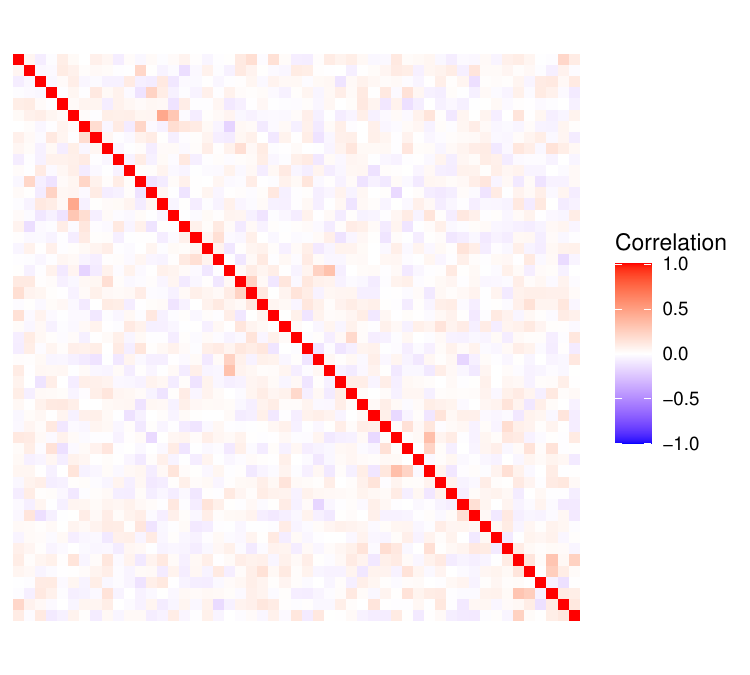}\hfill
\par

\vspace{1pt}
\begin{tabular}{@{}p{0.48\linewidth}@{}p{0.48\linewidth}@{}}
\hspace{0.3in}\scriptsize\(\bm{\rho_0=0}\) 
\end{tabular}
\end{minipage}
\\ 
\begin{minipage}[t]{0.05\linewidth} 
\rotatebox[origin=c]{90}{\small Live} 
\end{minipage}&
\begin{minipage}[t]{0.5\linewidth}\centering
\includegraphics[width=0.48\linewidth]{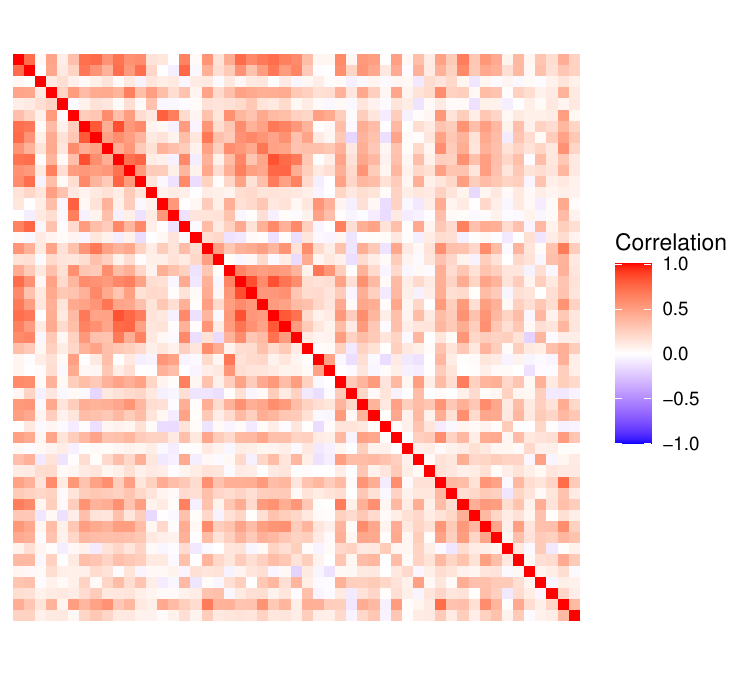}\hfill
\includegraphics[width=0.48\linewidth]{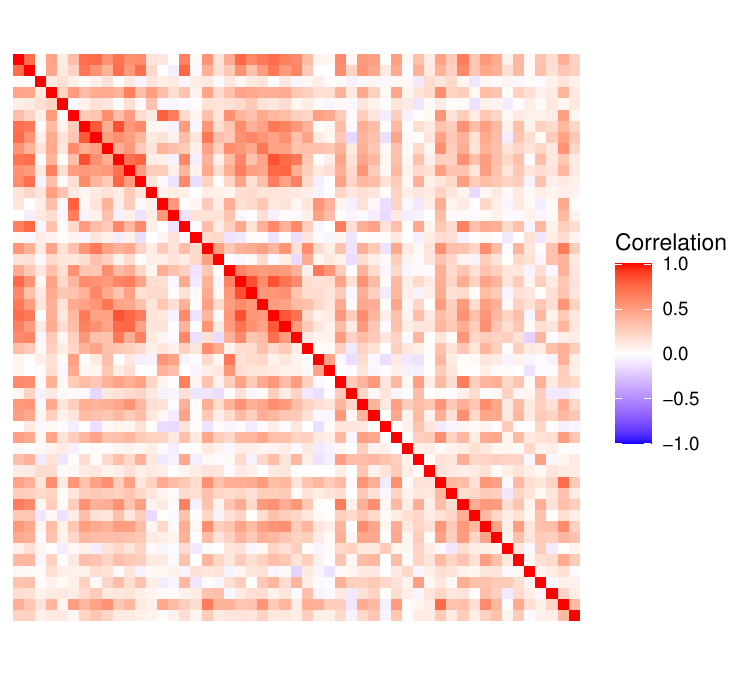}\par

\vspace{1pt}
\begin{tabular}{@{}p{0.48\linewidth}@{}p{0.48\linewidth}@{}}
\hspace{0.3in}\scriptsize\(\rho_0=0.3\) &\hspace{0.3in} \scriptsize\(\rho_0=0.5\)
\end{tabular}
\end{minipage}
&
\begin{minipage}[t]{0.5\linewidth}\centering

\includegraphics[width=0.48\linewidth]{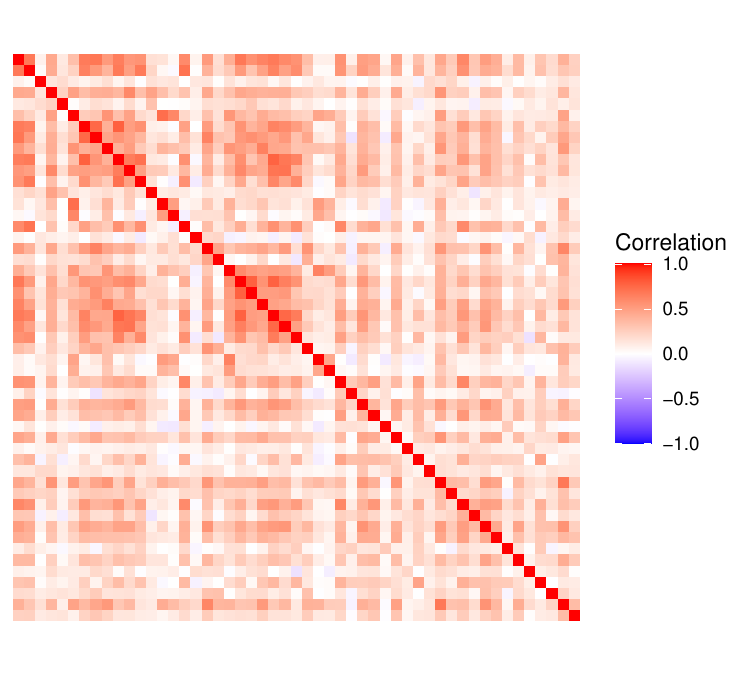}\hfill
\includegraphics[width=0.48\linewidth]{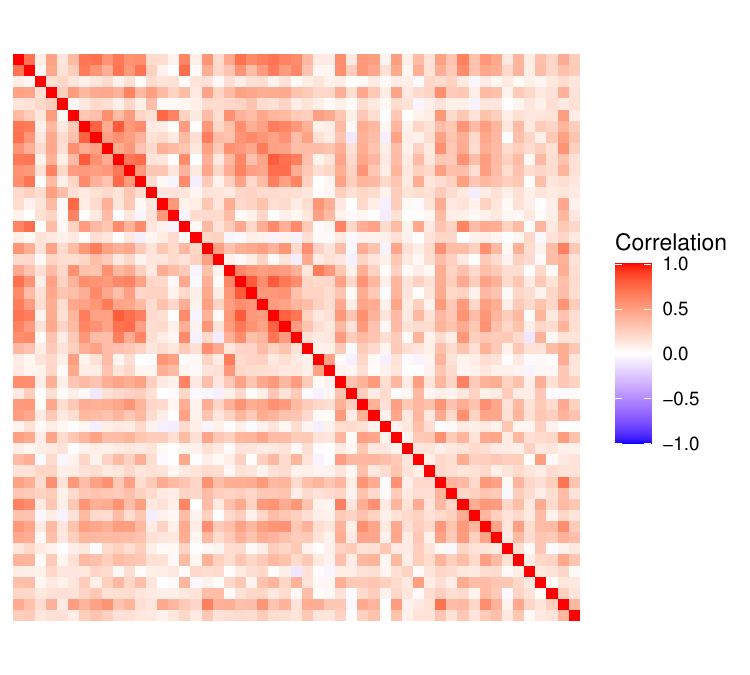}\par

\vspace{1pt}
\begin{tabular}{@{}p{0.48\linewidth}@{}p{0.48\linewidth}@{}}
\hspace{0.3in}\scriptsize\(\bm{\rho_0=0.3}\) &\hspace{0.3in} \scriptsize\(\rho_0=0.5\)
\end{tabular}
\end{minipage}

\end{tabular}
\vskip-0.1in
\caption{Posterior mean correlation matrices from the two rats. The hyperparameter settings in bold are considered better based on the information criteria value proposed in \cite{vehtari2017practical}, which always correspond to informative prior with $\nu = 83$ for both rats.}
 \label{fig: rats}
\end{figure}    

Next, we show the results of the $\iw/\siw_1$ model. We first present additional data preparation as follows. 

\paragraph{Additional data preparation for the $\iw/\siw_1$ model.} We select $20$ regions for both dead and live rats, recall that our analysis with numeric methods is illustrated with $K=20$. For the dead rat, we remove the regions that have a large sample correlation with other regions, possibly due to acquisition artefacts. For the live rat, we select such that the sample correlation between selected regions does not have many large values; otherwise, we observed failures of the mixture model by always updating $\eta_n$ as $1$, which means it only considers $\iw$ prior. 
The sample correlation between the 20 selected regions for the live rat is represented by the histogram in Figure \ref{fig: 20 regions alive}. 
\begin{figure}
    \centering
    \includegraphics[width=0.5\linewidth]{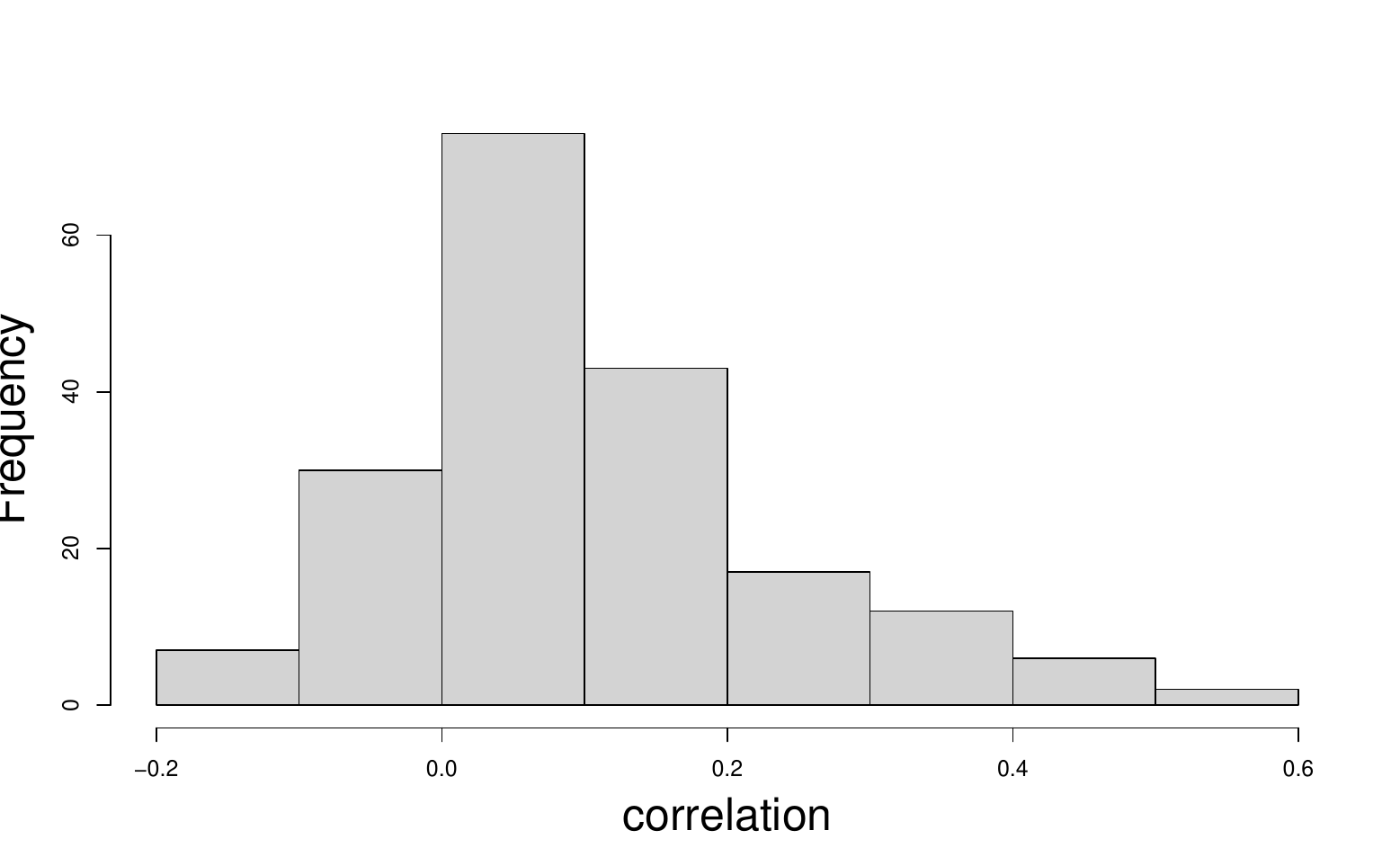}
    \caption{Sample correlations between the 20 selected regions of the live rat.}
    \label{fig: 20 regions alive}
\end{figure}
We rescale the raw data, because the original sample variances are very small ($< 0.1$), which can not reach variances $\mathbb{E}(\bSigma_{kk'})$ produced by $\siw_1$. All $20$ components are rescaled so that they have the same sample variance, namely, $\bm S_{kk} \equiv \bm S_{k'k'}$. The new data sample variance is chosen such that $\sigma_{1k} \equiv 5$. For the dead rat, we set $\rho_0 = \rho_1 = 0$, and test different values of $\nu_1$ and $\nu_2$. For the live rat, we test different values of $\rho_0$ and $\rho_1$ such that 
$\mathbb{E}(\bR_{kk'})
\approx 
0.5\rho_0 + 0.5\times 0.066\rho_1
>0$, with a range of values of $\nu_0$ and $\nu_1$. The details on the posterior sampling are in Appendix \ref{app: post inf}. 


The results of the mixture model are presented in Figure~\ref{fig:iwsiw_mean_both}.
For the dead rat, we considered four values for the $\iw$ degrees of freedom,
$\nu_0\in\{22,32,42,52\}$, combined with four values for the $\siw_1$ degrees of freedom,
$\nu_1\in\{10,15,18,36\}$. For all the resulting $16$ mixture models, the posterior weight
assigned to the $\iw$ component is equal to $0$, meaning that the posterior distribution
coincides with the $\siw_1$ posterior. Since all these models lead to the same qualitative
pattern, Figure~\ref{fig:iwsiw_mean_both} only displays the heat map corresponding to
$\nu_0=52$ and $\nu_1=10$. We see that, thanks to the $\siw_1$ component, the posterior mean of the mixture provides an ``perfect" estimation of the dead rat connectivity, which is expected to be absent. We also compare the tested
models using the predictive criterion of \cite{vehtari2017practical}. Since the posterior collapses onto the $\siw_1$ component in all cases, the predictive fit is essentially
driven by $\nu_1$. It increases with $\nu_1$, with the best values obtained for
$\nu_1=36$.

For the living rat, Figure~\ref{fig:iwsiw_mean_both} reports one representative
hyperparameter setting, namely
\[
\rho_1=0.3,\qquad \nu_1=40,\qquad \rho_0=0.5,\qquad \nu_0=54.
\]
Similar qualitative patterns are obtained for other choices of $\rho_r$ and $\nu_r$,
$r=0,1$. The corresponding results are reported in Appendix~\ref{sec: alive}. Figure~\ref{fig:iwsiw_mean_both} shows that the posterior mean of the mixture is a compromise between the posterior means obtained when each component prior is fitted separately under a Gaussian likelihood. In particular, the $\iw$ component improves the performance of the $\siw_1$ model by mitigating the excessive shrinkage induced by $\siw_1$ alone.

Together, the two rat examples illustrate that the mixture model broadens the range of data scenarios in which the two component priors can be effectively applied.

\begin{figure}[htbp]
    \centering

    \begin{subfigure}[t]{\textwidth}
        \centering
        \hskip-0.4in$\iw/\siw_1$ \hspace{0.8in} $\iw$ \hspace{1.2in} $\siw_1$

        \includegraphics[width=0.3\linewidth]{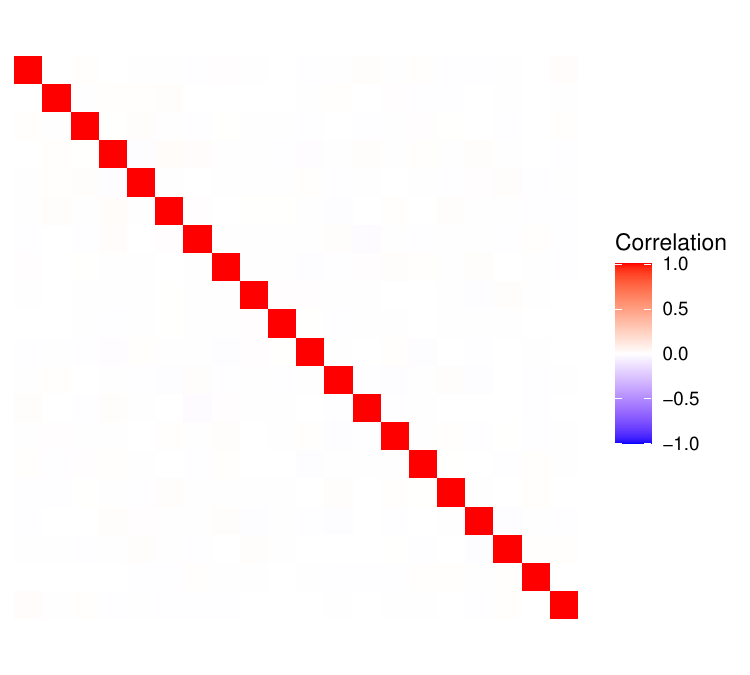}\hspace{-0.3in}
        \includegraphics[width=0.3\linewidth]{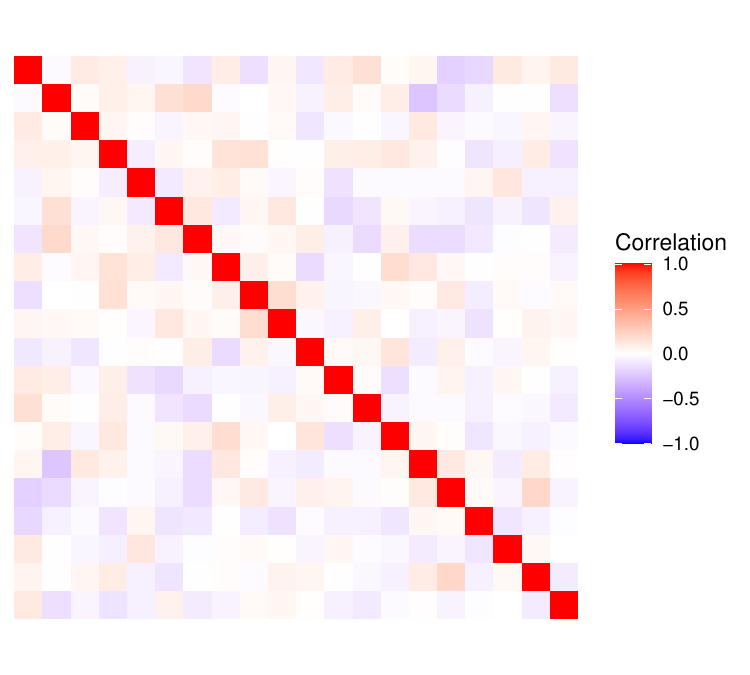}\hspace{-0.3in}
        \includegraphics[width=0.3\linewidth]{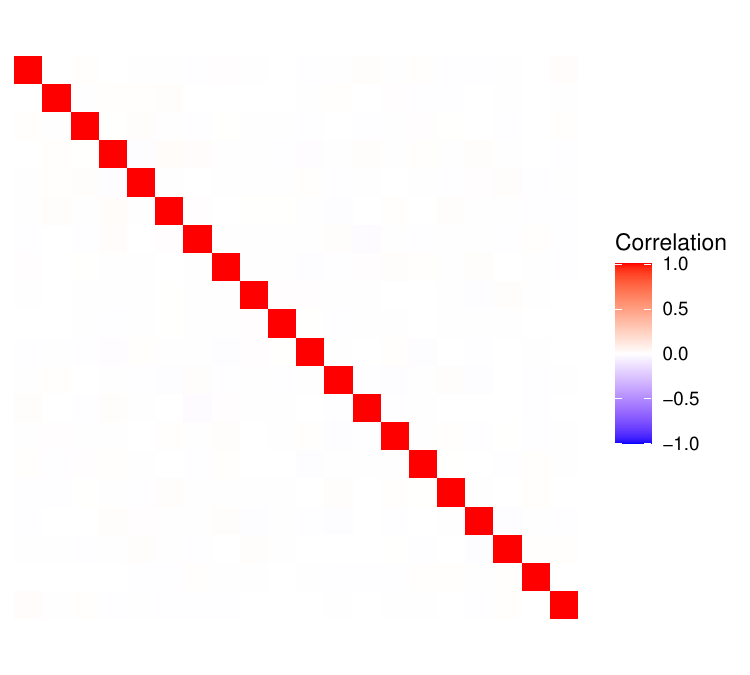}

        \caption{\textit{Dead rat}. Left: mixture model defined by the $\iw/\siw_1$ prior with $\rho_1 = 0, \nu_1 = 10, \rho_0 = 0, \nu_0 = 52$. Middle: individual model defined by the $\iw$ prior with $\rho_0 = 0, \nu_0 = 52$. Right: individual model defined by the $\siw_1$ prior with $\rho_1 = 0, \nu_1 = 10$. The posterior mixture weight is $\eta_n = 0$. }
        \label{fig:iwsiw_mean_dead}
    \end{subfigure}

    \vspace{0.6em}

    \begin{subfigure}[t]{\textwidth}
        \centering
        \hskip-0.4in$\iw/\siw_1$ \hspace{0.8in} $\iw$ \hspace{1.2in} $\siw_1$

        \includegraphics[width=0.3\linewidth]{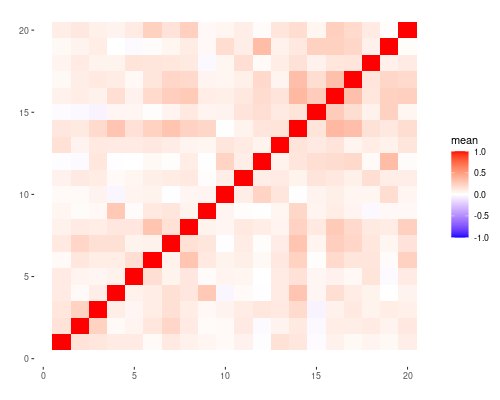}\hspace{-0.3in}
        \includegraphics[width=0.3\linewidth]{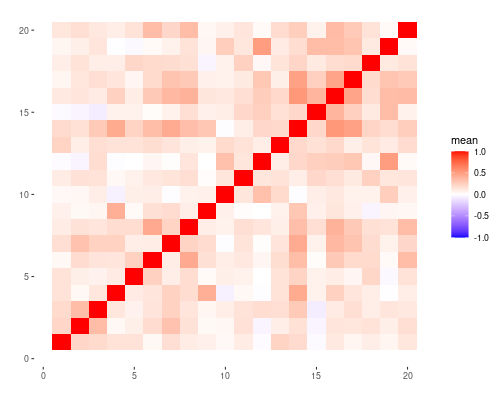}\hspace{-0.3in}
        \includegraphics[width=0.3\linewidth]{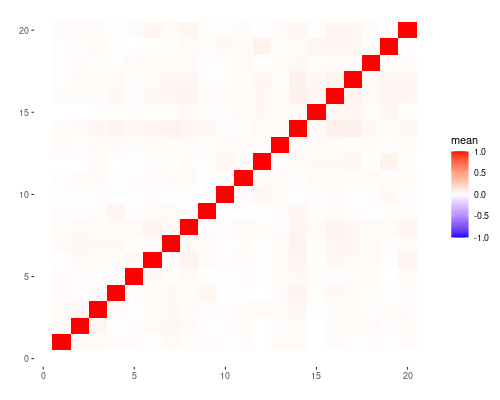}

        \caption{\textit{Live rat}. Left: mixture model defined by the $\iw/\siw_1$ prior with $\rho_1 = 0.3, \nu_1 = 40, \rho_0 = 0.5, \nu_0 = 54$. Middle: individual model defined by the $\iw$ prior with $\rho_0 = 0.5, \nu_0 = 54$. Right: individual model defined by the $\siw_1$ prior with $\rho_1 = 0.3, \nu_1 = 40$. The posterior mixture weight is $\eta_n = 0.64$. }
        \label{fig:iwsiw_mean_alive}
    \end{subfigure}

    \caption{Posterior mean correlation matrices for dead and live rats.}
    \label{fig:iwsiw_mean_both}
\end{figure}

\subsection{Detection of regional pairs of significant connectivity}\label{sec: detection}
The previous results reflect how Bayesian inference can help incorporate expert knowledge when available. However, the induced point estimation does fully show the advantages of posterior distributions, we also need to consider uncertainty evaluation and the richer post-fitting inference. We propose to explore the posterior distributions with the inference of significant connectivity, based on the notion of credible set. The procedures for $\iw$ and the mixture prior differ slightly. We first present the procedure for the $\iw$ model as follows and its results. 

\bigskip
\noindent
\textbf{Edge detection.}
\begin{enumerate}
    \item \textit{Retrieve the posterior distribution of $\bR_{kk'}$.} We obtain first samples of covariance matrices from an $\iw$ posterior. Thea correlation sample from each covariancea correlation sample. The empirical distribution of correlation samples represent the posterior marginal distribution of $\bR_{kk'}$.
    \item \textit{Determining based on credible sets.} We construct the shortest credible set of probability 0.9 from each posterior distribution of $\bR_{kk'}$, then check if $0$ falls in the set. If no, a ``significant" edge is detected. An illustration of this step is given in Figure \ref{fig: detection}.
\begin{figure}[ht]
    \centering    \includegraphics[width=\linewidth]{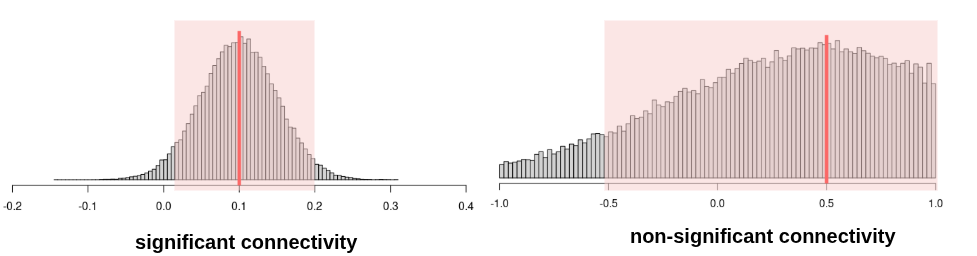}
    \caption{Posterior distributions of two $\bR_{kk'}$ are shown. By taking into account the uncertainty that is variance, Bayesian inference gives a more robust result than thresholding point estimates. }
    \label{fig: detection}
\end{figure}
\end{enumerate}
 
Figure \ref{fig: competitor} shows the results of edge detection from the $\iw$ model. We also show the detection results from the competitor \cite{acharyyaBayesianHierarchicalModeling2023}. Since their observations/samples are covariance matrices, we use the posterior predictive distributions of the induced correlation matrices. The covariance samples are calculated using the raw time series. Especially, since their model is only able to deal with small dimensions in real data context, we reduced the problem dimension to the one considered in \cite{acharyyaBayesianHierarchicalModeling2023}, which are 7 regions.
We use the same way to elicit hyperparameters as in the paper. 
\begin{figure}
    \centering
\includegraphics[width=0.25\textwidth]{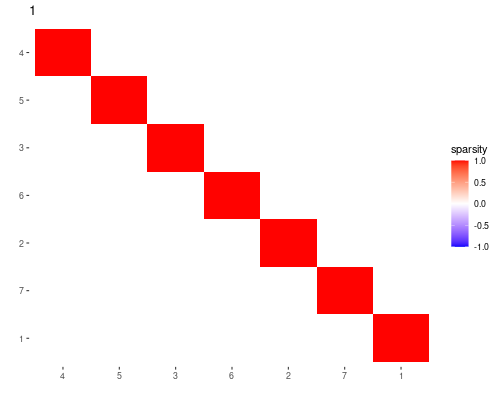}
    \hskip -0.25in
\includegraphics[width=0.25\textwidth]{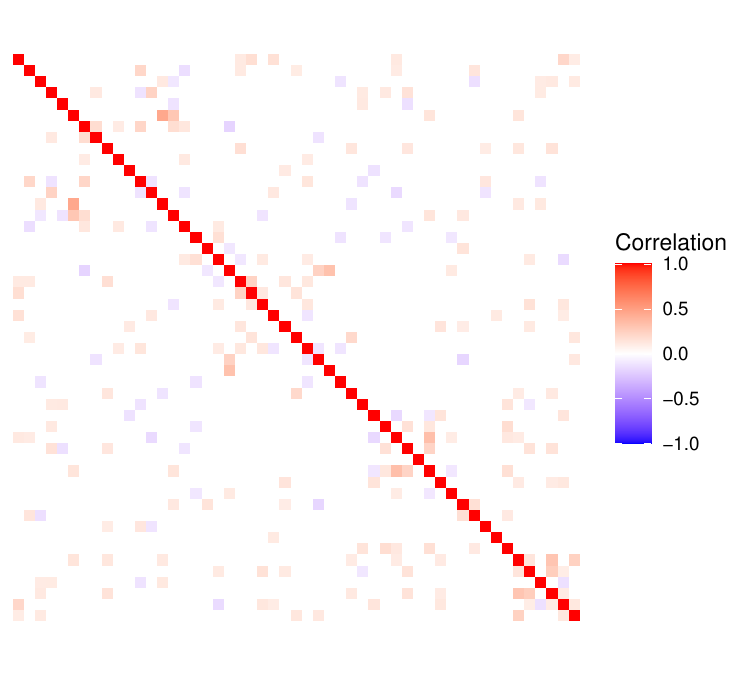}
    \hskip -0.25in
\includegraphics[width=0.25\textwidth]{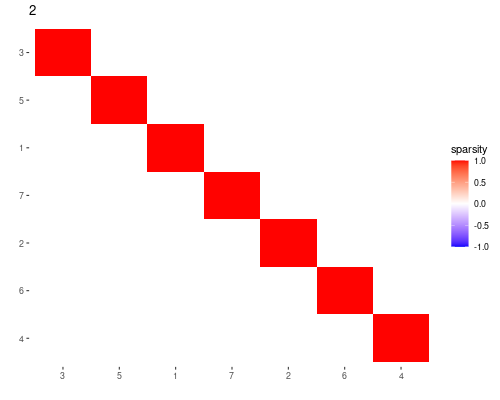}
    \hskip -0.25in
\includegraphics[width=0.25\textwidth]{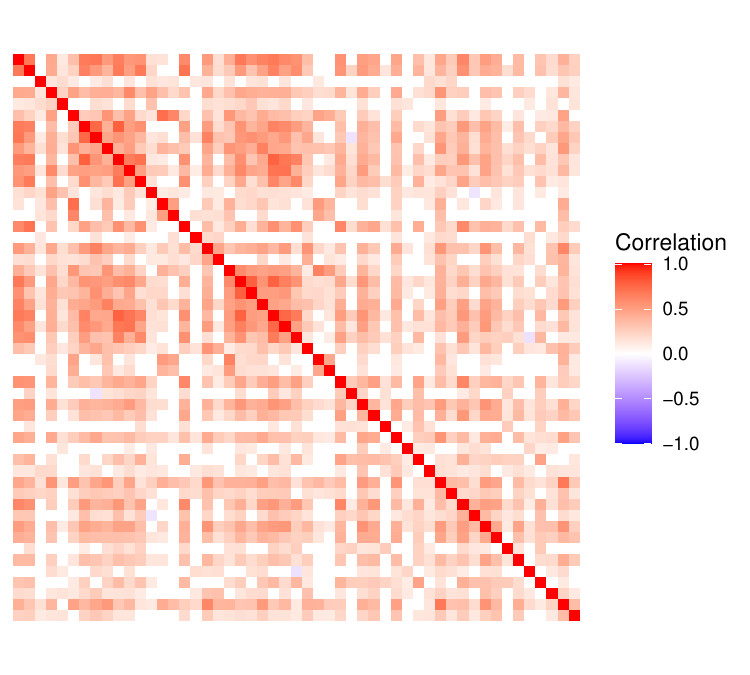}\\
\scriptsize
(a) \hspace{1.3in} (b) \hspace{1.3in} (c) \hspace{1.3in} (d)
\vskip-0.1in
\caption{Posterior mode correlation matrices masked by insignificant pairs, of the dead rat (a,b) and the live rat data (c,d). (a,c) the competitor model considering only the first 7 regions of the two rats; (b) the proposed model with $\rho = 0, \nu = 29$; (d) the proposed model with $\rho = 0.3, \nu = 29$. }
 \label{fig: competitor}
\end{figure}  
For both dead and live rats, all edges are considered not significant by the competitor. This is because the variance of each edge is too large, possibly caused by bad quality of covariance samples, which are at the same time correlated, bad-conditioned, and with small sample size. By contrast, our model is able to give reasonable results.

We also performed edge detection for the $\iw/\siw_1$ model. However, the standard procedure requires modification because it relies on examining the shortest credible set (i.e., the Highest Posterior Density interval). Since the posterior distribution under the $\iw/\siw_1$ mixture prior can be severely bimodal, its shortest credible set may consist of disconnected intervals or erratically encompass a secondary mode even when its associated probability mass is exceedingly small. This mathematical instability can mislead inference, particularly when testing whether a null value of zero is excluded from the interval.

To circumvent this issue, we propose relying on the credible set of a transformed sample mean estimator of the posterior, denoted by $\bM_{kk'}(N)$ and defined as:
\begin{equation}\label{eq: scaled sample mean}
\begin{aligned}
   \bM_{kk'}(N) &=  \sqrt{N}(\Bar{\bR}_{kk',N} - \mathbb{E}(\bR_{kk',n})) + \mathbb{E}(\bR_{kk',n}), \\
   &\mbox{where } \Bar{\bR}_{kk',N} = \frac{1}{N}\sum\limits_{n=1}^N \bR_{kk',n},   \\
   &\mbox{ and } \bR_{kk',n} \stackrel{iid}{\sim}  \mbox{posterior distribution of }\bR_{kk'}, \quad n = 1, \ldots, N.
\end{aligned}
\end{equation}
This transformation elegantly leverages the Central Limit Theorem to construct a moment-matched Gaussian approximation of the marginal posterior. By the Central Limit Theorem, the distribution of the unscaled sample mean $\Bar{\bR}_{kk',N}$ converges to a Gaussian. Centering this term and scaling it by $\sqrt{N}$ exactly recovers the true overall posterior variance, $\text{Var}(\bR_{kk'})$, while preserving the exact posterior mean. Consequently, the distribution of $\bM_{kk'}(N)$ perfectly mimics the posterior's first two moments but forces a unimodal, Gaussian shape. Using this smoothed approximation yields continuous, strictly contiguous credible intervals, resulting in robust, stable edge detection.

The adapted procedure is given as follows.
\noindent
\paragraph*{Edge detection.}
\begin{enumerate}
    \item \textit{Retrieve the posterior distribution of $\bR_{kk'}$.} We obtain first samples of covariance matrices from the posterior distribution in Proposition \ref{propIWSIWpost}. Then, we derive from each covariance sample a correlation sample, denoted by $\bR_{kk',n}, \; n=1,\ldots, N$. 
    \item \textit{Deriving the distribution of $\bM_{kk'}(N)$.} 
    \begin{enumerate}
        \item Estimating $\mathbb{E}(\bR_{kk',n})$ using all $N$ samples by $\sum\limits_{n=1}^N \bR_{kk',n}/N$.
        \item Partitioning the $N$ samples in $L$ batches: $\{\bR_{kk',n_l}: n_l = (l-1)S + 1, \ldots, lS\}$, with $S = \lfloor N/L\rfloor$.
        \item \label{edge detection 2c} Calculating 
        $\bM_{kk'}(S)$ using Equation \eqref{eq: scaled sample mean} for each of the $L$ batches.
    \end{enumerate}
    \item \label{edge detection 3} \textit{Determining based on credible sets.} Construct the shortest credible set of probability 0.9 from the empirical distribution of the $L$ $\bM_{kk'}(S)$, then check if $0$ falls in the set. If no, a ``significant" edge is detected.
\end{enumerate}
\begin{figure}[!htbp]
    \centering

    \begin{subfigure}[t]{\textwidth}
        \centering
        \hskip-0.4in$\iw/\siw_1$ \hspace{0.8in} $\iw$ \hspace{1.2in} $\siw_1$

        \includegraphics[width=0.3\linewidth]{images/post_corr_mean_deadrat_rho0_0_rho1_0_nu0_52_nu1_10.pdf}\hspace{-0.3in}
        \includegraphics[width=0.3\linewidth]{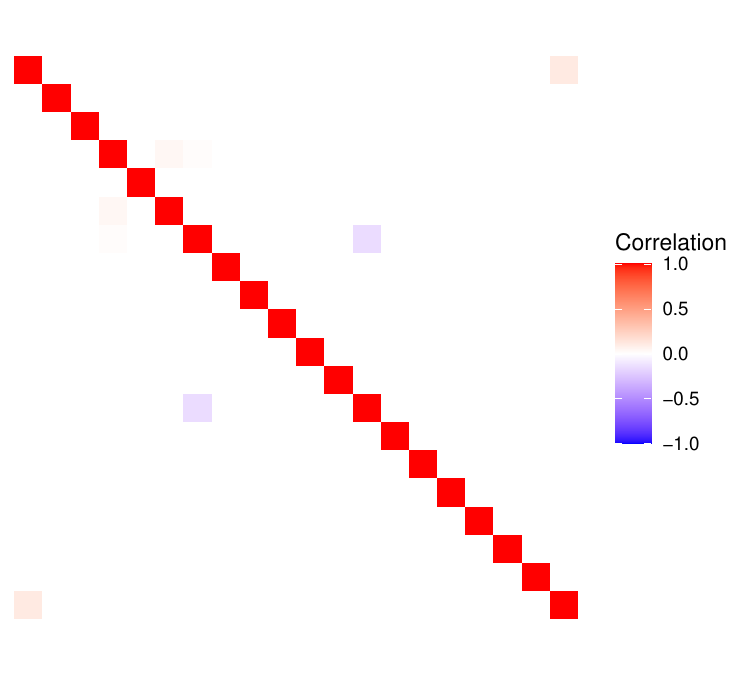}\hspace{-0.3in}
        \includegraphics[width=0.3\linewidth]{images/post_corr_mean_deadrat_rho0_0_rho1_0_nu0_52_nu1_10.pdf}

        \caption{\textit{Dead rat.} Left: mixture model defined by the $\iw/\siw_1$ prior with $\rho_1 = 0, \nu_1 = 10, \rho_0 = 0, \nu_0 = 50$. Middle: individual model defined by the $\iw$ prior with $\rho_0 = -0.05, \nu_0 = 174$. Right: individual model defined by the $\siw_1$ prior with $\rho_1 = 0.9, \nu_1 = 18$.}
        \label{fig:dead_rat}
    \end{subfigure}

    \vspace{0.6em}

    \begin{subfigure}[t]{\textwidth}
        \centering
        \hskip-0.4in$\iw/\siw_1$ \hspace{0.8in} $\iw$ \hspace{1.2in} $\siw_1$

        \includegraphics[width=0.3\linewidth]{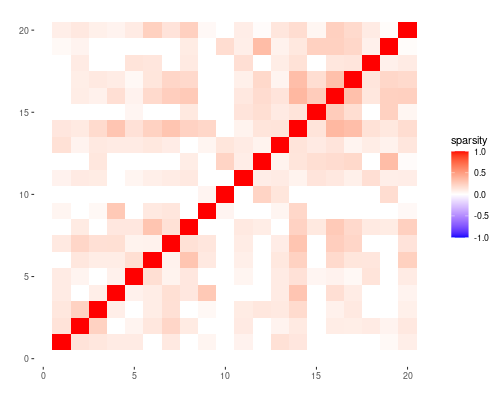}\hspace{-0.3in}
        \includegraphics[width=0.3\linewidth]{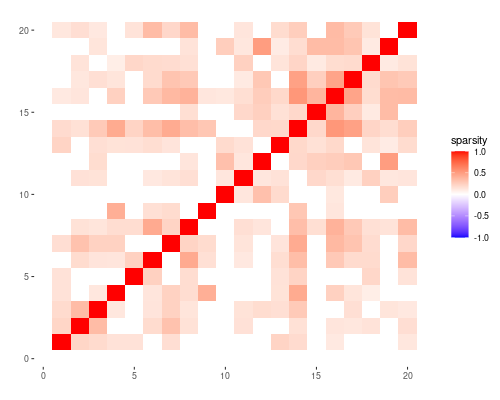}\hspace{-0.3in}
        \includegraphics[width=0.3\linewidth]{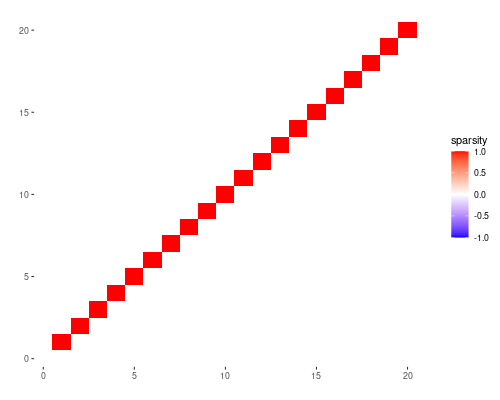}

        \caption{\textit{Alive rat}. Left: mixture model defined by the $\iw/\siw_1$ prior with $\rho_1 = 0.3, \nu_1 = 40, \rho_0 = 0.5, \nu_0 = 54$. Middle: individual model defined by the $\iw$ prior with $\rho_0 = 0.5, \nu_0 = 54$. Right: individual model defined by the $\siw_1$ prior with $\rho_1 = 0.3, \nu_1 = 40$. The posterior mixture weight is $\eta_n = 0.64$. We did not reorder the regions as in Figure \ref{fig: rats} to facilitate the comparison across results.}
        \label{fig:alive_rat}
    \end{subfigure}

    \caption{Edge detection results for dead and live rats (top/bottom).}
    \label{fig:rats_edge_detection}
\end{figure}

The corresponding algorithm for this procedure is given in Appendix \ref{app: detection mix}. The detection results of both rats are given in Figure \ref{fig:rats_edge_detection}. Subfigure \ref{fig:dead_rat} shows that for the dead rat, the result from the mixture prior is as good as the one from the $\siw_1$ component, and outperforms the $\iw$ prior. For the live rat, in general $\iw$ model and the mixture model give similar results, such as the isolation of region 10 corresponding to right hippocampus. Meanwhile, as the mean correlations shrink, as expected, some significant links in the $\iw$ models disappeared such as $7-11$ (Left secondary motor cortex - Right periaqueductal gray), $6-11$ (Left entorhinal cortex - Right periaqueductal gray). However the change in the detection result is not always linear with respect to the change in the mean. For example, the mixture deletes some large correlation values in the $\iw$ model for example $1-20$ (Right amygdalo-piriform area - Left paratenial thalamic nucleus), $3-6$ (Right temporal association cortex - Left entorhinal cortex), and $3-8$ (Right temporal association cortex - Left temporal association cortex). These correlations correspond to larger posterior variances in the mixture model than in the $\iw$ model, so that their $0.9$ credible interval covers $0$. By contrast some small correlations are considered significant by the mixture whereas they do not appear in the result of $\iw$ model, for example $1-11$ (Right amygdalo-piriform area - Right periaqueductal gray), $1-18$ (Right amygdalo-piriform area - Left periaqueductal gray), and $1-19$ (Right amygdalo-piriform area - Left basal forebrain). Some of these changes proposed by the mixture model are consistent with the neuroscientific knowledge. 
Therefore, the mixture model is not a simple weighted average of the component models, which has its own advantage and offers another choice of prior. 

To assess the robustness of the mixture model, it is instructive to examine the posterior variance of the correlation under a standard Inverse-Wishart prior. Recall the approximation for the posterior variance of the Inverse Wishart component:
\begin{equation}
    \mathbb{V}_0(\bR_{kk'} \mid \x_{1:n}) \approx \frac{(1- (\mathbf{P}_{0, n,kk'})^2)^2}{\nu_{0,n} - K + 1}, \quad \text{where } \nu_{0,n} = \nu_0 + n.
\end{equation}
In our plug-in empirical Bayes setting, the updated mean parameter is formed as a convex combination of the prior and the sample correlation:
\begin{equation}
 \mathbf{P}_{0,n,kk'} 
 = \frac{n_0}{n_0+n}\P_{0,kk'} + \frac{n}{n_0+n} \frac{\bm S_{kk'}}{\sqrt{\bm S_{kk}}\sqrt{\bm S_{k'k'}}}, \quad \text{with } n_0 = \nu_0 - K - 1.
\end{equation}
This structural form reveals a critical vulnerability. As the sample correlation grows exceptionally large (approaching $\pm 1$), it drags the posterior mean $\mathbf{P}_{0,n,kk'}$ toward the extremes. Consequently, the numerator of the variance, $(1- (\mathbf{P}_{0, n,kk'})^2)^2$, rapidly collapses toward zero. This phenomenon is particularly problematic in high-dimensional or low-sample-size regimes, where outliers or severe collinearity can produce spuriously high sample correlations. In such scenarios, the standard Inverse Wishart model suffers from a double penalty: it overestimates the strength of the correlation while simultaneously reporting an artificially near-zero posterior variance. This overconfidence easily results in the detection of false-positive edges.

By contrast, the mixture model fundamentally mitigates this risk by incorporating structural uncertainty. By the law of total variance, the marginal posterior variance under the mixture naturally decomposes into the weighted average of the component variances plus the variance between the component means (see Property \ref{prop: iw siw post corr}):
\begin{equation}
\mathbb{V}(\bR_{kk'} \mid \x_{1:n} ) \approx \eta_n \mathbb{V}_0 + (1-\eta_n)\Big(0.09 + \exp(-0.23\nu_{1,n} - 1.55)\Big) + \eta_n(1-\eta_n)(\mathbf{P}_{0,n,kk'})^{2}.
\end{equation}
The crucial element here is the strictly positive between-component variance term, $\eta_n(1-\eta_n)(\mathbf{P}_{0,n,kk'})^{2}$, which arises because the expected prior correlation for the Shrinkage Inverse Wishart component is centered at zero (see subsection \ref{app: mean var mix}). If an outlier artificially inflates the sample correlation and drives $\mathbf{P}_{0,n,kk'}$ toward $\pm 1$, the standard Inverse Wishart variance ($\mathbb{V}_1$) still shrinks, but the between-component variance heavily penalizes this extreme shift by growing proportionally to the square of the mean. Consequently, this inflation term acts as an intrinsic safety mechanism, artificially widening the credible intervals when the data conflict with the mixture components. Thus, the mixture model does not act merely as a simplistic weighted average of two priors; it introduces a robust regularization effect that actively suppresses overconfident false discoveries.


\begin{remark}
In general, the shrinkage effect on correlation values induced by $\siw_1$ is particularly beneficial when the sample correlations exhibit abnormally extreme values close to $\pm 1$. This may occur in the presence of outliers and/or in high-dimensional regimes such as $n \approx K$. In such cases, shrinking correlation values helps prevent extreme estimates and typically yields more robust inference. This robustness brought by shrinking the eigengap of the covariance matrix, hence shrinking both covariance and induced correlation entries, shares the same spirit as classical frequentist regularization methods, such as ridge-type corrections and the Ledoit--Wolf estimator, which replace $\bm S$ by $\bm S+\lambda I_K$. Compared with these methods, which shrink matrix entries in a linear fashion, the density reweighting in the mixture is more sophisticated and can bring additional advantages, such as a more robust posterior variance as illustrated above. 
\end{remark}
\begin{remark}
In settings where data quality is limited (e.g., due to outliers and/or small sample size), one may wish to enforce the shrinkage effect of $\siw_1$. In this case, updating mixture weight by the Bayes rule can be inconvenient, since it can drive the weight toward $1$ and effectively eliminate the $\siw_1$ component. One solution is to consider a fixed-weight mixture posterior (i.e., without updating $\eta$), which can be formalized within a generalized Bayes framework. We leave further discussion to the perspectives.
\end{remark}


\section{Conclusion and perspectives}
In this work, we introduce new Bayesian models to infer a single subject’s functional connectivity graph from regional fMRI time series. These models provide a posterior distribution for each pair of regions, thereby quantifying the uncertainty associated with their functional connectivity. Using these posterior distributions, we derived point estimates of the connectivities, yielding a fully connected undirected weighted graph. We also proposed a procedure to identify significant connections, resulting in a sparse graph representation. Thanks to the explicit links between the prior hyperparameters and the resulting mean and variance, the proposed priors offer greater control over the shape and location of the distributions. This facilitates the incorporation of domain-specific expert knowledge. Furthermore, we investigated how the prior hyperparameters affect the posterior mean and variance, and numerical experiments on real datasets validated the theoretical findings. These results illustrate three key benefits of Bayesian modeling: uncertainty quantification, richer post-fitting inference, and flexible integration of expert information. Beyond knowledge elicitation, the prior also plays an important regularizing role as remarked previously. This makes Bayesian estimation naturally robust when the data quality is limited such as in high-dimensional settings and/or outliers are present. 

From an application perspective, most existing Bayesian models target group-level analyses. The proposed model therefore enriches the methodological landscape for single-subject analysis. Within the specific subcategory of static functional connectivity network inference from resting-state fMRI for a single subject, our approach outperforms the only existing method \citep{acharyyaBayesianHierarchicalModeling2023}, both in terms of estimation quality and the dimensionality of the problems it can handle.

From a statistical viewpoint, we introduced two new priors for correlation matrices, based on the $\iw$ distribution and an $\iw$–$\siw_1$ mixture. These priors allow explicit control over the marginal prior distributions of individual correlation entries, while remaining highly computationally convenient when coupled with a Gaussian likelihood.
These methodologies can extend beyond the domain of functional connectivity, since correlation is widely used in other fields to analyze dependence, in particular in combination with thresholding to construct a graph, such as in portfolio allocation \citep{pafka2004estimated}, financial networks \citep{emmert2010influence}, genetic networks \citep{butte2000discovering}, and climate networks \citep{oliveira2023networks}. Our inference models, together with the edge detection scheme, address the need for inference on correlation-based networks while allowing an easy integration of prior knowledge.
Among these applications, some domains often face difficulties related to high dimensionality, as in genetic networks, or to the presence of outliers, as in portfolio allocation. In such settings, the robustness of Bayesian methods induced by the regularization effect of the prior becomes a major advantage. Overall, the proposed methods offer a versatile tool for inferring both correlation structures and the networks derived from them, with potentially wide applicability.

Several directions for future work are of interest. A first one is to further investigate the \emph{fixed-weight mixture posterior}. As discussed at the end of Section~\ref{sec: detection}, this approach prevents the $\siw_1$ component from vanishing under the standard Bayes rule. This feature is particularly desirable in low-quality data regimes (e.g., outliers and/or $n\approx K$), where maintaining a non-negligible $\siw_1$ contribution can shrink abnormally large sample correlations and thus yield more robust inference than a pure $\iw$ model. More specifically, the fixed-weight mixture posterior can be defined through a generalized Bayes rule\footnote{Recall that the classical Bayesian posterior can be characterized as the minimizer of
\begin{equation}\label{eq:J_classical}
q^\star = \argmin_{q \in \mathcal{P}}\E_{q \sim \bSigma}[-\log p( X\mid \bSigma)]\;+\;\mathcal{KL}(q\|\pi),
\end{equation}
where $\mathcal P$ denotes the set of densities on $\mathbb S_{++}^K$, $\pi$ is the prior and $p(X\mid \bSigma)$ is the likelihood. The data term $\E_{q \sim \bSigma}[-\log p( X\mid \bSigma)]$ drives the posterior update.}:
\begin{equation}\label{eq:J_q_only}
q^\star = \argmin_{q \in \mathcal Q_\eta}\E_{q \sim \bSigma}[-\log p(X\mid \bSigma)]\;+\;\mathcal{KL}_\eta(q\|\pi_{\iw},\pi_{\siw_1}),
\end{equation}
where
\begin{equation}\label{eq:mixture_KL_def}
\mathcal{KL}_\eta(q\|\pi_{\iw},\pi_{\siw_1})
\;=\;
\eta\,\mathcal{KL}(q_0\|\pi_{\iw})+(1-\eta)\,\mathcal{KL}(q_1\|\pi_{\siw_1}),
\end{equation}
$\pi_{\iw},\pi_{\siw_1}$ are respectively the densities of $\mathcal{IW}(\P_0, \Vec{\sigma}_0, \nu_0)$ and $\siw_1(\P_1, \Vec{\sigma}_1, \nu_1)$, and
\begin{equation}\label{eq: Q_eta}
\mathcal Q_\eta \;:=\;\Big\{q:\ q(\Sigma)=\eta q_0(\Sigma)+(1-\eta)q_1(\Sigma),\ \ q_0,q_1\in\mathcal P\Big\}.
\end{equation}
The optimization problem constrains the posterior distribution to have a form of mixture with fixed weight $\eta$, and assumes a priori via the $\mathcal{KL}$ term that the two components are $\iw$ and $\siw_1$. One can show that the resulting generalized posterior is an $\iw/\siw_1$ mixture with fixed weight $\eta$ and component-wise classical posteriors:
\[
q^\star(\Sigma)\;=\;\eta\,\pi_{\iw}(\Sigma\mid X)\;+\;(1-\eta)\,\pi_{\siw_1}(\Sigma\mid X),
\]
where $\pi_{\iw}(\Sigma\mid X)$ and $\pi_{\siw_1}(\Sigma\mid X)$ are respectively the densities of $\mathcal{IW}(\bm\P_{0,n}, \Vec{\bm \sigma}_{0,n}, \nu_{0,n})$ and $\siw_1(\bm\P_{1,n}, \Vec{\bm \sigma}_{1,n}, \nu_{1,n})$ defined in Proposition~\ref{propIWSIWpost}.%
 Such generalized posterior allows to reuse the results on $\iw$ and $\siw_1$ developed in this paper, and can be potentially a promising direction.

Secondly, as noted above, the subcategory of single-subject static functional connectivity network inference remains largely unexplored in Bayesian frameworks and warrants further attention. More realistic data scenarios could also be considered, such as learning directly from voxel-level time series rather than aggregated regional signals. Finally, comparing two single subjects and identifying the subregions that contribute most to their differences is another promising topic within this subcategory.

\bibliographystyle{agsm}
\bibliography{main}

\newpage

\appendix
\section{Proofs}
\subsection{Proposition \ref{prop: IW prior}}\label{app: proof mean var iw}
We prove a more generic result indicated by Remark \ref{rem: mean var iw}, whose complete statement is recalled below. 

\noindent
\textbf{Proposition 1} (Extended version)\label{prop: IW prior generic}
\textit{Let $\bSigma \sim \mathcal{IW}(\P, \Vec{\sigma},\nu)$, and denote $\bR = R(\bSigma)$.
Then we have for $k\neq k'$ and as $\nu - K \rightarrow \infty$:
 \begin{enumerate}
 \vspace{-0.05in}
 \item $\bR \rightarrow \P$ in probability,
    \item $ \mathbb{E}(\bR_{kk'}) =
    \P_{kk'} \left(1-\frac{1-(\P_{kk'})^2}{2(\nu - K+2)}\right) +  o\left(\frac{1}{(\nu-K)^2}\right),$ 
     \vspace{-0.05in}
    \item $\mathbb{V}
    (\bR_{kk'}) = \frac{(1-(\P_{kk'})^2)^2}{\nu - K+1} + O\left(\frac{1}{(\nu-K)^3}\right).$ 
\end{enumerate}   
Additionally, we have the exchangeability on $(\bR_{kk'})_{k < k'}$, given $\P_{kk'} \equiv \rho, \forall k \neq k'$:
\begin{equation}
Q\bR Q^{-1} \stackrel{d}{=} \bR,
\end{equation}
for any permutation matrix $Q$.}

Proposition \ref{prop: IW prior} is the special case of Proposition 1 (Extended version) with $K$ fixed.

\begin{proof}
\noindent\textbf{Proof of point (1).}

Let 
\[
\bSigma \sim \mathcal{IW}(\Psi,\nu), \qquad \nu > K-1,
\]
and write both $\bSigma$ and $\Psi$ in block form as
\[
\bSigma =
\begin{pmatrix}
\bSigma_{aa} & \bSigma_{ab} \\
\bSigma_{ba} & \bSigma_{bb}
\end{pmatrix},
\qquad
\Psi =
\begin{pmatrix}
\Psi_{aa} & \Psi_{ab} \\
\Psi_{ba} & \Psi_{bb}
\end{pmatrix},
\]
with the blocks conformable in size.  
A standard marginal property of the inverse–Wishart distribution \citep{gupta2018matrix} states that
\[
\bSigma_{aa} \;\sim\; 
\mathcal{IW}_{p_a}\!\left( \Psi_{aa} ,\; \nu - p_b \right),
\qquad p_b = \mathrm{dim}(\bSigma_{bb}).
\]

We set $p_a=2$ and extract the $2\times 2$ submatrix
\[
\bSigma_{aa} =
\begin{pmatrix}
\bSigma_{11} & \bSigma_{12} \\
\bSigma_{21} & \bSigma_{22}
\end{pmatrix},
\]
whose induced correlation is exactly 
\(
\bR_{12} = \bSigma_{12}/\sqrt{\bSigma_{11}\bSigma_{22}}.
\)

By the above marginal property,
\[
\begin{pmatrix}
\bSigma_{11} & \bSigma_{12} \\
\bSigma_{21} & \bSigma_{22}
\end{pmatrix}
\sim 
\mathcal{IW}_{2}\!\left(
\begin{pmatrix}
\Psi_{11} & \Psi_{12} \\
\Psi_{21} & \Psi_{22}
\end{pmatrix},
\; \nu-(K-2)
\right).
\]

Using the standard sample representation of the inverse–Wishart, we have
\[
\begin{pmatrix}
\bSigma_{11} & \bSigma_{12} \\
\bSigma_{21} & \bSigma_{22}
\end{pmatrix}
\;\overset{d}{=}\;
\left(\sum_{i=1}^{\nu-(K-2)}
\begin{pmatrix}
\bm x_i^1 \\ \bm x_i^2
\end{pmatrix}
\begin{pmatrix}
\bm x_i^1 & \bm x_i^2
\end{pmatrix}\right)^{-1},
\]

where 
\begin{equation}\label{eq: 0}
\begin{pmatrix}
\bm x_i^1 \\ \bm x_i^2
\end{pmatrix}
\overset{\text{iid}}{\sim}
\mathcal{N}\!\left(
0,\;
\begin{pmatrix}
\Psi_{11} & \Psi_{12} \\
\Psi_{21} & \Psi_{22}
\end{pmatrix}^{-1}
\right) = \mathcal{N}\!\left(
0,\;
C_\Psi 
\begin{pmatrix}
\Psi_{22} & -\Psi_{12} \\
-\Psi_{21} & \Psi_{11}
\end{pmatrix}
\right).    
\end{equation}

On the other hand, we have 
\begin{equation}
\begin{aligned}
    \left(\sum_{i=1}^{\nu-(K-2)}
\begin{pmatrix}
\bm x_i^1 \\ \bm x_i^2
\end{pmatrix}
\begin{pmatrix}
\bm x_i^1 & \bm x_i^2
\end{pmatrix}\right)^{-1} &= 
\begin{pmatrix}
\sum\limits_{i=1}^{\nu-(K-2)} (\bm x_i^1)^2 & \sum\limits_{i=1}^{\nu-(K-2)}\bm x_i^1\bm x_i^2 \\
\sum\limits_{i=1}^{\nu-(K-2)}\bm x_i^2\bm x_i^1 & \sum\limits_{i=1}^{\nu-(K-2)} (\bm x_i^2)^2
\end{pmatrix}^{-1}    \\
&= C_x \begin{pmatrix}
\sum\limits_{i=1}^{\nu-(K-2)} (\bm x_i^2)^2 & -\sum\limits_{i=1}^{\nu-(K-2)}\bm x_i^2\bm x_i^1 \\
-\sum\limits_{i=1}^{\nu-(K-2)}\bm x_i^1\bm x_i^2 & \sum\limits_{i=1}^{\nu-(K-2)} (\bm x_i^1)^2
\end{pmatrix}.
\end{aligned}
\end{equation}

Here, $C_\Psi$ and $C_x$ are positive constants equal to the reciprocals of the respective matrix determinants. Notice that these constants cancel out when computing the correlation. Thus, 
\begin{equation}
    \bR_{12} = \bSigma_{12}/\sqrt{\bSigma_{11}\bSigma_{22}} \;\overset{d}{=}\; -\sum\limits_{i=1}^{\nu-(K-2)}\bm x_i^1\bm x_i^2/\sqrt{\sum\limits_{i=1}^{\nu-(K-2)} (\bm x_i^2)^2 \sum\limits_{i=1}^{\nu-(K-2)} (\bm x_i^1)^2 } \;\overset{d}{=}\; -\,\widehat{\mathrm{corr}}(\bm x_i^1, \bm x_i^2), 
\end{equation}
where $-\widehat{\mathrm{corr}}(\bm x_i^1, \bm x_i^2)$ is the sample correlation estimator of $\mathrm{corr}(\bm x_i^1, \bm x_i^2)$. From Equation \ref{eq: 0}, we have 
\[
\mathrm{corr}(\bm x_i^1,\bm x_i^2)
= -\,\frac{\Psi_{12}}{\sqrt{\Psi_{11}\Psi_{22}}}
= -\P_{12}. 
\]
Given the consistency of sample correlation, we have 
\[
\widehat{\mathrm{corr}}(\bm x_i^1,\bm x_i^2)
\;\xrightarrow{\;P\;}\;
-\P_{12}
\qquad\text{as }\nu-(K-2)\to\infty.
\]
Therefore,
\[
\bR_{12} 
\;\xrightarrow{\;P\;}\;
\P_{12}.
\]

Finally, relabelling indices via permutation matrices shows that for any
$k\neq k'$,
\[
\bR_{kk'} \xrightarrow{\;P\;} \P_{kk'} 
\qquad\text{as }\nu-K\to\infty.
\]

\medskip
\noindent\textbf{Proof of points (2).}
From the argument above for point~(1), for any fixed pair $(k,k')$,
\begin{equation}
    \bR_{kk^\prime} \;\overset{d}{=}\; -\sum\limits_{i=1}^{n}\bm x_i^k\bm x_i^{k^\prime}/\sqrt{\sum\limits_{i=1}^{n} (\bm x_i^{k^\prime})^2 \sum\limits_{i=1}^{n} (\bm x_i^k)^2 }, 
\end{equation}
with
\begin{equation}\label{eq: 1}
\begin{pmatrix}
\bm x_i^k \\ \bm x_i^{k^\prime}
\end{pmatrix}
\overset{\text{iid}}{\sim}
\mathcal{N}\!\left(
0,\;
constant 
\begin{pmatrix}
\Psi_{{k^\prime}{k^\prime}} & -\Psi_{k{k^\prime}} \\
-\Psi_{{k^\prime}k} & \Psi_{kk}
\end{pmatrix}
\right),    
\end{equation}
of size
\[
n = \nu - K + 2.
\]
Thus 
\begin{equation}
    \bR_{k{k^\prime}} \;\overset{d}{=}\; -\sum\limits_{i=1}^{n}\bm x_i\bm y_i/\sqrt{\sum\limits_{i=1}^{n} (\bm x_i)^2 \sum\limits_{i=1}^{n} (\bm y_i)^2 }, 
\end{equation}
with
\begin{equation}
\begin{pmatrix}
\bm x_i \\ \bm y_i
\end{pmatrix}
\overset{\text{iid}}{\sim}
\mathcal{N}\!\left(
0,\; 
\begin{pmatrix}
1 & -\P_{kk\prime} \\
-\P_{kk\prime} & 1
\end{pmatrix}
\right),
\end{equation}
because the variances $constant\Psi_{ii}, r =0, 1$ in Equation \eqref{eq: 1} are canceled out in the division. We denote 
\begin{equation}
    \bm s_{11} = \frac{1}{n}\sum\limits_{i=1}^{n}\bm x_i^2, \; \bm s_{12} = \frac{1}{n}\sum\limits_{i=1}^{n}\bm x_i\bm y_i, \; \bm s_{22} = \frac{1}{n}\sum\limits_{i=1}^{n}\bm y_i^2.
\end{equation}
To approximate the expectation 
\begin{equation}
\mathbb{E}\bR_{k{k^\prime}} = \mathbb{E}\frac{-\bm s_{12}}{\sqrt{\bm s_{11}\bm s_{22}}},
\end{equation}
we apply the Taylor expansion on the function $g(\bm\theta) = \frac{-\bm s_{12}}{\sqrt{\bm s_{11}\bm s_{22}}}$ with $\bm \theta = (\bm s_{11}, \bm s_{22}, \bm s_{12})$ at the point 
$\theta_0 = \mathbb{E}\bm \theta = (1,1,-P_{kk^\prime})$ as follows.
\begin{equation}
g(\bm\theta) = g(\theta_0) + \nabla g(\theta_0)^\top (\bm\theta - \theta_0) + \frac{1}{2}(\bm\theta - \theta_0)^\top \bm H_g(\theta_0)(\bm\theta - \theta_0) + o(\|\bm\theta\|^2). 
\end{equation}
Then 
\begin{equation}
    \mathbb{E} g(\bm\theta) = g(\theta_0) + \frac{1}{2} tr(H_g(\theta_0) \mathbb{V}\bm\theta) + o(\frac{1}{n^2}). 
\end{equation}
Through basic calculation, we have 
\begin{equation}
    H_g(\theta_0) = 
\begin{pmatrix}
\frac{3}{4}\P_{kk\prime} & \frac{1}{4}\P_{kk\prime}& \frac{1}{2} \\
\frac{1}{4}\P_{kk\prime} & \frac{3}{4}\P_{kk\prime}& \frac{1}{2} \\
\frac{1}{2} & \frac{1}{2}& 0
\end{pmatrix},
\end{equation}
and 
\begin{equation}
\mathbb{V}\bm\theta = 
\begin{pmatrix}
\frac{2}{n} & \frac{2}{n}\P_{kk\prime}^2& -\frac{2}{n}\P_{kk\prime} \\
\frac{2}{n}\P_{kk\prime}^2 & \frac{2}{n}& -\frac{2}{n}\P_{kk\prime} \\
-\frac{2}{n}\P_{kk\prime}& -\frac{2}{n}\P_{kk\prime}& \frac{1+\P_{kk\prime}^2}{n}.
\end{pmatrix}    
\end{equation}
The calculation of $\mathbb{V}\bm\theta$ involves Gaussian $4$-th moments, which can be computed using Isserlis's theorem.
Then we have 
\begin{equation}
    \mathbb{E} g(\bm\theta) =  \P_{kk\prime}  (1 - \frac{1-\P_{kk\prime}^2}{2n})+ o(\frac{1}{n^2}). 
\end{equation}
This completes the proof of Point (2).

\noindent\textbf{Proof of points (3).}
\medskip\noindent
Using again the Taylor expansion, we have 
\[
\begin{aligned}
\mathbb{V} g(\bm\theta)
&= \nabla g(\theta_0)^\top \,\mathbb{V}(\bm \theta)\, \nabla g(\theta_0)
   + \mathrm{Cov}(\nabla g(\theta_0)^\top (\bm\theta - \theta_0),\frac{1}{2}(\bm\theta - \theta_0)^\top \bm H_g(\theta_0)(\bm\theta - \theta_0)) + O\!\left(\frac1{n^4}\right)\\
&= \nabla g(\theta_0)^\top \,\mathbb{V}(\bm \theta)\, \nabla g(\theta_0)
   + \mathbb{E}(\frac{1}{2}\nabla g(\theta_0)^\top (\bm\theta - \theta_0)(\bm\theta - \theta_0)^\top \bm H_g(\theta_0)(\bm\theta - \theta_0)) + O\!\left(\frac1{n^4}\right)\\
   &= \nabla g(\theta_0)^\top \,\mathbb{V}(\bm \theta)\, \nabla g(\theta_0)
   + O\!\left(\frac1{n^3}\right).    
\end{aligned}
\]
A direct calculation yields
\[
\nabla g(\theta_0)^\top = 
-(\frac{\P_{kk\prime}}{2},
    \frac{\P_{kk\prime}}{2},
    1).
\]
Thus 
\[
\nabla g(\theta_0)^\top \,\mathbb{V}(\bm \theta)\, \nabla g(\theta_0)
= \frac{(1 - \P_{kk^\prime})^2}{n}.
\]
This proves point (3).

\textbf{Proof of the exchangeability.}
We first show the exchangeability in the case where $\sigma_k \equiv 1, \forall k = 1, ..., K.$ Let $\bR$ be the induced correlation matrix of $\bSigma$, thus for any permutation matrix $Q$,  $Q\bR Q^{-1}$ is the induced correlation matrix of $Q \bSigma Q^{-1}$. Recall the sample representation of $\iw$:  $\bSigma \sim \iw(P,\Vec{1},\nu)$ if and only if 
\begin{equation}
    \bSigma \stackrel{d}{=}  \left[\sum\limits_{i=1}^\nu
\mathbf{x}_i\mathbf{x}_i^\top\right]^{-1}, \; \mathbf{x}_i \stackrel{i.i.d.}{\sim} \mathcal{N}(0,P^{-1}).
\end{equation}
Thus 
\begin{equation}
    Q \bSigma Q^{-1} \stackrel{d}{=}  \left[\sum\limits_{i=1}^\nu
Q\mathbf{x}_i (Q\mathbf{x}_i)^\top\right]^{-1}.
\end{equation}
Note that $Q\mathbf{x}_i \stackrel{i.i.d.}{\sim} \mathcal{N}(0,QP^{-1}Q^{-1})$, and $P^{-1}$ takes the form of 
\[
\begin{cases}
&[P^{-1}]_{ii}=\frac{1+(K-2)\rho}{(1-\rho)\bigl(1+(K-1)\rho\bigr)},
\qquad i=1,\dots,K, \\
&[P^{-1}]_{ij}=-\,\frac{\rho}{(1-\rho)\bigl(1+(K-1)\rho\bigr)},
\qquad i\neq j, i=1,\dots,K, \mbox{ and } j=1,\dots,K
\end{cases}.
\]
Thus $Q \bSigma Q^{-1} \stackrel{d}{=} \bSigma$, leading to $Q \bR Q^{-1} \stackrel{d}{=} \bR$. 

In the second step, we show the exchangeability for generic values of $\sigma_k.$ We rely on the fact that for $\bSigma \sim \mathcal{IW}(\P, \Vec{\sigma},\nu)$ with any $\Vec{\sigma}$, the induced correlation $\bR$ satisfies:
\[\bR \stackrel{d}{=} \bR_0, \mbox{ where } \bR_0 \mbox{ is induced from } \bSigma \sim \mathcal{IW}(\P, \Vec{1},\nu).\]

Let $\widehat{\bm W} = \sum\limits_{i=1}^\nu
\mathbf{x}_i\mathbf{x}_i^\top, \; \mathbf{x}_i \stackrel{i.i.d.}{\sim} \mathcal{N}(0,\Delta^{-1} P^{-1} \Delta^{-1})$, with $\Delta = \operatorname{Diag}(\Vec{\sigma})$, we have 
{
\small
\begin{equation}
\begin{aligned}
    \bR &\stackrel{d}{=} 
    \begin{pmatrix}
     [\widehat{\bm W}^{-1}]_{11}&&\\
     &\ddots&\\
     &&[\widehat{\bm W}^{-1}]_{KK}
    \end{pmatrix}^{\frac{1}{2}}
    \widehat{\bm W}^{-1}
        \begin{pmatrix}
     [\widehat{\bm W}^{-1}]_{11}&&\\
     &\ddots&\\
     &&[\widehat{\bm W}^{-1}]_{KK}
    \end{pmatrix}^{\frac{1}{2}} \\
    &\stackrel{d}{=} 
    \begin{pmatrix}
     [\widehat{\bm W}^{-1}]_{11}&&\\
     &\ddots&\\
     &&[\widehat{\bm W}^{-1}]_{KK}
    \end{pmatrix}^{\frac{1}{2}} \Delta \Delta^{-1}
    \widehat{\bm W}^{-1} \Delta^{-1} \Delta 
        \begin{pmatrix}
     [\widehat{\bm W}^{-1}]_{11}&&\\
     &\ddots&\\
     &&[\widehat{\bm W}^{-1}]_{KK}
    \end{pmatrix}^{\frac{1}{2}}\\
    &\stackrel{d}{=} \Delta^{\frac{1}{2}} 
    \begin{pmatrix}
     [\widehat{\bm W}^{-1}]_{11}&&\\
     &\ddots&\\
     &&[\widehat{\bm W}^{-1}]_{KK}
    \end{pmatrix}^{\frac{1}{2}} \Delta^{\frac{1}{2}} \Delta^{-1}
    \widehat{\bm W}^{-1} \Delta^{-1} \Delta^{\frac{1}{2}} 
        \begin{pmatrix}
     [\widehat{\bm W}^{-1}]_{11}&&\\
     &\ddots&\\
     &&[\widehat{\bm W}^{-1}]_{KK}
    \end{pmatrix}^{\frac{1}{2}} \Delta^{\frac{1}{2}} \\
    &\stackrel{d}{=} 
    \begin{pmatrix}
     [\widetilde{\bm W}^{-1}]_{11}&&\\
     &\ddots&\\
     &&[\widetilde{\bm W}^{-1}]_{KK}
    \end{pmatrix}^{\frac{1}{2}}
    \widetilde{\bm W}^{-1}
        \begin{pmatrix}
     [\widetilde{\bm W}^{-1}]_{11}&&\\
     &\ddots&\\
     &&[\widetilde{\bm W}^{-1}]_{KK}
    \end{pmatrix}^{\frac{1}{2}},
\end{aligned}
\end{equation}
}
where $\widetilde{\bm W} = \Delta  \widehat{\bm W}\Delta$. Note that  $\widetilde{\bm W} = \sum\limits_{i=1}^\nu
\Tilde{\mathbf{x}}_i\Tilde{\mathbf{x}}_i^\top,$ with $\Tilde{\mathbf{x}}_i = \Delta\mathbf{x}_i \stackrel{i.i.d.}{\sim} \mathcal{N}(0,P^{-1}).$ Thus $\bR \stackrel{d}{=} \bR_0$. Given the  exchangeability of $\bR_0$, we have the exchangeability of $\bR$.
\qedhere

\end{proof}

\subsection{Theorem \ref{thmsiw}}\label{app: siw1_first_order}

\thmsiw*
\begin{proof}
Fix $k\neq k'$ and set $g(\Sigma):=R_{kk'}(\Sigma)$. Rewrite $\P$ as $\P=I_K+\varepsilon A$ where $A$ is the off-diagonal part of $\P$ up to $\varepsilon$, with
$\varepsilon = \max_{i\neq j}|\P_{ij}|$. Denote expectations under $\P(\varepsilon)=I_K+\varepsilon A$ by $\mathbb E_\varepsilon[\cdot]$ and $\mathbb E_0=\mathbb E_{\varepsilon=0}$.

The approximation of 
$\mathbb E_\varepsilon[\bR_{kk'}]$ is given by the first order Taylor expansion, in a neighborhood of $\varepsilon = 0$, that is:
\[\mathbb E_\varepsilon[\bR_{kk'}] = \E_0[\bR_{kk'}] + \frac{d}{d\varepsilon}\mathbb E_\varepsilon[\bR_{kk'}]\Big|_{\varepsilon=0}\varepsilon + O(\varepsilon^2).\]
The choice of the expansion point is followed by the fact that at $\varepsilon=0$, $\siw_1(I_K, \bm 1, \nu)$ is highly analytically favorable, with $\E_0[\bR_{kk'}] = 0$, and computable $\frac{d}{d\varepsilon}\mathbb E_\varepsilon[\bR_{kk'}]\Big|_{\varepsilon=0}$. In the following, we calculate $\frac{d}{d\varepsilon}\mathbb E_\varepsilon[\bR_{kk'}]\Big|_{\varepsilon=0}.$

Let $q_\varepsilon(\Sigma)$ be the unnormalized $\siw_1$ density:
\[
q_\varepsilon(\Sigma)
=
\exp\!\Big(-\tfrac12\mathrm{tr}(\Psi(\varepsilon)\Sigma^{-1})\Big)\,h(\Sigma),
\quad
h(\Sigma)=|\Sigma|^{-\nu}\Big(\prod_{i<j}(\lambda_i-\lambda_j)\Big)^{-1}.
\]
Then $\mathbb E_\varepsilon[g]=\int g(\Sigma)q_\varepsilon(\Sigma)d\Sigma/\int q_\varepsilon(\Sigma)d\Sigma$. Using the trick 
\[\ \partial_\varepsilon\log q_\varepsilon(\Sigma) = \ \partial_\varepsilon q_\varepsilon(\Sigma)  / q_\varepsilon(\Sigma),  \]
we have
\[
\frac{d}{d\varepsilon}\mathbb E_\varepsilon[g]
=
\mathrm{Cov}_\varepsilon\!\Big(g(\bSigma),\ \partial_\varepsilon\log q_\varepsilon(\bSigma)\Big).
\]
Since $\P'(\varepsilon)= A$ and $\log h(\bSigma)$ does not depend on $\varepsilon$,
\[
\partial_\varepsilon\log q_\varepsilon(\bSigma)
=
-\frac12\,\mathrm{tr}(A \,\bSigma^{-1}).
\]
Therefore
\begin{equation}
\label{eq:deriv_cov_general_A}
\frac{d}{d\varepsilon}\mathbb E_\varepsilon[\bR_{kk'}]\Big|_{\varepsilon=0}
=
-\mathrm{Cov}_0\!\Big(\bR_{kk'},\ \frac12\mathrm{tr}(A \,\bSigma^{-1})\Big).
\end{equation}

Expanding the trace and using $\mathrm{diag}(A)=0$ gives
\[
\frac12\mathrm{tr}( A \,\bSigma^{-1})
=
\sum_{1\le i<j\le K}\,A_{ij}\,(\bSigma^{-1})_{ij}.
\]

Thus
\begin{equation}
\label{eq:deriv_sum_general_A}
\frac{d}{d\varepsilon}\mathbb E_\varepsilon[\bR_{kk'}]\Big|_{\varepsilon=0}
=
-\sum_{i<j} A_{ij}\,
\mathrm{Cov}_0\!\big(\bR_{kk'},(\bSigma^{-1})_{ij}\big).
\end{equation}


At $\varepsilon=0$ we have $\P=I_K$.
For any sign-flip matrix $S=\mathrm{diag}(s_1,\dots,s_K)$ with $s_m\in\{\pm1\}$,
define the map $\Sigma\mapsto S\Sigma S$.
We show now that $\bSigma\stackrel{d}{=}S\bSigma S$. Firstly, 
\[
\mathrm{tr}(\P(S\Sigma S)^{-1})
=\mathrm{tr}(\P\,S\Sigma^{-1}S)
=\mathrm{tr}(S\P S\,\Sigma^{-1})
=\mathrm{tr}(\P\Sigma^{-1}),
\]
because $\P$ is diagonal and hence $S\P S=\P$.
Moreover, $|S\Sigma S|=|\Sigma|$ and $S\Sigma S$ has the same eigenvalues as $\Sigma$,
so the remaining factors $|\Sigma|^{-\nu}$ and $\prod_{i<j}(\lambda_i-\lambda_j)^{-1}$ are unchanged.
Therefore, for any integrable $f$,
\begin{equation}
\label{eq:flip_invariance}
\mathbb E_0[f(\bSigma)]=\mathbb E_0[f(S\bSigma S)].
\end{equation}

Now note that under $\Sigma\mapsto S\Sigma S$,
\[
R_{kk'}(\Sigma)\ \mapsto\ s_ks_{k'}\,R_{kk'}(\Sigma),
\qquad
(\Sigma^{-1})_{ij}\ \mapsto\ s_is_j\,(\Sigma^{-1})_{ij}.
\]
Hence the product transforms as
\[
R_{kk'}(\Sigma)\,(\Sigma^{-1})_{ij}\ \mapsto\ (s_ks_{k'}s_is_j)\,R_{kk'}(\Sigma)\,(\Sigma^{-1})_{ij}.
\]

Fix a pair $(i,j)$ with $i<j$.
If $\{i,j\}\neq\{k,k'\}$, then at least one of $k$ or $k'$ is not in $\{i,j\}$.
Choose a sign-flip matrix $S$ that flips exactly that index, e.g. take $s_{k'}=-1$ and all other $s_m=+1$
when $k'\notin\{i,j\}$. Then $s_ks_{k'}=-1$ but $s_is_j=+1$, so
\[
R_{kk'}(\Sigma)\,(\Sigma^{-1})_{ij}\ \mapsto\ -\,R_{kk'}(\Sigma)\,(\Sigma^{-1})_{ij}.
\]
Using invariance \eqref{eq:flip_invariance}, we get
\[
\mathbb E_0\!\big[R_{kk'}(\bSigma)\,(\bSigma^{-1})_{ij}\big]
=
\mathbb E_0\!\big[R_{kk'}(S\bSigma S)\,((S\bSigma S)^{-1})_{ij}\big]
=
-\mathbb E_0\!\big[R_{kk'}(\bSigma)\,(\bSigma^{-1})_{ij}\big],
\]
hence
\begin{equation}
\label{eq:cross_moment_zero}
\mathbb E_0\!\big[R_{kk'}(\bSigma)\,(\bSigma^{-1})_{ij}\big]=0
\qquad \text{for all }(i,j)\neq(k,k').
\end{equation}

Also, by the same argument with $(\bSigma^{-1})_{ij}$ replaced by $1$, we have $\mathbb E_0[\bR_{kk'}]=0$.
Therefore, for $(i,j)\neq(k,k')$,
\[
\mathrm{Cov}_0(\bR_{kk'},(\bSigma^{-1})_{ij})
=
\mathbb E_0[\bR_{kk'}(\bSigma^{-1})_{ij}] - \mathbb E_0[\bR_{kk'}]\mathbb E_0[(\bSigma^{-1})_{ij}]
=
0.
\]
Plugging this into \eqref{eq:deriv_sum_general_A}, we obtain 
\begin{equation}
\label{eq:only_kk_term_survives}
\frac{d}{d\varepsilon}\mathbb E_\varepsilon[\bR_{kk'}]\Big|_{\varepsilon=0}
=
-A_{kk'}\ \mathbb E_0\!\big[\bR_{kk'}(\bSigma^{-1})_{kk'}\big].
\end{equation}

Now we  evaluate $\mathbb E_0[\bR_{kk'}(\bSigma^{-1})_{kk'}]$. 
At $\epsilon = 0$,
the eigenvector matrix $\bm U$ is Haar and independent of $\bm \Lambda$.
Moreover, the $\bm\lambda_i$ are independent inverse-gamma:
\[
\bm\lambda_i \stackrel{i.i.d.}{\sim} \mathrm{Inv\text{-}Gamma}(\nu-1,\ \tfrac12).
\]

Write $\bm{\Sigma=U\Lambda U^\top}$ and $\bm{\Sigma^{-1}=U\Lambda^{-1}U^\top}$. Then
\[\bm 
\Sigma_{kk'}=\sum_{i=1}^K \bm\lambda_i \bm u_{ki}\bm u_{k'i},
\qquad
(\bSigma^{-1})_{kk'}=\sum_{j=1}^K \bm\lambda_j^{-1} \bm u_{kj}\bm u_{k'j}.
\]
Thus
\[ 
\bSigma_{kk'}(\bSigma^{-1})_{kk'}
=
\sum_{i,j=1}^K \frac{\bm\lambda_i}{\bm\lambda_j}\,\bm u_{ki}\bm u_{k'i}\bm u_{kj}\bm u_{k'j}.
\]
Define
\[
a:=\mathbb E[\bm u_{k1}^2\bm u_{k'1}^2],\qquad
b:=\mathbb E[\bm u_{k1}\bm u_{k'1}\bm u_{k2}\bm u_{k'2}],\qquad (k\neq k').
\]
Haar/sphere fourth-moment identities give $a=\frac{1}{K(K+2)}$.
Using orthogonality of distinct rows $\sum_{i=1}^K \bm u_{ki}\bm u_{k'i}=0$, we have
\[
0=\mathbb E\Big[\Big(\sum_{i=1}^K \bm u_{ki}\bm u_{k'i}\Big)^2\Big]
=Ka+K(K-1)b,
\qquad\Rightarrow\qquad
b=-\frac{1}{K(K-1)(K+2)}.
\]
Given $\Lambda$ and $U$ are independent, we have
\[
\mathbb E[\bSigma_{kk'}(\bSigma^{-1})_{kk'}]
=Ka + b\sum_{i\neq j}\mathbb E\Big[\frac{\bm \lambda_i}{\bm \lambda_j}\Big].
\]
Since the $\lambda_i$ are iid inverse-gamma, for $i\neq j$ we have independence and
\[
\mathbb E\Big[\frac{\bm \lambda_i}{\bm \lambda_j}\Big]
=\mathbb E[\bm \lambda]\ \mathbb E[\bm \lambda^{-1}]
=\frac{1}{2(\nu-2)}\cdot 2(\nu-1)
=\frac{\nu-1}{\nu-2}.
\]
Thus 
\[
\mathbb E_0[\bSigma_{kk'}(\bSigma^{-1})_{kk'}]
=
Ka + b\,K(K-1)\frac{\nu-1}{\nu-2}
=
\frac{1}{K+2}-\frac{1}{K+2}\frac{\nu-1}{\nu-2}
=
-\frac{1}{(\nu-2)(K+2)}.
\]
Also,
\[
\mathbb E_0[\bSigma_{kk}]
=\sum_{i=1}^K \mathbb E[\bm \lambda_i]\mathbb E[\bm u_{ki}^2]
=K\cdot \frac{1}{2(\nu-2)}\cdot\frac{1}{K}
=\frac{1}{2(\nu-2)}.
\]
Given 
\[
\bm R_{kk'}(\bSigma^{-1})_{kk'}
=
\bSigma_{kk'}(\bSigma^{-1})_{kk'}\ (\bSigma_{kk}\bSigma_{k'k'})^{-1/2}, 
\]
we show in the following that 
\[
 \E \left[ \bSigma_{kk'}(\bSigma^{-1})_{kk'} (\bSigma_{kk}\bSigma_{k'k'})^{-1/2} \right] = \frac{\mathbb E_0[\bSigma_{kk'}(\bSigma^{-1})_{kk'}]}{\mathbb E_0[\bSigma_{kk}]}
+
O\!\left(\frac{1}{ K^2\nu}\right),
\]
where 
\[\frac{\mathbb E_0[\bSigma_{kk'}(\bSigma^{-1})_{kk'}]}{\mathbb E_0[\bSigma_{kk}]} = -\frac{2}{K+2}. \]
Denote $E_0[\bSigma_{kk}]$ by $\mu$, and let $\bm \delta = \bSigma_{kk} / \mu - 1$, we write 
\[(\bSigma_{kk})^{-1/2} = \mu^{-1/2}(1+\bm \delta)^{-1/2}. \]
Using Taylor theorem at $0$, we have 
\[(\bSigma_{kk})^{-1/2} = \mu^{-1/2}\left(1 - \frac{\delta}{2} + \frac{3}{8} (1 + \bm \xi)^{-5/2} \bm \delta^2 \right), \quad \bm \xi \in (0, \bm \delta).\]
Similarly, let $\bm \delta^\prime = \bSigma_{k^\prime k^\prime} / \mu - 1$, we have 
\[(\bSigma_{k^\prime k^\prime})^{-1/2} = \mu^{-1/2}\left(1 - \frac{\bm \delta^\prime}{2} + \frac{3}{8} (1 +\bm \xi^\prime)^{-5/2} (\bm \delta^\prime)^2 \right), \quad \bm \xi^\prime \in (0, \bm \delta^\prime).\]
Thus 
\[(\bSigma_{kk}\bSigma_{k^\prime k^\prime})^{-1/2} = \mu^{-1}\left(1 - \frac{\bm \delta + \bm \delta^\prime}{2} + \bm R \right),\]
where 
\[  
\begin{aligned}
 \bm R &= \frac{3}{8}(1 + \bm \xi^\prime)^{-5/2} (\bm \delta^\prime)^2 + \frac{3}{8}(1 + \bm \xi)^{-5/2} \bm \delta^2 + \frac{9}{64}(1 + \bm \xi^\prime)^{-5/2}(1 + \bm \xi)^{-5/2} (\bm \delta^\prime)^2\bm \delta^2    \\
 & + \frac{3}{16}(1 + \bm \xi^\prime)^{-5/2} \bm \delta(\bm \delta^\prime)^2 + \frac{3}{16}(1 + \bm \xi)^{-5/2} (\bm \delta^\prime)\bm \delta^2 + \frac{1}{4}\bm \delta^\prime\bm \delta.
\end{aligned}\]
Denote $\bSigma_{kk'}(\bSigma^{-1})_{kk'}$ by $\bm Z$, therefore 
\[
\begin{aligned}
&\E \left[ \bm Z (\bSigma_{kk}\bSigma_{k'k'})^{-1/2} \right] - \mathbb E_0[\bm Z]\mu^{-1} =    \E \left[ \bm Z (\bSigma_{kk}(\bSigma_{k'k'})^{-1/2} - \mu^{-1})\right]\\
&=\mu^{-1}\E \left[ \bm Z \left(- \frac{\bm \delta + \bm \delta^\prime}{2} + \bm R \right)\right].
\end{aligned} \]
Note that $\E \left[ \bm Z  \bm R \right] \leq \sqrt{\E \left[ \bm Z^2  \right]\E \left[ \bm R^2 \right]}.$ Given $\nu > 5$, $\E \left[ \bm R^2 \right] < \infty$. When $\E \left[ \bm Z \left(- \frac{\bm \delta + \bm \delta^\prime}{2} \right)\right], \E \left[ \bm Z \bm R \right]$ are both finite, $\E \left[ \bm Z \bm R \right] = o(\E \left[ \bm Z \left(- \frac{\bm \delta + \bm \delta^\prime}{2} \right)\right]).$ 

On the other hand:
\[
\mathbb E_0[\bm Z\bm \delta]
=
-\frac{6}{(K+2)(K+4)(\nu-3)(\nu-2)}.
\]
The proof see Lemma \ref{lem: moment haar}. Thus given $=\frac{1}{2(\nu-2)}$, we have 
\[\mu^{-1}\E \left[ \bm Z \left(- \frac{\bm \delta + \bm \delta^\prime}{2} + \bm R \right)\right] = O\left(\frac{1}{K^2\nu}\right).\]

Plugging this into \eqref{eq:only_kk_term_survives} gives
\begin{equation}
\label{eq:deriv_final}
\frac{d}{d\varepsilon}\mathbb E_\varepsilon[\bR_{kk'}]\Big|_{\varepsilon=0}
=
\frac{2}{K+2}\,A_{kk'} + O\!\left(\frac{1}{K^2\nu}\right)\,A_{kk'}.
\end{equation}

Lastly, we apply Taylor expansion in $\varepsilon$:
\[
\mathbb E_\varepsilon[\bR_{kk'}]
=
\varepsilon\,\frac{d}{d\varepsilon}\mathbb E_\varepsilon[\bR_{kk'}]\Big|_{\varepsilon=0}
+O(\varepsilon^2).
\]
Using \eqref{eq:deriv_final} and $\P_{kk'}=\varepsilon A_{kk'}$ for $k\neq k'$,
\[
\mathbb E(\bR_{kk'})
=
\frac{2}{K+2}\,\P_{kk'} + O(\varepsilon^2)+O\!\Big(\frac{1}{K^2\nu}\Big).
\]
This completes the proof.
\end{proof}

\begin{lemma}\label{lem: moment haar}
Let $\bSigma\sim \siw_1(I_K, \mathbf{1}, \nu)$ with $K > 3, \nu>3$.
Fix $k\neq k'$ and define
\[
\bm X:=\bSigma_{kk},\quad \mu:=\mathbb E_0[\bm X],\quad \bm \delta:=\frac{\bm X-\mu}{\mu},\quad
\bm Z:=\bSigma_{kk'}(\bSigma^{-1})_{kk'}.
\]
Then
\[
\mathbb E_0[\bm Z\bm \delta]=-\frac{6}{(K+2)(K+4)(\nu-3)(\nu-2)}.
\]
\end{lemma}

\begin{proof}

Write $\bSigma=\bm U\bm \Lambda \bm U^\top$ with $\bm U\in O(K)$ Haar and $\bm \Lambda=\mathrm{diag}(\bm \lambda_1,\dots,\bm \lambda_K)$.
For $\siw_1(I_K, \mathbf{1}, \nu)$,  $\bm U$ is Haar and independent of $\bm \Lambda$.
Moreover,
\[
\bm \lambda_i \stackrel{iid}{\sim} \mathrm{Inv\text{-}Gamma}(\alpha,\beta),
\qquad \alpha=\nu-1,\ \beta=\tfrac12.
\]
Therefore (for $\nu>3$) the moments
\[
m_1:=\mathbb E[\bm \lambda]=\frac{1}{2(\nu-2)},\qquad
m_2:=\mathbb E[\bm \lambda^2]=\frac{1}{4(\nu-2)(\nu-3)},\qquad
m_{-1}:=\mathbb E[\bm \lambda^{-1}]=2(\nu-1)
\]
are finite.

By definition,
\[
\bm Z\bm \delta = \bm Z\frac{\bm X-\mu}{\mu}=\frac{\bm Z\bm X}{\mu}-\bm Z
\quad\Longrightarrow\quad
\mathbb E_0[\bm Z\bm \delta]=\frac{\mathbb E_0[\bm Z\bm X]}{\mu}-\mathbb E_0[\bm Z].
\]
We have $\mu=\mathbb E_0[\bm X]=m_1$ because
$\bm X=\sum_{p=1}^K \bm \lambda_p \bm u_{kp}^2$ and $\mathbb E[\bm u_{kp}^2]=1/K$.
Also, the known exact identity (computed with Haar 4th moments) is
\[
\mathbb E_0[\bm Z]=\mathbb E_0[\bSigma_{kk'}(\bSigma^{-1})_{kk'}]
=-\frac{1}{(\nu-2)(K+2)}.
\]
Hence it suffices to compute $\mathbb E_0[\bm Z\bm X]=\mathbb E_0[\bSigma_{kk'}(\bSigma^{-1})_{kk'}\bSigma_{kk}]$.

Let $\bm a_i:=\bm u_{ki}\bm u_{k'i}$ and $\bm b_p:=\bm u_{kp}^2$ (with $k\neq k'$).
Then
\[
\bSigma_{kk'}=\sum_{i=1}^K \bm \lambda_i \bm a_i,\qquad
(\bSigma^{-1})_{kk'}=\sum_{j=1}^K \bm \lambda_j^{-1} \bm a_j,\qquad
\bSigma_{kk}=\sum_{p=1}^K \bm \lambda_p \bm b_p,
\]
so
\[
\bm Z\bm X=\sum_{i,j,p=1}^K \frac{\bm \lambda_i\bm \lambda_p}{\bm \lambda_j}\,\bm a_i \bm a_j \bm b_p.
\]
Since $(\bm \lambda_i)$ are independent of $\bm U$ and iid, and $\bm U$ is Haar,
\[
\mathbb E_0[\bm Z\bm X]=\sum_{i,j,p}\mathbb E\!\Big[\frac{\bm \lambda_i\bm \lambda_p}{\bm \lambda_j}\Big]\,
\mathbb E[\bm a_i \bm a_j \bm b_p].
\]

For a fixed row $k$ and distinct row $k'$, the needed sixth-order Haar integrals can be computed via
Weingarten calculus / orthogonal polynomial integrals \citep{banica2011polynomial}.
The final identities (valid for $K > 3$) are:
\[
\mathbb E[\bm a_i^2 \bm b_i]=\frac{3}{K(K+2)(K+4)},
\qquad
\mathbb E[\bm a_i^2 \bm b_p]=\frac{K+1}{(K-1)K(K+2)(K+4)}\ (p\neq i),
\]
\[
\mathbb E[\bm a_i \bm a_j \bm b_p]=-\frac{1}{(K-1)K(K+2)(K+4)}\quad (i\neq j,\ p\notin\{i,j\}),
\]
\[
\mathbb E[\bm a_i \bm a_j \bm b_i]=\mathbb E[\bm a_i \bm a_j \bm b_j]=-\frac{3}{(K-1)K(K+2)(K+4)}\quad (i\neq j).
\]

By iid inverse-gamma:
\[
\mathbb E\!\Big[\frac{\bm \lambda_i\bm \lambda_p}{\bm \lambda_j}\Big]=
\begin{cases}
m_1, & j=i, \forall p,\\
m_2\,m_{-1}, & i=p\neq j,\\
m_1^2\,m_{-1}, & i\neq j,\ p\neq j,\ p\neq i,\\
m_1, & p=j\neq i.
\end{cases}
\]
Using the combinatorial counts of triples $(i,j,p)$ in each case and the Haar moments above, we have
\[
\mathbb E_0[\bm Z\bm X]
=
-\frac{K\nu-3K+4\nu-6}{2\,(K+2)(K+4)\,(\nu-3)\,(\nu-2)^2}.
\]

Since $\mu=m_1=\frac{1}{2(\nu-2)}$ and $\mathbb E_0[\bm Z]=-\frac{1}{(\nu-2)(K+2)}$,
\[
\mathbb E_0[\bm Z\bm \delta]
=\frac{\mathbb E_0[\bm Z\bm X]}{\mu}-\mathbb E_0[\bm Z]
=
2(\nu-2)\,\mathbb E_0[\bm Z\bm X]+\frac{1}{(\nu-2)(K+2)}
=
-\frac{6}{(K+2)(K+4)(\nu-3)(\nu-2)}.
\]
\end{proof}

\subsection{Proposition \ref{propboundeigengap}}\label{app: propboundeigengap}

\propboundeigengap*

\begin{proof}
Let $\Sigma\in\mathbb{R}^{K\times K}$ be a covariance matrix, hence symmetric positive definite, and define
\[
\kappa(\Sigma)=\frac{\lambda_{\max}(\Sigma)}{\lambda_{\min}(\Sigma)}=: \frac{M}{m}\qquad(M>m>0).
\]
Let $D=\mathrm{diag}(\Sigma_{11},\dots,\Sigma_{KK})$ and 
\[
R \;=\; D^{-1/2}\Sigma D^{-1/2}
\]
be the corresponding correlation matrix. Fix $k\neq k'$ and set
\[
a:=\Sigma_{kk}>0,\qquad b:=\Sigma_{k'k'}>0,\qquad r:=R_{kk'}=\frac{\Sigma_{kk'}}{\sqrt{ab}}.
\]
Define the  orthogonal vectors $u:=e_k/\sqrt{a}$ and $v:=e_{k'}/\sqrt{b}$. Then
\[
u^\top \Sigma u = 1,\qquad v^\top \Sigma v = 1,\qquad u^\top \Sigma v = r.
\]
Let $\varepsilon=\mathrm{sign}(r)\in\{-1,0,1\}$ (with $\mathrm{sign}(0)=0$), and set
\[
w:=u+\varepsilon v,\qquad z:=u-\varepsilon v.
\]
Using $u^\top\Sigma v=r$ we get
\[
w^\top \Sigma w
= u^\top\Sigma u + v^\top\Sigma v + 2\varepsilon\,u^\top\Sigma v
= 2 + 2|r| = 2(1+|r|),
\]
\[
z^\top \Sigma z
= u^\top\Sigma u + v^\top\Sigma v - 2\varepsilon\,u^\top\Sigma v
= 2 - 2|r| = 2(1-|r|).
\]
Moreover, since $u\perp v$, we have
\[
\|w\|^2=\|z\|^2=\|u\|^2+\|v\|^2=\frac1a+\frac1b=:s>0.
\]

Now, for every nonzero vector $x$, the Rayleigh quotient bound for a symmetric positive definite matrix gives
\[
m\|x\|^2 \;\le\; x^\top\Sigma x \;\le\; M\|x\|^2.
\]
Applying this to $w$ and $z$ yields
\[
ms \le 2(1+|r|) \le Ms,
\qquad
ms \le 2(1-|r|) \le Ms.
\]
In particular,
\[
1+|r| \le \frac{Ms}{2}
\quad\text{and}\quad
1-|r| \ge \frac{ms}{2}.
\]
Dividing the two inequalities we obtain
\[
\frac{1+|r|}{1-|r|}
\;\le\;
\frac{(Ms/2)}{(ms/2)}
=
\frac{M}{m}
=
\kappa(\Sigma).
\]
Since $1-|r|>0$, we can solve for $|r|$:
\[
1+|r| \le \kappa(\Sigma)\,(1-|r|)
\;\Longleftrightarrow\;
|r|\,(1+\kappa(\Sigma)) \le \kappa(\Sigma)-1
\;\Longleftrightarrow\;
|r| \le \frac{\kappa(\Sigma)-1}{\kappa(\Sigma)+1}.
\]
\end{proof}

\subsection{Proposition \ref{propExchange}}\label{app: exchange}

\propExchange*

\begin{proof}
We have for any permutation matrix $Q$,  $Q\bR Q^{-1}$ is the induced correlation matrix of $Q \bSigma Q^{-1}$. Thus if $Q \bSigma Q^{-1} \stackrel{d}{=} \bSigma$, we have the desired result: $Q \bR Q^{-1} \stackrel{d}{=} \bR$. In the following, we show that $Q \bSigma Q^{-1} \stackrel{d}{=} \bSigma$. 

Since in the case of $\siw_1$, we do not have the handy tool of sample representation, we rely directly on the density. Given the $\siw_1$ density of $\bSigma$: 
\begin{equation}
    \pi(\Sigma) \propto \frac{\exp\left(-\frac{1}{2}\mbox{tr}\left( \Psi \Sigma^{-1}\right)\right)}{det(\Sigma)^\nu \left[\prod_{i < j} (\lambda_i - \lambda_j)\right]},
\end{equation} 
the density of $Q \bSigma Q^{-1}$ can be calculated through the change of variables:
\begin{equation}
\pi_Q(\Sigma') =    \pi(f^{-1}(\Sigma')) |det J_{f^{-1}}(\Sigma')|, 
\end{equation}
where $f(\Sigma) = Q \Sigma Q^{-1}$, thus $f^{-1}(\Sigma') = Q^{-1} \Sigma' Q$. Note that 
\[det( Q^{-1} \Sigma' Q) = det(\Sigma'), \; \lambda_i(Q^{-1} \Sigma' Q) = \lambda(\Sigma'), \forall i =1, ..., K, \]
and 
\[tr(\Psi (Q^{-1} \Sigma' Q)^{-1}) = tr(\Psi Q^{-1} \Sigma' Q)= tr(Q\Psi Q^{-1} \Sigma').\]
$\Psi = \Delta P \Delta = \sigma^2 P$ under the identity assumption of $\sigma_k$, thus is permutation-invariant, namely, $Q\Psi Q^{-1} = \Psi.$ Thus $\pi(f^{-1}(\Sigma')) = \pi(\Sigma').$ Lastly, 
\[
J_{f^{-1}}(\Sigma') = \frac{d \operatorname{vec}( f^{-1}(\Sigma'))}{d \operatorname{vec}(\Sigma')} = \frac{d \operatorname{vec}(Q^{-1}\Sigma' Q)}{d \operatorname{vec}(\Sigma') } = \frac{d (Q^{-1}\otimes Q^{-1}) \operatorname{vec}(\Sigma') }{d \operatorname{vec}(\Sigma') } = Q^{-1}\otimes Q^{-1}.
\]
Thus $|det J_{f^{-1}}(\Sigma')| = |det Q^{-1}|^{2k} = 1$. We have $\pi_Q(\Sigma') = \pi(\Sigma').$

\end{proof}

\subsection{Property \ref{propertyIWSIWprior}}\label{app: mean var mix}

\propertyIWSIWprior*

\begin{proof}
$\bR$ induced from $\bSigma \sim \iw/\siw_1(\eta, \P_r, \Vec{\sigma}_r, \nu_r, r=0,1)$ is also a mixture of $\bR_0$ and $\bR_1$, respectively induced from $\bSigma_0 \sim \mathcal{IW}(\P_0, \Vec{\sigma}_0, \nu_0)$ and $\bSigma_1 \sim \siw_1(\P_1, \Vec{\sigma}_1, \nu_1)$ with $b =1$. Direct calculation shows that:\begin{equation}
\begin{aligned}
\small 
&\mathbb{E} \bR_{kk'} = \eta \mathbb{E} \bR_{0,kk'} + (1 - \eta) \mathbb{E} \bR_{1,kk'}, \\  
&\mathbb{V} \bR_{kk'} = \eta \mathbb{V} \bR_{0,kk'} + (1 - \eta) \mathbb{V} \bR_{1,kk'} + \eta (1 - \eta) (\mathbb{E} \bR_{0,kk'} -  \mathbb{E}\bR_{1,kk'})^2.
\end{aligned}
\end{equation}
Plug in the formulas of $\mathbb{E} \bR_{r,kk'}$ and $\mathbb{V} \bR_{r,kk'}$ for $r=0,1$, given in Proposition \ref{prop: IW prior} and Property \ref{prop: siw}, we have the results for the mixture. Note that, we ignore the terms of $\nu$ in $\mathbb{E}(\bR_{kk'})$, and term of $\P_{1,kk'}$ in $\mathbb{V}(\bR_{kk'})$ in Proposition \ref{prop: IW prior} since they are not the main predictors of the expectation and variance. We ignore as well all terms of expectation in  $\mathbb{E}\bR_{1,kk'}$.
\end{proof}

\subsection{Proposition \ref{propIWSIWpost}}\label{app: IW/SIW post}

\propIWSIWpost*

\begin{proof}
Given the prior on $\bSigma$:
\begin{equation}
\begin{aligned}
    &\bSigma \mid \bm z = 0 \sim \mathcal{IW}(\P_0, \Vec{\sigma}_0, \nu_0), \\ 
    &\bSigma \mid \bm z=1 \sim \siw_1(\P_1, \Vec{\sigma}_1, \nu_1),  \\
    & \bm z \sim \mathcal{B}(\eta),
\end{aligned}
\end{equation}
we have the posterior distribution:
\begin{equation}
\begin{aligned}
    \mathbb{P}(\bSigma \mid \x_n) &= \frac{\mathbb{P}(\bSigma, \x_n)}{\mathbb{P}(\x_n)}  =  \frac{\sum\limits_{r=0}^1 \mathbb{P}(\bm z = r)\mathbb{P}(\bSigma, \x_n \mid \bm z = r)}{\mathbb{P}(\x_n)} \\
    &=\frac{\sum\limits_{r=0}^1 \mathbb{P}(\bm z = r)\mathbb{P}(\x_n \mid \bm z = r)\mathbb{P}(\bSigma \mid \x_n, \bm z = r)}{\mathbb{P}(\x_n)}.
\end{aligned}
\end{equation}
Note that 
$$\mathbb{P}(\x_n \mid \bm z = 0) = \int \mathbb{P}(\x_n \mid \bSigma) \mathbb{P}(\bSigma\mid \bm z = 0) d \bSigma = \mathbb{E}_{\bSigma \sim \mathcal{IW}(\P_0, \Vec{\sigma}_0, \nu_0)} \mathbb{P}(\x_n \mid \bSigma) = L(\x_{1:n} \mid \bm z = 0).$$
Similarly, $\mathbb{P}(\x_n \mid \bm z=1) = L(\x_{1:n} \mid \bm z=1).$ Given that $\mathbb{P}(\bm z = r)$, we have 
$$
\begin{aligned}
 &\frac{\sum\limits_{r=0}^1 \mathbb{P}(\bm z = r)\mathbb{P}(\x_n \mid \bm z = r)\mathbb{P}(\bSigma \mid \x_n, \bm z = r)}{\mathbb{P}(\x_n)} \\
 =& \; \eta_n\mathbb{P}(\bSigma \mid \x_n, \bm z = r) + (1-\eta_n)\mathbb{P}(\bSigma \mid \x_n, \bm z = r)   
\end{aligned}
$$

On the other hand, $\mathbb{P}(\bSigma \mid \bm z = 0)$ is the $\iw(\P^{1}, \Vec{\sigma}_0, \nu_{1})$, given its Gaussian conjugacy, $\mathbb{P}(\bSigma \mid \x_n, \bm z = 0)$ is $\iw(\P^{1,n}, \Vec{\sigma}_0^{(n)}, \nu_{0,n})$. Similarly, $\mathbb{P}(\bSigma \mid \x_n, \bm z=1)$  is $\siw_1(\P^{2,n}, \Vec{\sigma}_1^{(n)}, \nu_{1,n})$. Thus the posterior distribution is still a $\iw/\siw_1$ mixture with the weight $\eta_n$.     
\end{proof}

\subsection{Proposition \ref{propL}}\label{app: L1}
\ppropL*

\begin{proof}
    Denote $p_\iw(\bSigma; \Psi_0, \nu_0)$, $k_\iw(\bSigma; \Psi_0, \nu_0)$, and $c_\iw(\Psi_0, \nu_0)$ respectively the probability density function, the kernel (or unnormalized density function), and normalization constant of $\iw(\Psi_0, \nu_0)$, then we have 
\begin{equation}\label{eq: L1 calcul}
\begin{aligned}
     &L(\x_{1:n} \mid \bm z = 0) = \int\left[\prod_{i=1}^n p_{\mathcal{N}}(\bm x_i; 0, \bSigma)\right] p_{\iw}(\bSigma ; \Psi_0, \nu_{1}) d \bSigma   \\
     &= c_\iw(\Psi_0, \nu_0)\int\left[\prod_{i=1}^n p_{\mathcal{N}}(\bm x_i; 0, \bSigma)\right] k_\iw(\bSigma, \Psi_0, \nu_0) d \bSigma \\
     &= c_\iw(\Psi_0, \nu_0)\int(2\pi)^{-\frac{KT}{2}} |\bSigma|^{-\frac{n}{2}} \exp\left(\mbox{tr}\left( -\frac{1}{2}\bSigma^{-1}\bm S\right)\right) k_\iw(\bSigma, \Psi_0, \nu_0) d \bSigma \\
     &= (2\pi)^{-\frac{KT}{2}}c_\iw(\Psi_0, \nu_0)\int |\bSigma|^{-\frac{(n+\nu_0) + K+1}{2}} \exp\left(\mbox{tr}\left( -\frac{1}{2}\bSigma^{-1}\bm (\bm S+\Psi_0)\right)\right) k_\iw(\bSigma, \Psi_0, \nu_0) d \bSigma\\
     &=\frac{c_\iw(\Psi_0, \nu_0)}{(2\pi)^{\frac{KT}{2}}c_\iw(\Psi_0+ \bm S, \nu_0+n)}.
\end{aligned}
\end{equation}
where $\bm S = \sum\limits_{i=1}^n\bm x_i\bm x_i^\top$. 
Given the closed form of constant $c_\iw(\Psi_0, \nu_0) = |\Psi_0|^{\frac{\nu_0}{2}}/(2^{\frac{\nu_0K}{2}})\Gamma_K(\frac{\nu_0}{2})$ with $\Gamma_K$ multivariate gamma function, $L(\x_{1:n} \mid \bm z = 0)$ can be calculated explicitly. 
\end{proof}

\newpage

\section{Posterior inference for $\iw/\siw_1$ mixture}\label{app: post inf}
\subsection{Algorithm to calculate \texorpdfstring{$L(\x_{1:n} \mid \bm z=1)$}{}}

\begin{algorithm}
\caption{Estimation of $L(\x_{1:n} \mid \bm z=1)$ for $\x_{1:n} \stackrel{iid}{\sim} \mathcal{N}(0, \bSigma), \, 
\bSigma \sim \siw_1(\Psi_1, \nu_1)$}
\label{alg}
\begin{algorithmic}[1]
\Require degree of freedom $\nu_1$, scale matrix $\Psi_1$, data $\x_{1:n}$, sample size for proposal $M$, clipping size $M_n$.
\Ensure $\widehat{\log L}_1$.

\vspace{0.1in}
\textsc{/* Sampling from proposal */}

\For{$m = 1,\dots,M$}
  \State Draw $\alpha_{i,j} \stackrel{iid}{\sim} \mathcal{N}(0, 1), \; i,j =1, \ldots, K,$ and form the matrix $A := [\alpha_{i,j}]$
  \State Compute QR decomposition of $A$, denoted by $A = QR$, and set $\Gamma = Q$
  \For{$i = 1,\dots,K$}
      \State Draw $\lambda_i \sim \mathcal{IG}\bigl(\nu_1-1,\frac{1}{2}\Gamma_{i}^\top\Psi_1\Gamma_{i}\bigr)$
  \EndFor
  \State $\Log w_m = K\Log f_\Gamma(\nu_1-1) - \sum_{i=1}^K(\nu_1-1)\Log\left(\frac{1}{2}\Gamma_{i}^\top \Psi_1 \Gamma_{i}\right).$
  \State  Set $I_i, i =1, \ldots, K$ such that $ \lambda_{I_1} \geq \ldots \geq \lambda_{I_K}.$
    \State Define matrix $\Tilde{\Delta} := \mbox{Diag}\{\lambda_{I_1}, ..., \lambda_{I_K}\}$, and matrix $\Tilde{\Gamma}$ such that $\Tilde{\Gamma}_{:,i} = \left[\Gamma\right]_{:,I_i}$.
  \State Set $\Tilde{\Sigma}^{(m)} \leftarrow \Tilde{\Gamma} \,\Tilde{\Delta}\,\Tilde{\Gamma}^\top$.
\EndFor

\vspace{0.1in}
\textsc{/* Clipping */}
\State 
Let $\Log w_{(M_n)}$ the $M_n$-th greatest weight. For $m = 1, \ldots, M$, define  
\vspace{-0.1in}
\begin{equation}
    \begin{cases}
      \Log w_{m}^\prime =  \Log w_{(M_n)}, \; &\mbox{if} \;  \Log w_{m} >  \Log w_{(M_n)}, \\
      \Log w_{m}^\prime = \Log w_{m}, \; &\mbox{otherwise}.
    \end{cases}
\end{equation}

\textsc{/* Calculating the normalized weights */}
\State \label{line} 
For $m = 1, \ldots, M$, compute 
$p_m = \Tilde{w}_m/\sum_{m=1}^M \Tilde{w}_m,$ where $\Tilde{w}_m = \Exp(\Log w_m^\prime - \max_m\{ \Log w_m^\prime \}).$

\vspace{0.1in}
\textsc{/* Calculating the estimation */}

\vspace{0.05in}
\State Define $\log f_m = \sum_{i=1}^n\log p_\mathcal{N}(\bm x_i; 0, \Tilde{\Sigma}^{(m)})$, then $\log f^p_m = \log f_m + \log p_m$.

\vspace{0.03in}
\State $log L_1 =  \max_m\{\log f^p_m\} + \log(\sum_{i=1}^n\exp( \log f^p_m - \max_m\{\log f^p_m\}))$. 

\vspace{0.1in}
\Return $\widehat{\log L}_1.$
\end{algorithmic}
\end{algorithm}
The $\mbox{LogSumExp}$ operation on line $15$ is to calculate  $\log(\sum_{i=1}^n f^p_m)$ without explicitly using $f^p_m$ since they are either too large or too small thus causing overflow or underflow. This algorithm is an adaption of Algorithm 3 in \cite{jiang2025bayesian}, where a clipping step is integrated into the importance sampling framework in order to increase the robustness of the output. We offer the convergence result in the proposition \ref{prop: algo IS}. 
\begin{prop}\label{prop: algo IS}
Provided $\lim_{M\rightarrow \infty}M_n/ M = 0$ and $\nu_1 > 1$, we have 
$$\widehat{\log L}_1 \stackrel{P}{\longrightarrow} \log L(\x_{1:n} \mid \bm z=1).$$ Moreover, given $\lim_{M\rightarrow \infty}M_n/ \sqrt{M} = 0$, we have 
$$\sqrt{M}\left(\widehat{\log L}_1 - \log L(\x_{1:n} \mid \bm z=1)\right) \stackrel{D}{\longrightarrow} \mathcal{N}(0, \sigma^2),$$
with $\sigma^2$ some constant. 
\end{prop}
\textit{Proof:} 
Note that $\widehat{\log L}_1$ is an importance sampling estimator taking the form 
\begin{equation}
\Hat{\mu}^{IS}_{M,M_n}(f) =
\frac{\sum_{m=1}^M w_m^\prime f(\Tilde{\Sigma}^{(m)})}{\sum_{m=1}^M w_m^\prime}, 
\end{equation}
with $f(\Sigma) = \sum_{i=1}^n\log p_\mathcal{N}(\bm x_i; 0, \Sigma)$, namely the Gaussian likelihood. We apply Lemma C.3 and C.4 in the supplemental materials of \cite{jiangsiw25}. It suffices to show that $\mathbb{E}f(\Tilde{\Sigma}^{(m)})$ is finite, which is shown true in Example 2.5 in \cite{jiangsiw25} provided $\nu_1 > 1$. This completes the proof. 

The theoretical results imply that we can fix $M_n$ by a value less than $M$ to have a consistent estimator. For the real data set, $M$ is set $500000$ and $M_n = M^0.9$. The algorithm has computation advantage, thus we can set $M$ very large with large $M_n$ to get more precise and meanwhile more robust result. For example, for $K=20$, one run takes $2$ minutes. 3 runs give the estimates for when $\rho_1 = 0.3, \nu_1 = 40$: $-3745$, $-3741$, and $-3742$, which are very stable.

When $\rho_1 = 0$ thus $\Psi_1$ is proportional to the identity matrix $I$, we use the naive method to estimate $\log L(\x_{1:n} \mid \bm z=1)$, which is a sample mean estimator of $\mathbb{E}_{\bSigma \sim \siw}f(\bSigma)$. The sampling from $\siw_1$ is done by Algorithm 2 in \cite{jiangsiw25}, with $N=50000$, which is an exact and fast sampling. The large $N$ is enough for giving a robust result even the naive method can be less robust with smaller sample size $N$. 

\subsection{Calculation of posterior mean}\label{sec-app: cal post mean}
The posterior correlation mean is a weighted average of the correlation mean of $\iw$-inducing posterior and the one of $\siw_1$-inducing posterior. For the $\iw$ mean, we use R function and generate $1000$ covariance samples, then convert them to correlation sample, and calculate the sample mean as the estimation. For the $\siw_1$ mean, we use Algorithm 3 in \cite{jiangsiw25}, and generate $N=50000$ covariance samples with $M=50000, M_n = M^{0.8}$, then convert to the correlation samples, finally calculate the sample mean.

\subsection{Supplementary results on the rat data of $\iw / \siw$ model}\label{app: supp results}
\subsubsection{Dead rats}\label{sec: dead}
In the section, we report the results from other hyperparameter values of $\rho_r, \nu_r, r=0,1$ over the dead rat data set. In particular, we would like to illustrate two types of results that $\iw/\siw_1$ prior will produce and their corresponding scenarios, which are \textit{model selection} and \textit{model mixture}. The model selection happens when one of the two component priors is too far from the data, thus the mixture prior will choose the good prior with the posterior weight $\eta_n \in \{0, 1\}$. This can bee seen through the updating formula:
\begin{equation}
    \eta_{n} = \frac{\eta L(\x_{1:n} \mid \bm z = 0)}{\eta L(\x_{1:n} \mid \bm z = 0) + (1-\eta) L(\x_{1:n} \mid \bm z=1)}.  
\end{equation}
For example when $\siw_1$ prior is much further away from the data, resulting in that $L(\x_{1:n} \mid \bm z = 0) \gg L(\x_{1:n} \mid \bm z=1)$, then $\eta L(\x_{1:n} \mid \bm z = 0) + (1-\eta) L(\x_{1:n} \mid \bm z=1) \approx \eta L(\x_{1:n} \mid \bm z = 0)$, thus $\eta_{n} \approx 1$. This example also implies that the model selection does not have to happen when the prior weight is $0.5$, and both priors can be bad as long as one is much worse than the other. By contrast, the model mixture happens when the two priors are comparable, which similarly can be both good or both bad. Then the posterior weight $\eta_n$ will fall into $(0, 1)$, with both priors taken into account by the posterior results. 

We simulate both scenarios. Because the dead rat has ``ground truth'', which is that the correlations of regional pairs are all zero. Thus we create bad priors by setting $\P_{r,kk'}$ non-zero. Higher absolute values of $\P_{r,kk'}$ or/and higher $\nu_r$ means the prior is further away from the data. 
Firstly, Table \ref{tab: model choice} reports the simulations for differents scenarios, with a range of $\rho_{0,r}, \nu_r, r =0, 1$. 
\begin{table}[ht]
    \centering
    (a)
    
    \begin{tabular}{|c|ccc|}
      \hline
      \diagbox[height=2.5em, width=8em]{$(\rho_0, \nu_0)$}{$(\rho_1, \nu_1)$} &  $(0,10)$ & $(0,18)$ & $(0,36)$\\
    \hline
        $(0.1,50)$ & 0 & 0 & 0 \\
        $(0.1,90)$ & 0 & 0 & 0\\
        \hline
    \end{tabular}

    (b)
    
    \begin{tabular}{|c|cc|cc|cc|}
    \hline
      \diagbox[height=2.5em, width=8em]{$(\rho_0, \nu_0)$}{$(\rho_1, \nu_1)$} &  $(0.3,10)$ & $(0.3,36)$ & $(0.5,10)$ & $(0.5,36)$ & $(0.8,10)$  & $(0.8,36)$\\
    \hline
        $(0,50)$ & 0 & 0 & 0 & 1 & 1  & 1\\
        $(0,90)$ & 0 & 1 & 1 & 1 & 1  & 1\\
        \hline
    \end{tabular}
    \caption{Posterior weight $\eta_n$ when one of the prior components is much further away from the data. Recall that  $\rho_{0,r}, r =0, 1$ determine the prior correlation mean in the way: $\mathbb{E}(\bR_{kk'} \mid \small \bSigma \sim \mathcal{IW} \,) \approx 
    \rho_0 \mbox{ and }  \mathbb{E}(\bR_{kk'} \mid \small \bSigma \sim \siw \, ) \approx 0.066\rho_1$. }
    \label{tab: model choice}
\end{table}

In the subtable (a) of Table \ref{tab: model choice}, $\iw$ prior is further away from the data given that $\rho = 0.1$, thus the induced correlation entries $\bR_{kk'}, k \neq k'$ will centered around $0.1$. Meanwhile the induced correlation entries from $\siw_1$ prior will centered around $0$, which is the ground truth of the data. Consequently, all $\eta_n =0 $, implying that the mixture model chose the $\siw_1$ prior. In the case, the mixture model reduces to the individual model $\x_{1:n} \stackrel{iid}{\sim} \mathcal{N}(0, \bSigma), \, \bSigma \sim \siw_1(\P_1, \Vec{\sigma}_1, \nu_1).$ In the subtable (b) of Table \ref{tab: model choice}, $\siw_1$ prior is given positive $\rho_1$. However the induced correlation entries $\bR_{kk'}, k \neq k'$ are not exactly centered around $\rho_1$, instead, around $0.066\rho_1$ (recall the result in Figure \ref{fig: ER_K20}). Thus when $\rho_1$ is small such as $0.3$, the center of $\bR_{kk'}, k \neq k'$ is still very close to $0$. They can only become far away from the data when a sufficiently high $\nu_1$ is adopted, leading to a very peaky and concentrated prior distribution of $\bR_{kk'}$, hence leaving little mass on $0$. 

This explains why the model chose $\siw_1$ on the left part of the subtable (b). More specifically, when $\rho_1 = 0.3$, the prior mean of $\bR_{kk'}$ is approximately $0.02$, almost $0$. Thus only when $\nu_1$ is high with more mass on $0.02$ for $\siw_1$, meanwhile, $\nu_0$ is high as well with more mass on $0$ for $\iw$, the mixture can choose $\iw$, as for the case $(\rho_0, \nu_0) = (0, 90)$ versus $(\rho_1, \nu_1) = (0.3, 36)$. When increasing $\rho_1$, the mixture choose almost always $\iw$ for all $\nu_r, \, r =0, 1$. 

\begin{table}[ht]
    \centering
    \setlength{\tabcolsep}{4pt} 
    \renewcommand{\arraystretch}{1.2} 

    \caption{Posterior weight $\eta_n$ when the two priors are comparable. 
    $\rho_0$ is fixed at $-0.05$ in all models.}
    \label{tab: mixture}
    
    \begin{tabular}{@{}lc|cccccccc@{}}
        \toprule
        $\rho_1$ & $\nu_1$ & \multicolumn{8}{c}{$\nu_0$} \\
        \midrule
        \multicolumn{2}{c|}{} & $[24,54]$ & $56$ & $58$ & $[60,166]$ & $168$ & $170$ & $172$ & $174$ \\
        \cmidrule(lr){3-10}
        \rowcolor{lightergray} 
        0.5 & 18 & 0 & 0.22 & 0.87 & 1 & 0.99 & 0.96 & 0.91 & 0.80 \\
        \midrule
        \multicolumn{2}{c|}{} & $[24,38]$ & $40$ & $[42,258]$ & $260$ & $262$ & $264$ & $266$ & $268$ \\
        \cmidrule(lr){3-10}
         \rowcolor{lightergray}
         0.5 & 36 & 0 & 0.10 & 1 & 0.97 & 0.93 & 0.82 & 0.60 & 0.33 \\
        \midrule
        \multicolumn{2}{c|}{} & $[24,60]$ & $62$ & $64$ & $66$ & $[68,148]$ & $150$ & $152$ & $154$ \\
        \cmidrule(lr){3-10}
         \rowcolor{lightergray}
         0.6 & 10 & 0 & 0.36 & 0.86 & 0.98 & 1 & 0.98 & 0.95 & 0.90 \\
        \midrule
        \multicolumn{2}{c|}{} & $[24,32]$ & $34$ & $[36,330]$ & $332$ & $334$ & $336$ & $338$ & $340$ \\
        \cmidrule(lr){3-10}
        \rowcolor{lightergray}
         0.6 & 18 & 0 & 0.14 & 1 & 0.98 & 0.96 & 0.88 & 0.71 & 0.45 \\
        \bottomrule
    \end{tabular}
\end{table}

Next, we report the simulations for the scenario--model mixture--. To make the priors comparable, we set both $\rho_{0,r}, r =0, 1$ non-zero. More specifically, we focus\footnote{The lower bound of $\rho_{0,r}$ is $-0.053$ (see Definition \ref{def: iw}).} on $\rho_0 = -0.05$ and $\rho_1 \in \{0.5, 0.6\}$. We test $3$ values for $\nu_1$: $10,  18, 36$. A finer grid is adopted for $\nu_0$ in order to reach comparable priors, that we vary it from $24$ to $500$ by step size of $2$. The results are shown in Table \ref{tab: mixture}. First we can see that in all cases, $\eta_n$ starts from $0$ to $1$ then decreases again. This implies that the mixture considers that the $\iw$ prior gets closer to the data then drifts away again as $\nu_0$ constantly increases. This can be explained by a nearly zero $\rho_0$. Because when $\nu_0$ is very small, the prior distribution of correlation is too flat hence not enough mass on $0$, as it becomes more concentrated on $-0.05$ more mass come to $0$ alongside, then the mass leaves again when the distribution becomes peaky. Secondly, one can observe that the $\nu_0$ value for $\eta_n$ to decrease from $1$ (hence the two priors become comparable again) varies, depending on the far the $\siw_1$ prior away from the data. Further away, higher $\nu_0$ is needed to make the $\iw$ prior worse. For example, the $\siw_1$ prior with $(\rho_1, \nu_1) = (0.5, 18)$ is further away from the data than $(\rho_1, \nu_1) = (0.5, 36)$, thus $\nu_0$ needs to reach $260$ to make the $\iw$ prior comparably bad, instead of $170$.

\subsubsection{Live rats}\label{sec: alive}
Now, we show the result of live rat. We only focus on the --mixture model-- scenario for the live rat. We test $2$ $\siw_1$ priors: $(\rho_1, \nu_1) = (0.3, 18)$ and $(\rho_1, \nu_1) = (0.3, 40)$, and $2$ $\rho_0$ values: 0.4 and 0.5. We vary $\nu_0$ to see the evolution of mixture weight. Table \ref{tab: mixture alive} reports the results. 

\begin{table}[ht]
    \centering
    \caption{Posterior weight $\eta_n$ for the live rat. }
    \label{tab: mixture alive}

    \begin{minipage}[t]{0.48\textwidth} 
        \centering
        \setlength{\tabcolsep}{3pt} 
        \renewcommand{\arraystretch}{1.2}
        \subcaption*{$\rho_0 = 0.4$} 
        
        \begin{tabular}{@{}lc|cccccc@{}}
            \toprule
            $\rho_1$ & $\nu_1$ & \multicolumn{6}{c}{$\nu_0$} \\
            \midrule
            \multicolumn{2}{c|}{} & $30.3$ & $30.4$ & $30.5$ & $30.6$ & $30.7$ & $30.8$ \\
            \cmidrule(lr){3-8}
            \rowcolor{lightergray} 
            0.3 & 18 & 0.08 & 0.24 & 0.53 & 0.79 & 0.93 & 0.98 \\
            \midrule
            \multicolumn{2}{c|}{} &$37.4$ & $37.6$ & $37.8$ & $38.0$ & $38.2$  & $38.4$\\
            \cmidrule(lr){3-8}
             \rowcolor{lightergray}
             0.3 & 40 & 0.09 & 0.22 & 0.46 & 0.72 & 0.88 & 0.95\\
            \bottomrule
        \end{tabular}
    \end{minipage}%
    \hfill 
    \begin{minipage}[t]{0.48\textwidth} 
        \centering
        \setlength{\tabcolsep}{3pt} 
        \renewcommand{\arraystretch}{1.2}
        \subcaption*{$\rho_0 = 0.5$} 

        \begin{tabular}{@{}lc|cccccc@{}}
            \toprule
            $\rho_1$ & $\nu_1$ & \multicolumn{6}{c}{$\nu_0$} \\
            \midrule
            \multicolumn{2}{c|}{} & $34.5$ & $34.7$ & $34.9$ & $35.1$ & $35.3$ & $35.5$ \\
            \cmidrule(lr){3-8}
            \rowcolor{lightergray}
             0.3 & 18 & 0.10 & 0.31 & 0.64 & 0.87 & 0.96 & 0.99 \\
            \midrule
            \multicolumn{2}{c|}{} & $52$ & $53$ & $54$ & $55$ & $56$ & $57$ \\
            \cmidrule(lr){3-8}
            \rowcolor{lightergray}
             0.3 & 40 & 0.10 & 0.32 & 0.64 & 0.86 & 0.95 & 0.96\\
            \bottomrule
        \end{tabular}
    \end{minipage}
\end{table}

For both $\rho_0$, as $\nu_1$ increases from 18 to 40, $\nu_0$ need to increase correspondingly for $\iw$ prior to be comparable. This is because the average sample correlation of the live rat data is around $0.1$ as shown in Figure \ref{fig: 20 regions alive}, thus both $\rho_1 = 0.3$ and $\rho_0 = 0.4$ are far away from the data with increased $\nu_r$ bringing more gap. This is a consistent pattern with the results on the dead rat. Additionally, we observe that the range of $\nu_0$ within which the two priors are comparable is of different lengths across hyperparameter values.

\subsection{Edge detection}\label{app: detection mix}
To have a reliable detection result, we need a large amount of posterior samples so that we can have enough batches $L$ with each having enough size hence large $S$. If we use Algorithm 3 with a sufficiently large $N$, the machine will fail due to the memory limit. Thus in practice, we adapt the sampling algorithm to the edge detection procedure as in Algorithm \ref{too tired to give a name}. 
\begin{algorithm}
\caption{Edge detection}
\label{too tired to give a name}
\begin{algorithmic}[1]
\Require degree of freedom $\nu_1$, scale matrix $\Psi_1$, data $\x_{1:n}$, sample size for proposal $M$, clipping size $M_n$, number of batches $L$, batch size $S$, estimate $\Bar{R}$ for overall posterior mean $\mathbb{E}(\bR_{kk',n})$ obtained from Section \ref{sec-app: cal post mean}
\Ensure Binary graph $\mathcal{G}$.

\vspace{0.1in}
\textsc{/* Sampling from proposal */}

\State Run Algorithm \ref{alg} until Line \ref{line} to obtain importance samples and normalized weights $\Tilde{\Sigma}^{(m)}, p_m, m = 1, ..., M$.

\vspace{0.1in}
\textsc{/* Importance Resampling */}
\For{$l = 1, ..., L$}
\For{$s = 1, \ldots, S$}
\State $ m^{\star} \sim \mathcal{M}_M(1 |  P_0, \ldots,  p_M),$ 
\State $\Sigma^{(s)} := \Tilde{\Sigma}^{(m^\star)}.$ Convert $\Sigma^{(s)}$ to correlation matrix $R^{(s)}$.
\EndFor
\State Calculate the batch correlation mean $\Bar{R}_{S}$. 
\State Calculate the transformed sample mean estimator $M^{(l)} = \sqrt{S}(\Bar{R}_{S} - \Bar{R}) + \Bar{R}.$
\EndFor

\vspace{0.1in}
\textsc{/* Edge detection */}
\State For all $k, k = 1, ..., K$, perform Step \ref{edge detection 3} in the detection procedure with $M^{(l)}_{k,k'}, l = 1, ..., L$, and return a binary graph $\mathcal{G}$. 

\Return $\mathcal{G}$
\end{algorithmic}
\end{algorithm}
For the real data set in Section \ref{sec: app}, we set $M=10000, M_n = M^{0.6}, L = 1000, S= 500$. 

\end{document}